\documentclass[pdflatex,sn-basic]{sn-jnl}

\usepackage{amsmath,amsthm}
\usepackage[none]{hyphenat}  \usepackage{mathtools}  \usepackage{thmtools,thm-restate}
\usepackage{enumerate}
\usepackage[capitalize,noabbrev]{cleveref}
\normalbaroutside  \usepackage[textsize=tiny,disable]{todonotes}

\newcommand{\+}[1]{\mathbb{#1}}

\newcommand{\N}{\+{N}}
\renewcommand{\Z}{\+{Z}}
\renewcommand{\R}{\+{R}}
\renewcommand{\x}{\times}
\newcommand{\eps}{\varepsilon}
\newcommand{\nat}{\mathbb N}

\newcommand{\eqby}[2][=]{\stackrel{{\scalebox{0.5}{{#2}}}}{#1}}
\newcommand{\eqdef}{\eqby{def}}

\usepackage{wasysym} \newcommand{\rsymbol}{\ocircle}
\newcommand{\zsymbol}{\Box}
\newcommand{\osymbol}{\Diamond}
\newcommand{\xsymbol}{\odot}
\newcommand{\zstates}{\states_\zsymbol}
\newcommand{\rstates}{\states_\rsymbol}
\newcommand{\ostates}{\states_\osymbol}
\newcommand{\xstates}{\states_\xsymbol}
\newcommand{\reachset}{T}
\newcommand{\win}{\mathrm{win}}
\newcommand{\lose}{\mathrm{lose}}
\newcommand{\abs}[1]{\lvert#1\rvert}
\newcommand{\card}[1]{\abs{#1}}

\newcommand{\dist}{\mathcal{D}}
\newcommand{\support}{\mathtt{supp}}
\newcommand{\reach}[1]{\mathtt{Reach}(#1)}
\newcommand{\reachn}[2]{\mathtt{Reach}_{#1}(#2)}
\newcommand{\safety}[1]{\mathtt{Safety}(#1)}
\newcommand{\lrc}[1]{(#1)}

\newcommand{\ignore}[1]{}
\newcommand{\ord}{\mathbb O}
\newcommand{\tuple}[1]{\lrc{#1}}
\newcommand{\game}{{\mathcal G}}
\newcommand{\mdp}{{\mathcal M}}
\newcommand{\gametuple}{\tuple{\states,(\zstates,\ostates,\rstates),\transition,\probp}}

\newcommand{\states}{S}

\renewcommand{\state}{s}

\newcommand{\transition}{{\longrightarrow}}
\newcommand{\probp}{P}
\newcommand{\playset}{{\mathfrak R}}
\newcommand{\zstrat}{\sigma}
\newcommand{\ostrat}{\pi}

\newcommand{\zstratset}{\Sigma}
\newcommand{\ostratset}{\Pi}
\newcommand{\px}{\xsymbol}
\newcommand{\pz}{\zsymbol}
\newcommand{\po}{\osymbol}

\newcommand{\memconf}{{\sf m}}
\newcommand{\memconfset}{{\sf M}}

\newcommand{\expectval}{{\mathcal E}}
\newcommand{\probm}{{\mathcal P}}
\newcommand{\expectation}[1][]{ \expectval_{#1}}
\newcommand{\valueof}[2]{{\mathtt{val}_{#1}(#2)}}

\newcommand{\NthState}[1]{X_{#1}}
\newcommand{\poststar}[1]{\mathit{Post}^*(#1)}
\newcommand{\ferr}[1]{L(#1)}
\newcommand{\RVInT}[1]{[\NthState{#1} \in \reachset]}
\newcommand{\before}{\mathsf{beforeagain}}
\newcommand{\czerobeforeu}{{{c_0}\ \before\ u}}
\newcommand{\statesopt}{\states_\mathit{opt}}

\newcommand{\weighted}[1]{{\mathcal W}_{#1}}

\jyear{2022}

\theoremstyle{thmstyleone}\newtheorem{theorem}{Theorem}\newtheorem{proposition}[theorem]{Proposition}\newtheorem{lemma}[theorem]{Lemma}
\newtheorem{claim}[theorem]{Claim}

\newtheorem{remark}{Remark}

\newtheorem{definition}{Definition}

\newtheoremstyle{mythm}
{0.5\baselineskip}
{0.5\baselineskip}
{\itshape}
{}
{\bfseries}
{:}
{.5em}
{}
\theoremstyle{mythm}\newtheorem*{flbr*}{First Lower-Bound result (Q1)}
\newtheorem*{mupr*}{Main Upper-Bound result (Q2)}
\newtheorem*{mlpr*}{Main Lower-Bound result (Q3)}

\newtheoremstyle{myhope}
{0.5\baselineskip}
{0.5\baselineskip}
{\itshape}
{}
{\bfseries}
{:}
{.5em}
{}
\theoremstyle{myhope}\newtheorem{hope}{Hope}

\begin{document}

\title{Strategy Complexity of Reachability 
in Countable Stochastic 2-Player Games}

\author[1]{\fnm{Stefan}~\sur{Kiefer}}\email{stekie@cs.ox.ac.uk}
\author[2]{\fnm{Richard}~\sur{Mayr}}\email{rmayr@inf.ed.ac.uk}
\author[3]{\fnm{Mahsa}~\sur{Shirmohammadi}}\email{mahsa@irif.fr}
\author[4]{\fnm{Patrick}~\sur{Totzke}}\email{totzke@liverpool.ac.uk}

\affil*[1]{
\orgname{University of Oxford},
\country{UK}}
\affil*[2]{
\orgname{University of Edinburgh},
\country{UK}}
\affil*[3]{
\orgdiv{IRIF \& CNRS}, \orgname{Universit\'e Paris cit\'e},
\country{France}}
\affil*[4]{
\orgname{University of Liverpool},
\country{UK}}

\abstract{
We study countably infinite stochastic 2-player games with reachability objectives.
Our results provide a complete picture of the memory
requirements of $\eps$-optimal (resp.\ optimal) strategies. These results depend on the size of the players' action sets and
on whether
one requires strategies that are uniform (i.e., independent of the start state).

Our main result is that $\eps$-optimal (resp.\ optimal) Maximizer strategies
require infinite memory if Minimizer is allowed infinite action sets.
This lower bound holds even under very strong restrictions.
Even in the special case of infinitely branching turn-based reachability
games, even if all states allow an almost surely winning Maximizer strategy,
strategies with a step counter plus finite private memory are still useless.

Regarding \emph{uniformity}, we show that
for Maximizer there need not exist memoryless (i.e., positional) 
uniformly $\eps$-optimal strategies even in the special
case of finite action sets or in finitely branching turn-based games.
On the other hand, in games with finite action sets,
there always exists a uniformly $\eps$-optimal Maximizer strategy
that uses just one bit of public memory.

 }



\maketitle

\newcommand{\new}[1]{\textcolor{red}{#1}}
\renewcommand{\new}[1]{{#1}}
\newcommand{\verynew}[1]{\textcolor{red}{#1}}
\renewcommand{\verynew}[1]{{#1}}

\section{Introduction}\label{sec:intro}

\new{We study 2-player zero-sum stochastic games on countably
  \footnote{\verynew{Our proofs of upper bounds do not carry over to
      uncountable state spaces. E.g., we partition events into as many (by cardinality)
      parts as there are states and then rely on sigma-additivity of measures.
      Our lower bounds trivially carry over to uncountable state spaces.}} 
  infinite graphs. This section outlines the background and our contribution.
Formal definitions of games, strategies, memory, etc., are given in \Cref{sec:prelim}.
}

Stochastic games were first introduced by Shapley in his seminal 1953 work~\cite{shapley1953}, 
and model dynamic interactions in which the environment
responds randomly to players' actions.
Shapley's games were generalized by \cite{Gillette1958} and \cite{KumarShiau} 
to allow infinite state and action sets and non-termination.
They play a central role in the solution of many problems
in economics, see \cite{sorin1992,nowak2005,jaskiewiczN11,solan2015stochastic,bacharach2019}, evolutionary biology, e.g., \cite{raghaven2012}, 
and computer science, see \cite{AlfaroH01,neyman2003,AltmanAMM05, Altman2007,solan2015stochastic,svorevnova2016,bouyer2016} among others.

In general concurrent games, in each state \new{both Maximizer and Minimizer} independently choose
an action
and the next state is determined
according to a pre-defined distribution that depends on the chosen pair of
actions.
Turn-based games (also called switching-control games) are a subclass where
each state is owned by some player and only this player gets to choose an action.
These games were studied first in the 1980s and 90s in \cite{filar1980,filar1981ordered,Vrieze1983,vrieze1987,CONDON1992203}
but have recently received much attention by computer
scientists, for instance in \cite{GimbertH10,chen2013prism,bouyer2016,KieferMSW17a,BertrandGG17}.
An even more special case of stochastic games are \emph{Markov Decision
Processes (MDPs)}:
MDPs are turn-based games where all controlled states are Maximizer states.
Since Minimizer is passive, they are also called games against nature.

In order to get the strongest results,
we will show that our lower bound results hold even for
the special subclass of turn-based games
while our upper bounds hold even for general games.

A strategy for a player is a function that, given a history of a play,
determines the next action of the player.
Objectives are defined via functions that assign numerical rewards to plays,
and the Maximizer (resp.\ Minimizer) aim to maximize (resp.\ minimize)
the expected reward.
A central result in zero-sum 2-player stochastic games with finite action sets 
is the existence of a \emph{value} for the large class of Borel measurable
objectives \cite[]{martin_1998,Maitra-Sudderth:1998} (i.e., that
$\sup_{\it Max}\inf_{\it Min} = {\it value} = \inf_{\it Min}\sup_{\it Max}$
over Maximizer/Minimizer strategies).
In particular, this implies the existence of $\eps$-optimal strategies for
every $\eps >0$ and either player, i.e., strategies that enforce that the outcome of a game is
$\eps$-close to its value, regardless of the behavior of the other player.
Optimal strategies ($\eps$-optimal for $\eps=0$)
need not exist in general, but their properties have been
studied in those cases where they do exist, for example in \cite{Puterman:book,kucera_2011,KieferMSW17a,KMST2020c}.

The nature of good strategies in stochastic games -- that is $\eps$-optimality vs.~optimality,
and their memory requirements --
is relevant in computer science \cite[]{BKKNK14,KieferMSW17a,KMST2020c}, in particular, in the sense of computability~\cite[]{kucera_2011}.
It is also recognized as \new{a} central notion 
in branches of mathematics and economics, especially
operations research
\cite[]{MaitraS07},
probability theory
\cite[]{flesch2018}, 
game theory
\cite[]{FleschHMP21, LarakiMS13, MaitraS07}
and economic theory \cite[]{aumann1981survey,bacharach2019,KALAI1990131}. 

The simplest type of strategy bases its decisions only on the current state,
and not on the history of the play. Such strategies are called
\emph{memoryless} or \emph{positional}.
\footnote{A closely related concept is a \emph{stationary} strategy, which
  also bases decisions only on the current state. However, 
some authors
call a strategy
``stationary $\eps$-optimal''
if it is $\eps$-optimal from every state,
and call it ``semi-stationary''
if it is $\eps$-optimal only from the fixed initial state.
Since this difference is important in our work, we avoid the term
``stationary'' here.
Instead, if a strategy is $\eps$-optimal from every state then we call
it \emph{uniformly $\eps$-optimal}.
I.e., $\eps$-optimal stationary strategies are uniformly $\eps$-optimal
\new{memoryless} strategies.
}
\new{
By default, we assume that strategies can use randomization (i.e., use mixed
actions), while the subclass of deterministic (pure) strategies are
limited to choosing a single pure action at each state.
\emph{Memoryless randomized (MR)} strategies choose a mixed action at each
state, while \emph{memoryless deterministic (MD)} strategies choose a pure
action at each state, both independently of the history.
}

More complex strategies might use some finite amount of 
memory.
The strategy chooses an action depending only on the current state and the current memory mode.
The memory mode can be updated in every round according to the current state,
the observed chosen actions and the next state.
We assume perfect\new{-}information games, so the actions and states are observable
at the end of every round.
In general, for strategies that are not deterministic but use randomization, this memory update may also be randomized.
Therefore, in the case of games, a player does not necessarily know for sure
the current memory mode of the other player.
It may be advantageous for a player to keep his memory mode hidden from the
other player. We distinguish between \emph{public memory}, where the strategies' memory mode is public knowledge, and \emph{private memory}, which is hidden from the opponent.
A step counter is \new{an infinite memory device corresponding to} a discrete clock that is incremented after every round.
We consider this to be a type of public memory, because the update is deterministic and the memory mode can be
deduced by the opponent.
Strategies that use only a step counter are called \emph{Markov strategies}.
Combinations of the above are possible, e.g., a strategy that uses a step
counter and an additional finite public/private general purpose memory.
The amount/type of memory and randomization required for a good
($\eps$-optimal, resp.\ optimal) strategy for a
given objective is also called its \emph{strategy complexity}.

\subsection*{The Reachability Objective}
With a reachability objective, a play is defined as winning for Maximizer iff it
visits a defined target state (or a set of target states) at least once.
Thus Maximizer aims to maximize the probability that the target is reached.
Dually, Minimizer aims to minimize the probability of reaching the target.
So, from Minimizer's point of view, this is the dual \emph{safety objective} of
avoiding the target.

Reachability is arguably the simplest objective in games on
graphs.
It can trivially be encoded into the usual reward-based objectives, i.e.,
every play that reaches the target gets reward~$1$ and all other plays get
reward~$0$.
\new{
Moreover, it can be encoded into many other objectives 
including B\"uchi, Parity and average-payoff conditions,
by turning the target vertex into a good (for the new objective) sink.}

Despite their apparent simplicity, reachability games are not trivial.
While both players have optimal MD strategies in finite-state turn-based
reachability games \cite[]{CONDON1992203}; see also \cite[Proposition 5.6.c, Proposition 5.7.c]{kucera_2011},
this does not carry over to
finite-state concurrent reachability games.
\new{
A counterexample 
where Maximizer has no optimal strategy
is the \emph{Hide-or-Run} game \cite[Example 1]{Everett1957},
also see \cite{KumarShiau,AlfaroHK98}.
}

In countably infinite reachability games, Maximizer does not have an optimal
strategy even if the game is turn-based, in fact not even in countably
infinite MDPs that are finitely branching \cite[]{KMSW2017}.
On the other hand, \cite[Proposition A]{Ornstein:AMS1969} shows that Maximizer has
$\eps$-optimal MD strategies in countably infinite MDPs.
Better yet, the MD strategies can be made uniform, i.e., independent of the
start state.\footnote{\verynew{This memoryless uniformity does not carry over
to MDPs with uncountable state spaces by \cite[Theorem A]{Ornstein:AMS1969}.}}
This led to the question whether Ornstein's results can be generalized from
MDPs to countably infinite stochastic games.
\cite{Secchi97}, Corollary~3.9, proved the following.
\begin{proposition}\label{prop:story-state-of-the-art}
Maximizer has $\eps$-optimal memoryless (MR)
strategies in countably infinite concurrent reachability games with finite
action sets.
\end{proposition}

However, these MR strategies are not uniform, i.e., they depend
on the start state.
In fact, \cite{Raghavan-Nowak:1991} showed that there cannot exist
any uniformly $\eps$-optimal memoryless Maximizer strategies
in countably infinite concurrent reachability games with finite
action sets. Their counterexample is called the \emph{Big Match on $\N$}
which, in turn, is inspired by the \emph{Big Match}
\cite[]{Gillette1958,solan2015stochastic,BlackwellFerguson,Hansen2018}. 
Several fundamental questions remained open:
\begin{description}
\item[Q1.]
Does the negative result of \cite{Raghavan-Nowak:1991} still hold
in the special case of countable \emph{turn-based} (finitely branching) reachability games?
\item[Q2.]
    If \emph{uniformly} $\eps$-optimal Maximizer strategies cannot be
    memoryless, \verynew{how much memory do they need?}
\item[Q3.]
Does the positive result of Secchi (\Cref{prop:story-state-of-the-art} above) still hold if the
restriction to finite action sets is \verynew{relaxed}?
The question is meaningful, since 
concurrent games where only one player has \new{countably} infinite action sets are still
determined \cite[Theorem 11]{Flesch-Predtetchinski-Sudderth:2020}
(though not if both players have infinite action sets, unless one imposes
other restrictions).
Moreover, what about infinitely branching turn-based reachability games?
How much memory do good Maximizer strategies need in these cases?
\end{description}

\subsection*{Our Contribution}
Our  results, summarized in~\Cref{maxtable,mintable}, provide a comprehensive
view on the strategy complexity of (uniformly) $\eps$-optimal strategies for
reachability (and also about optimal strategies when they exist).

\begin{table*}[t]
\noindent\resizebox{\textwidth}{!}{
\begin{tabular}{|l|c|c|c|c|c|}
    \hline
    \multirow{2}{1.45cm} {Maximizer} & countable & turn-based games & turn-based games & concurrent games \\
                                & MDPs      & finite branching & infinite branching & finite action sets \\
    \hline\hline
    \multirow{3}{1.25cm}{{$\eps$-optimal}} & MD & MD & $\infty$-memory & MR\\
    &   \cite[Thm.~B]{Ornstein:AMS1969} & \cite[Proposition 5.7.c]{kucera_2011}, & [\Cref{thm:no-sc-plus-finite}] & \cite[Cor.~3.9]{Secchi97}\\
    && [\Cref{lem:conc-reach-non-uniform}] &&\\
    \hline
    \multirow{3}{1.25cm}{{Uniform $\eps$-optimal}} & MD & no MR [\Cref{thm:TB-BMI-zplus}]; & $\infty$-memory & no MR \cite[]{Raghavan-Nowak:1991};\\
     & \cite[Thm.~B]{Ornstein:AMS1969}  & det.\ public 1-bit & [\Cref{thm:no-sc-plus-finite}] & rand.\ public 1-bit,\\
     && [\Cref{thm:conc-reach-uniform}] && [\Cref{thm:conc-reach-uniform}]\\
    \hline
    \multirow{4}{1.25cm}{{Optimal}} & MD & no FR \cite[Prop.~5.7.b]{kucera_2011}; & $\infty$-memory & $\infty$-memory\\
    & \cite[Prop.~B]{Ornstein:AMS1969} & No Markov [\Cref{prop:turn-fb-optmax-lower}] & [\Cref{thm:no-sc-plus-finite}] & [\Cref{prop:conc-optmax}]\\
    && step counter + det.~public 1-bit, && \\
    && [\Cref{thm:turn-fb-optmax-upper}] &&
    \\
    \hline
    \multirow{2}{1.25cm}{{Almost sure}} & MD & MD & $\infty$-memory & MR\\
    & \cite[Prop.~B]{Ornstein:AMS1969} & \cite[Theorem~5.3]{KieferMSW17a} &
    [\Cref{thm:no-sc-plus-finite}] & [\Cref{thm:conc-reach-as}]\\
    \hline
\end{tabular}
}
\smallskip
\caption{The strategy complexity of Maximizer for the reachability
    objective. Since optimal and Almost sure (a.s.) winning strategies 
    are not guaranteed to exist,
    the results in the two bottom rows are conditioned upon their
    existence.
    ``$\infty$-memory'' means that even randomized strategies with a
    step counter plus an arbitrarily large finite private
    memory do not suffice. Deterministic strategies are useless in
    concurrent games, regardless of memory.}\label{maxtable}
\end{table*}
 
\begin{table*}[t]
\noindent\resizebox{\textwidth}{!}{
\begin{tabular}{|l|c|c|c|c|c|}
    \hline
    \multirow{2}{1.7cm} {Minimizer} & turn-based games & turn-based games & concurrent games \\
                                &  finite branching & infinite branching & finite action sets \\
    \hline\hline
    \multirow{2}{1.7cm}{{(Uniform) $\eps$-optimal}} & MD & no FR \cite[Thm.~3]{KMSW2017}; & \\
    & \cite[Thm.~3.1]{BBKO:IC2011} & det.\ Markov [\Cref{thm:min-eps-optimal}] & MR, \cite[Thm.~1]{Raghavan-Nowak:1991}\\
    \hline
    Optimal & MD \cite[Thm.~3.1]{BBKO:IC2011} &  $\infty$-memory [\Cref{prop:infbranch-optmin}] & MR \cite[Thm.~1]{Raghavan-Nowak:1991}\\ \hline
\end{tabular}
}
\smallskip
  \caption{The strategy complexity of Minimizer for the reachability
    objective. 
    Since optimal Minimizer strategies do not need to exist for infinitely
    branching games (unlike in the other cases), the result of
    \Cref{prop:infbranch-optmin} is conditioned upon their
    existence.
    Deterministic strategies are useless in concurrent games, regardless of memory.
  }\label{mintable}
\end{table*}
 
Our first result strengthens the negative result of \cite{Raghavan-Nowak:1991} to the turn-based case.
\begin{flbr*}(\Cref{thm:TB-BMI-zplus})
There exists a finitely branching turn-based version of the Big Match on $\N$ where
Maximizer still does not have any uniformly $\eps$-optimal MR strategy.
\end{flbr*}
Our second result solves the open question about uniformly $\eps$-optimal
Maximizer strategies. While uniformly $\eps$-optimal Maximizer strategies
cannot be memoryless, $1$ bit of memory is enough.
\begin{mupr*}(\Cref{thm:conc-reach-uniform}) 
In concurrent games with finite action sets and reachability objective, for any $\eps>0$,
Maximizer has a uniformly $\eps$-optimal public-memory 1-bit strategy.
This strategy can be chosen as deterministic if the game is turn-based and
finitely branching. 
\end{mupr*}

Our main contribution (\Cref{thm:story-main}) addresses Q3.
It determines the strategy complexity
of Maximizer in infinitely branching reachability games.
\new{Our result is a strong lower bound, and we present the path towards it 
by disproving a sequence of hopeful conjectures towards upper bounds.}
\begin{hope} \label{hope:general}
In turn-based reachability games, Maximizer has $\eps$-optimal MD strategies.
\end{hope}
This is motivated by the fact that the property holds if the game is
finitely branching \cite[Proposition 5.7.c]{kucera_2011}
and \Cref{lem:conc-reach-non-uniform} or if it is just an MDP as in
\cite{Ornstein:AMS1969}.

One might even have hoped for \emph{uniformly} $\eps$-optimal MD strategies,
i.e., strategies that do not depend on the start state of the game,
but this hope was crushed by the answer to Q1.

Let us mention a concern about Hope~\ref{hope:general} as stated (i.e., disregarding uniformity).
Consider any turn-based reachability game that is finitely branching, and let $x \in [0,1]$ be the value of the game.
The proof of \cref{prop:story-state-of-the-art} actually shows that for every
$\eps > 0$, Maximizer has both a strategy
and a time horizon $n \in \N$ such that for all Minimizer strategies,
the game visits the target state with probability at least $x-\eps$ \emph{within the first $n$ steps} of the game.
There is no hope that such a guarantee on the time horizon can be given in
infinitely branching games.
Indeed, consider the infinitely many states $f_0, f_1, f_2, \ldots$, where $f_0$ is the target state and for $i>0$ state $f_i$ leads to $f_{i-1}$ regardless of the players' actions, and an additional Minimizer state, $u$, where Minimizer chooses, by her action, one of the $f_i$ as successor state.
In this game, starting from~$u$, Maximizer wins with probability~$1$ (he is passive in this game).
Minimizer cannot avoid losing, but her strategy determines when $f_0$ is visited.
This shows that a proof of Hope~\ref{hope:general} would require different methods.

In case Hope~\ref{hope:general} turns out to be false, there are various plausible weaker versions.
Let us briefly discuss their motivation.

\begin{hope} \label{hope:MR}
Hope~\ref{hope:general} is true if MD is replaced by MR.
\end{hope}
This is motivated by \Cref{prop:story-state-of-the-art},
i.e., that in concurrent games with finite action sets for both players, Maximizer has $\eps$-optimal MR strategies.
In fact, \cite[Theorem 12.3]{Flesch-Predtetchinski-Sudderth:2020} implies that
this holds even under the weaker assumption that just Minimizer has finite
action sets (while Maximizer is allowed infinite action sets).

\begin{hope} \label{hope:optimal}
Hope~\ref{hope:general} is true if Maximizer has an optimal strategy.
\end{hope}
This is motivated by the fact that in MDPs with \emph{B\"uchi} objective (i.e., the player tries to visit a set of target states infinitely often), if the player has an optimal strategy, he also has an MD optimal strategy.
The same is not true for $\eps$-optimal strategies
as shown in \cite{KMSW2017}.
This example shows that although optimal strategies do not always exist, if they do exist, they may be simpler.

\begin{hope} \label{hope:almost-sure}
Hope~\ref{hope:general} is true if Maximizer has an almost surely winning strategy, i.e., a strategy that guarantees him to visit the target state with probability~$1$.
\end{hope}
This is weaker than Hope~\ref{hope:optimal}, because an almost surely winning strategy is necessarily optimal.

In a turn-based game, let us associate to each state~$\state$ its \emph{value}, which is the value of the game when started in~$s$.
We call a controlled step $\state \to \state'$ \emph{value-decreasing}
(resp., \emph{value-increasing}),
if the value of $\state'$ is smaller (resp., larger) than the value of $\state$.
It is easy to see that Maximizer cannot do value-increasing steps and 
Minimizer cannot do value-decreasing steps, but the opposite is possible in general. 

\begin{hope} \label{hope:no-value-dec}
Hope~\ref{hope:general} is true if Maximizer does not have value-decreasing steps.
\end{hope}

\begin{hope} \label{hope:no-value-inc}
Hope~\ref{hope:general} is true if Minimizer does not have value-increasing steps.
\end{hope}

Hopes~\ref{hope:no-value-dec} and~\ref{hope:no-value-inc} are motivated
by the fact that sometimes the absence of Maximizer value-decreasing steps
or the absence of Minimizer value-increasing steps implies the existence of
optimal Maximizer strategies and then Hope~\ref{hope:optimal} might apply.
For example, in finitely branching turn-based reachability games,
the absence of Maximizer value-decreasing steps or the absence of Minimizer
value-increasing steps implies the existence of optimal Maximizer strategies, and they can be chosen MD \cite[Theorem~5]{KieferMSW17a}.

\begin{hope} \label{hope:acyclic}
Hope~\ref{hope:general} is true for games with acyclic game graph.
\end{hope}
This is motivated, e.g., by the fact that in \emph{safety} MDPs (where the only active player tries to \emph{avoid} a particular state~$f$) with acyclic game graph and infinite action sets the player has $\eps$-optimal MD strategies \cite[Corollary~26]{KMST2020c}.
The same does not hold without the acyclicity assumption \cite[Theorem~3]{KMSW2017}.

\begin{hope} \label{hope:Markov}
Hope~\ref{hope:general} is true if Maximizer can additionally use a step counter to choose his actions.
\end{hope}
This is weaker than Hope~\ref{hope:acyclic}, because by using a step counter
Maximizer effectively makes the game graph acyclic.
However, the reverse does not hold. Not every acyclic game graph has an implicit step counter.

\begin{hope} \label{hope:finite-mem}
In turn-based reachability games, Maximizer has $\eps$-optimal strategies that use only finite memory.
\end{hope}
This is motivated, e.g., by the fact that in MDPs
with acyclic game graph and B\"uchi objective, the player has $\eps$-optimal
deterministic strategies that require only 1 bit of memory, but no
$\eps$-optimal MR strategies \cite[]{KMST:ICALP2019}.

It might be advantageous for Maximizer to keep his memory mode private.
This motivates the following final weakening of Hope~\ref{hope:finite-mem}.

\begin{hope} \label{hope:finite-private-mem}
In turn-based reachability games, Maximizer has $\eps$-optimal strategies that use only private finite memory.
\end{hope}

\medskip
The main contribution of this paper is to crush all these hopes.
That is, Hope~\ref{hope:general} is false, even if all weakenings proposed in Hopes \ref{hope:MR}-\ref{hope:finite-private-mem} are imposed \emph{at the same time}.
Specifically, we show the following theorem \verynew{(stated in more detail as \cref{thm:no-sc-plus-finite} later on)}.
\begin{theorem}\label{thm:story-main}
There is a turn-based reachability game (necessarily,
by \cref{prop:story-state-of-the-art},
with infinite action sets for Minimizer) with the following properties:
\begin{enumerate}
 \item for every Maximizer state, Maximizer has at most two actions to choose from;
 \item for every state Maximizer has a strategy to visit the target state with probability~$1$, regardless of Minimizer's strategy;
\item for every Maximizer strategy that uses only a step counter and  private finite memory and randomization, for every $\eps > 0$, Minimizer has a strategy so that the target state is visited with probability at most~$\eps$.
\end{enumerate}
\end{theorem}

This lower bound trivially carries over to concurrent stochastic games with infinite
Minimizer action sets, and for all Borel objectives
that subsume reachability, e.g., B\"uchi, co-B\"uchi, Parity, 
average-reward and total-reward.

To put this result into perspective, we show in \Cref{sec:restrictedmin}
that it is crucial that Minimizer can use infinite branching (resp.\ infinite
actions sets) \emph{infinitely often}.
If the game is restricted such that Minimizer can use infinite actions sets
only \emph{finitely often} in any play then Maximizer still has
uniformly $\eps$-optimal public 1-bit strategies.

While optimal Maximizer strategies need not exist in general,
it is still relevant to study the case where they do exist.
If Minimizer can use infinite branching (resp.\ infinite actions sets)
then \Cref{thm:story-main} shows that Maximizer needs infinite memory even in
that case.
\new{
However, optimal Maximizer strategies in finitely branching turn-based games
can be chosen to use a step counter plus 1~bit of public memory,
while just a step counter is not enough; cf.~\Cref{sec:optimalmax}.
}

Finally, in \Cref{sec:minimizer}, we \new{determine} the strategy complexity of
Minimizer for all cases.
In particular, Minimizer has uniformly $\eps$-optimal memoryless strategies in
turn-based games that are infinitely branching but acyclic.
 
\section{Preliminaries and Notations}\label{sec:prelim}
A \emph{probability distribution} over a countable set $\states$ is a function
$f:\states\to[0,1]$ with $\sum_{\state\in \states}f(\state)=1$.
Let $\support(f) \eqdef \{\state \mid f(\state) >0\}$ denote the support of $f$.
We write 
$\dist(\states)$ for the set of all probability distributions over $\states$. 

\new{
We study perfect-information 2-player stochastic games between
the two players \emph{Maximizer} (also denoted as $\pz$)
and \emph{Minimizer} (also denoted as $\po$).}

\subsection*{2-Player Concurrent Stochastic Games}
A 2-player concurrent game $\game$ is played on a countable set of states $\states$.
For each state $\state \in \states$ there are nonempty countable action sets
$A(\state)$ and $B(\state)$ for Maximizer and Minimizer, respectively.
Let $Z \eqdef \{(\state,a,b) \mid \state \in \states, a \in A(\state), b \in B(\state)\}$.
For every triple $(\state,a,b) \in Z$
there is a distribution $p(\state,a,b) \in \dist(\states)$ over successor states.
\new{We call a state~$s \in \states$ a \emph{sink} state, or \emph{absorbing}, if $p(s,a,b) = s$ for all $a \in A(s)$ and $b \in B(s)$.}
The set of \emph{plays} from an initial state $\state_0$ is given by the infinite
sequences in $Z^\omega$ where the first triple contains $\state_0$.
The game from $\state_0$ is played in stages $\N=\{0,1,2,\dots\}$.
At every stage $t \in \N$, the play is in some state $\state_t$.
Maximizer chooses an action $a_t \in A(\state_t)$ and Minimizer chooses
an action $b_t \in B(\state_t)$. 
The next state $\state_{t+1}$ is then chosen according to the distribution
$p(\state_t,a_t,b_t)$.
(Since we just consider the reachability objective here,
we don't define a reward function.)

\subsection*{2-Player Turn-based Stochastic Games}
\new{
A special subclass of concurrent stochastic games are turn-based games,
where in each round either Maximizer or Minimizer is passive (i.e., has just a single action to play).
Turn-based games are often represented in a form that explicitly separates
local decisions into Maximizer-controlled ones, Minimizer-controlled ones,
and random decisions. Thus one describes the turn-based game as 
$\game=\gametuple$ where the countable set of
states $\states$ is partitioned into
the set~$\zstates$ of states controlled by Maximizer ($\pz$),
the set~$\ostates$ of states controlled by Minimizer ($\po$) and
\emph{random states} $\rstates$.}
The relation $\mathord{\transition}\subseteq\states\times\states$ is the transition relation.
We write $\state\transition{}\state'$ if $\tuple{\state,\state'}\in \transition$, and we assume that each state~$\state$ has a \emph{successor} state $\state'$ with $\state \transition \state'$.
The probability function $\probp:\rstates \to \dist(\states)$ assigns to each random 
state~$\state \in \rstates$ a probability distribution over its 
successor states. 

The game~$\game$ is called \emph{finitely branching} if each state has only finitely many successors;
otherwise, it is \emph{infinitely branching}.
A game is \emph{acyclic} if the underlying  graph~$(S,\transition)$ is acyclic.
Let $\xsymbol \in \{\zsymbol,\osymbol\}$.
At each stage $t$, if the game is in state $\state_t \in \xstates$ 
then player $\xsymbol$ chooses a successor state~$\state_{t+1}$ with~$\state_t\transition{}\state_{t+1}$;
otherwise the game is in a random state~$\state_t\in \rstates$ and proceeds
randomly to~$\state_{t+1}$ with probability~$\probp(\state_t)(\state_{t+1})$.
If $\xstates=\emptyset$, we say that player~$\px$ is \emph{passive}, and the
game is a \emph{Markov decision process (MDP)}.
A \emph{Markov chain} is an MDP where both players are passive.

\subsection*{Strategies and Probability Measures}
\new{The set of \emph{histories} at stage $n$, with $n\in \mathbb{N}$, is denoted by $H_n$. That is,
$H_0 \eqdef \states$ and $H_n \eqdef Z^n \times \states$ for all $n>0$.
Let $H \eqdef \bigcup_{n \in \N} H_n$ be the set of all histories.
For each history $h = (s_0,a_0,b_0) \cdots (s_{n-1},a_{n-1},b_{n-1}) s_{n} \in H_n$, let $\state_h \eqdef s_{n}$ denote the final state in $h$.
}
In the special case of turn-based games, the history can be represented by
the sequence of states $\state_0\state_1\cdots \state_h$, where
$\state_i\transition{}\state_{i+1}$ for all~$i\in \mathbb{N}$.
\new{We say that a history $h \in H$ \emph{visits} the set of states $\reachset \subseteq \states$ at stage $t$ if $\state_t \in \reachset$.}

A mixed action for Maximizer (resp.\ Minimizer) in state $\state$
is a distribution over $A(\state)$ (resp.\ $B(\state)$).
A \emph{strategy} for Maximizer (resp.\ Minimizer)
is a function $\zstrat$ (resp.\ $\ostrat$) that to each history $h \in H$ assigns
a mixed action $\zstrat(h) \in \dist(A(\state_h))$
(resp.\ $\ostrat(h) \in \dist(B(\state_h))$ for Minimizer).
For turn-based games this means instead a distribution
$\zstrat(h) \in \dist(\states)$ over successor states if
$\state_h \in \zstates$ (and similarly for Minimizer with
$\state_h \in \ostates$).
Let $\zstratset$ (resp.\ $\ostratset$) denote the set of strategies for
Maximizer (resp.\ Minimizer).

An initial state $\state_0$ and a pair of strategies $\zstrat, \ostrat$
for Maximizer and Minimizer
induce a probability measure on sets of plays.
We write $\probm_{\game,\state_0,\zstrat,\ostrat}({\playset})$ for the probability of a 
measurable set of plays $\playset$ starting from~$\state_0$.
\new{
More generally, if $f: Z^\omega \to \R$ is a measurable reward function on
plays then we write $\expectval_{\game,\state_0,\zstrat,\ostrat}(f)$
for the expected reward w.r.t.~$f$ and $\probm_{\game,\state_0,\zstrat,\ostrat}$.
The case of measurable sets of plays $\playset$ is subsumed by this, since
we can choose $f$ as the indicator function of $\playset$.
}
These measures are initially defined for the cylinder sets and extended to the sigma
algebra by Carath\'eodory's unique extension
theorem~\cite[]{billingsley-1995-probability}.

\subsection*{Objectives}
We consider the reachability objective for Maximizer.
Given a set $\reachset \subseteq \states$ of states, 
the \emph{reachability} objective
$\reach{\reachset}$ is the set of plays that visit $\reachset$ at least once.
From Minimizer's point of view, this is the dual
\emph{safety} objective~$\safety{\reachset} \eqdef Z^\omega \setminus \reach{\reachset}$
of plays that never visit~$T$.
Maximizer (resp.\ Minimizer) attempts to maximize (resp.\ minimize)
the probability of $\reach{\reachset}$.

\new{
For any subset of states $R\subseteq \states$,
let $\reachn{R}{\reachset}$ denote the objective of visiting $\reachset$ while remaining in~$R$ before visiting~$\reachset$.
For $X \subseteq \N$, let $\reachn{X}{\reachset}$ denote the objective of visiting~$\reachset$
in some number of rounds $n \in X$.
For $n \in \N$ let
$\reachn{n}{\reachset} \eqdef \reachn{{\{k \mid k \le n\}}}{\reachset}$
denote the objective of reaching $\reachset$ in at most $n$ rounds.
}

\subsection*{Value and Optimality}
For a game $\game$, initial state $\state_0$ and objective $\playset$ the \emph{lower value} is
defined as
\[
\alpha(\state_0) \eqdef \sup_{\zstrat \in \zstratset} \inf_{\ostrat \in \ostratset}
\probm_{\game,\state_0,\zstrat,\ostrat}({\playset})
\]
Similarly, the \emph{upper value} is defined as
\[
\beta(\state_0) \eqdef \inf_{\ostrat \in \ostratset} \sup_{\zstrat \in \zstratset}
\probm_{\game,\state_0,\zstrat,\ostrat}({\playset})
\]
The inequality $\alpha(\state_0) \le \beta(\state_0)$ trivially holds.
If $\alpha(\state_0) = \beta(\state_0)$, then this quantity is called the
\emph{value} of the game, denoted by $\valueof{\game,\playset}{\state_0}$.
Reachability objectives, like all Borel objectives, have value \cite[]{Maitra-Sudderth:1998}.
For $\eps >0$, a strategy $\zstrat \in \zstratset$ from $\state_0$ for Maximizer is called
\emph{$\eps$-optimal} if
$\forall \ostrat \in \ostratset.\, \probm_{\game,\state_0,\zstrat,\ostrat}({\playset})
\ge \valueof{\game,\playset}{\state_0} - \eps$.
Similarly, a strategy $\ostrat \in \ostratset$ from $\state_0$ for Minimizer is called
$\eps$-optimal if
$\forall \zstrat \in \zstratset.\, \probm_{\game,\state_0,\zstrat,\ostrat}({\playset})
\le \valueof{\game,\playset}{\state_0} + \eps$.
If a strategy is $0$-optimal we simply call it optimal.

\subsection*{Memory-based Strategies}
A \emph{memory-based strategy} $\zstrat$ of Maximizer is a strategy that
can be described by a tuple $(\memconfset, \memconf_0, \zstrat_\alpha, \zstrat_\memconf)$
where $\memconfset$ is the set of memory modes, $\memconf_0 \in \memconfset$
is the initial memory mode, and
the functions $\zstrat_\alpha$ and $\zstrat_\memconf$ describe how actions are
chosen and memory modes updated; see below.
A play according to $\zstrat$ 
generates a random sequence of memory
states $\memconf_0, \dots, \memconf_t, \memconf_{t+1}, \dots$ from a given set
of memory modes $\memconfset$, where $\memconf_t$ is the memory mode at
stage $t$.
The strategy $\zstrat$ selects the action at stage $t$ according to a distribution that
depends only on the current state $\state_t$ and the memory $\memconf_t$.
Maximizer's action $a_t$ is chosen via a distribution
$\zstrat_\alpha(\state_t, \memconf_t) \in \dist(A(\state_t))$.
(Minimizer's action is $b_t$).
The next memory mode $\memconf_{t+1}$ of Maximizer is chosen according to a
distribution $\zstrat_m(\state_t, a_t, b_t, \state_{t+1}) \in \dist(\memconfset)$
that depends on the chosen actions and the observed outcome.
The memory is \emph{private} if the other player
cannot see the memory mode. Otherwise, it is \emph{public}.

Let $\zstrat[\memconf]$ denote the memory-based strategy $\zstrat$ that starts
in memory mode $\memconf$.
In cases where the time is relevant and the strategy has access to the time (by
using a step counter)
$\zstrat[\memconf](t)$ denotes the strategy $\zstrat$ in memory mode $\memconf$ at time~$t$.

A \emph{finite-memory strategy} is one where $\card{\memconfset} < \infty$.
A \emph{$k$-memory strategy} is a memory-based strategy with at most $k$
memory modes, i.e., $\card{\memconfset} \le k$.
A $2$-memory strategy is also called a \emph{1-bit strategy}.
A strategy is \emph{memoryless} (also called positional)
if $\card{\memconfset}=1$.
A strategy is called \emph{Markov} if it uses only a step counter but no additional memory.
A strategy is \emph{deterministic} (also called pure) if the distributions for the
action and memory update are Dirac. Otherwise, it is called \emph{randomized}
(or mixed).
\new{Memoryless randomized strategies are also called MR 
and memoryless deterministic strategies are also called MD.
Similarly, randomized (resp.\ deterministic) finite-memory strategies are also called FR (resp.\ FD).
}

\new{
A finite-memory strategy $\zstrat$ is called \emph{uniformly $\eps$-optimal}
for an objective $\playset$ iff
$\forall \state \in \states.\forall \ostrat.\, \probm_{\game,\state,\zstrat[\memconf_0],\ostrat}({\playset})
\ge
\valueof{\game,\playset}{\state} - \eps$, i.e.,
the strategy performs well from every state.
}

The definitions above carry over directly to the simpler turn-based games
where we have
chosen/observed transitions instead of actions.
 
\section{Uniform Strategies in Concurrent Games}\label{sec:uniconcurr}

\begin{figure}[t]
\begin{minipage}[l]{0.3\textwidth}
    \includegraphics[width=\textwidth, angle=0]{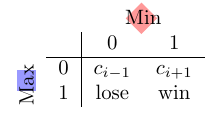}
\end{minipage}
\hfill
\begin{minipage}[r]{0.6\textwidth}
    \includegraphics[width=\textwidth, angle=0]{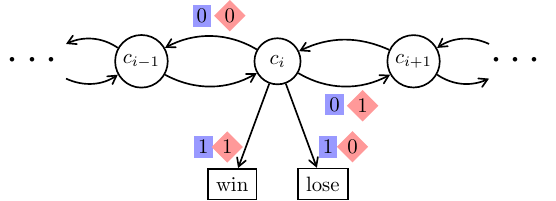}
\end{minipage}
\smallskip
   \caption{The Concurrent Big Match on~$\Z$; see~\Cref{def:concurrent-big-match}. On the right is a depiction of the game graph; on the left we see how joint actions from state $c_i$ are resolved.
   }
 \label{fig:bigmatchZconcurr}
\end{figure}
First, we consider the concurrent version of the Big Match on the integers,
in the formulation of \cite{FLS:1995}.
\begin{definition}[Concurrent Big Match on~$\mathbb{Z}$]\label{def:concurrent-big-match}
This game is shown in~\Cref{fig:bigmatchZconcurr}.
The state space is $\{c_i \mid i \in \Z\} \cup \{\win,\lose\}$, where the states $\win$
and $\lose$ are absorbing.
Both players have the action set $\{0,1\}$ at each state.
If Maximizer chooses action $1$ in $c_i$ then the game is decided in this round:
If Minimizer chooses $0$ (resp.\ $1$) then the game goes to $\lose$ (resp.\ $\win$).
If Maximizer chooses action $0$ in $c_i$ and Minimizer chooses action $0$ (resp.\ $1$)
then the game goes to $c_{i-1}$ (resp.\ $c_{i+1}$).
Maximizer wins iff state $\win$ is reached or
$\liminf\{i \mid \mbox{$c_i$ visited}\} = - \infty$.
\end{definition}
\begin{theorem}[{\cite{FLS:1995}, Theorem 1.1}]\label{thm:FLS1995}
In the concurrent Big Match on~$\mathbb{Z}$, shown in~\Cref{fig:bigmatchZconcurr}, every state $c_i$
has value $1/2$. An optimal strategy for Minimizer is to toss a fair
coin at every stage.
Maximizer has no optimal strategy, but for any start state $c_x$ and any positive integer
$N$, he can win with probability $\ge N/(2N+2)$ by choosing action $1$
with probability $1/(n+1)^2$ whenever the current state is
$c_i$ with $i=x+N-n$ for some~$n \ge 0$.
\end{theorem}
The  concurrent Big Match on~$\mathbb{Z}$ is not a reachability game, due to
its particular winning condition. However, the following slightly modified
version (played on $\N$) is a reachability game.

\begin{definition}[Concurrent Big Match on $\N$]\label{def:concurrent-big-match-zplus}
This game is shown in~\Cref{fig:bigmatchNconcurr}.
The state space is $\{c_i \mid i \in \N\} \cup \{\lose\}$ where 
$\lose$ and $c_0$ are absorbing.
Both players have the action set $\{0,1\}$ at each state.
If Maximizer chooses action $1$ in $c_i$ then the game is decided in this round:
If Minimizer chooses $0$ (resp.\ $1$) then the game goes to $\lose$ (resp.\ $c_0$).
If Maximizer chooses action $0$ in $c_i$ and Minimizer chooses action $0$ (resp.\ $1$)
then the game goes to $c_{i-1}$ (resp.\ $c_{i+1}$).

Maximizer wins iff $c_0$ is reached, i.e., we have the reachability objective
$\reach{\{c_0\}}$.
\end{definition}
The following theorem summarizes results on the concurrent Big Match on $\N$
by combining results from~\cite{FLS:1995} and~\cite{Raghavan-Nowak:1991}.

\begin{figure}[t]
\begin{minipage}[l]{0.3\textwidth}
    \includegraphics[width=\textwidth, angle=0]{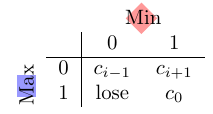}
  \end{minipage}
\hfill
\begin{minipage}[r]{0.65\textwidth}

    \includegraphics[width=\textwidth, angle=0]{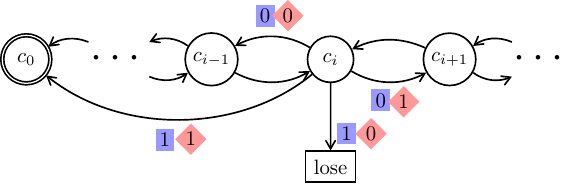}
\end{minipage}
\caption{The Concurrent Big Match on~$\N$; see~\Cref{def:concurrent-big-match-zplus}. On the right is a depiction of the game graph; on the left we see how joint actions from state $c_i$ are resolved.
   }
\label{fig:bigmatchNconcurr}
\end{figure}

\begin{theorem}\label{thm:BMI-zplus}
    Denote by~$\game$ the concurrent Big Match game  on $\N$, as shown in~\Cref{fig:bigmatchNconcurr}, and let \new{$x\in\N$}. Then,
\begin{enumerate}
\item\label{thm:BMI-zplus-1}
$\valueof{\game}{c_x} = (x+2)/(2x+2) \ge 1/2$.
\item\label{thm:BMI-zplus-2}
For every start state $c_x$ and $N\ge 0$,
Maximizer can win with probability $\ge N/(2N+2)$ by choosing action $1$
with probability $1/(n+1)^2$ whenever the current state is
$c_i$ with $i=x+N-n$ for some $n \ge 0$.
\item\label{thm:BMI-zplus-3}
For any $\eps < 1/2$
there is no uniformly $\eps$-optimal memoryless (MR) strategy for
Maximizer.
Every MR Maximizer strategy $\zstrat$ attains arbitrarily little from $c_x$ as
$x \to \infty$.
Formally,
$
\limsup_{x \to \infty}\inf_{\ostrat}\probm_{\game,c_x,\zstrat,\ostrat}(\reach{\{c_0\}})=0.
$
\end{enumerate}
\end{theorem}
\begin{proof}
Item~\ref{thm:BMI-zplus-1} follows directly from \cite[Proposition 5.1]{FLS:1995}.

Item~\ref{thm:BMI-zplus-2} follows from \cref{thm:FLS1995}, since it is easier
for Maximizer to win in the game of \cref{def:concurrent-big-match-zplus}
than in the game of \cref{def:concurrent-big-match}.

Towards item~\ref{thm:BMI-zplus-3}, we follow the proof of \cite[Lemma 4]{Raghavan-Nowak:1991}.
Let $\zstrat$ be an MR Maximizer strategy and $f(x)$ the probability that
$\zstrat$ picks action $1$ at state $c_x$.
There are two cases.

In the first case $\sum_{x \ge 1} f(x) < \infty$. Let $\ostrat$ be the strategy of
Minimizer that always picks action $1$. For all $x \ge 1$ we have
\begin{align*}
    \probm_{\game,c_x,\zstrat,\ostrat}(\reach{\{c_0\}})
    \quad
    \le \quad &\ f(x)  + (1-f(x))f(x+1) \\
        &+ (1-f(x))(1-f(x+1))f(x+2) \\
        &+ \dots\ \le \sum_{k=x}^\infty f(k) < \infty.
\end{align*}
Thus $\limsup_{x \to \infty}\probm_{\game,c_x,\zstrat,\ostrat}(\reach{\{c_0\}})=0$.

In the second case $\sum_{x \ge 1} f(x) = \infty$.
Let $\ostrat$ be the strategy of
Minimizer that always picks action $0$. For all $x \ge 1$ we have
\begin{align*}
\probm_{\game,c_x,\zstrat,\ostrat}(\reach{\{c_0\}})
&=(1-f(x))(1-f(x-1))\cdots (1-f(1)) \\
&=\prod_{k=1}^x (1-f(k))\\
&\le\frac{1}{1+ \sum_{k=1}^x f(k)}
\end{align*}
\new{For the final inequality we refer the reader to \cref{prop:Weierstrass} in Appendix \ref{app:technical}.}
Thus $\limsup_{x \to \infty}\probm_{\game,c_x,\zstrat,\ostrat}(\reach{\{c_0\}})=0$.

Since, by item~\ref{thm:BMI-zplus-1}, 
$\valueof{\game}{c_x} = (x+2)/(2x+2) \ge 1/2$
for every $x \ge 0$, the MR strategy $\zstrat$ cannot
be uniformly $\eps$-optimal for any $\eps < 1/2$.
\end{proof}

While uniformly $\eps$-optimal Maximizer strategies cannot be memoryless,
we show that they can be chosen with just 1 bit of public memory in
\cref{thm:conc-reach-uniform}.

First we need an auxiliary lemma that is essentially known; see, e.g.,
\cite{MaitraSudderth:DiscreteGambling} Section 7.7, and
\cite{Flesch-Predtetchinski-Sudderth:2020} Theorem 12.1.
We extend it slightly to fit our purposes, i.e., for the proof of
\cref{thm:conc-reach-uniform} below.

\begin{lemma} \label{lem:conc-reach-non-uniform}
  Consider a concurrent game with countable state space~$S$,
  \emph{finite} action sets for Minimizer at every state and
  unrestricted (possibly infinite) action sets for Maximizer.

  For the
  reachability objective $\reach{\reachset}$, for every finite set $S_0 \subseteq S$ of initial states, and for every $\eps>0$,
there exists a memoryless strategy~$\sigma$ and a finite set of states~$R \subseteq S$ such that
for all $s_0 \in S_0$
 \[
  \inf_{\ostrat \in \ostratset} \probm_{\state_0,\zstrat,\ostrat}(\reachn{R}{\reachset}) \ \ge \ \valueof{\reach{\reachset}}{\state_0} - \eps\,,
 \]
where $\reachn{R}{\reachset}$ denotes the objective of visiting $\reachset$ while remaining in~$R$ before visiting~$T$.
If the game is turn-based and finitely branching at Minimizer-controlled
states, there is a deterministic (i.e., MD) such strategy~$\sigma$.
\end{lemma}
\begin{proof}
Since Minimizer's action sets are finite,
by \cite[Theorem 11.1]{Flesch-Predtetchinski-Sudderth:2020}, the game has \new{a} value.
Moreover, using the finiteness of Minimizer's action sets again,
it follows from
\cite[Theorem 12.1]{Flesch-Predtetchinski-Sudderth:2020}
that for all $s \in S$
\begin{equation} \label{eq:MS-DiscreteGambling-7-7}
 \lim_{n \to \infty} \valueof{\reachn{n}{\reachset}}{s} \ = \ \valueof{\reach{\reachset}}{s}\,,
\end{equation}
where $\reachn{n}{\reachset}$ denotes the objective of visiting~$\reachset$
within at most $n$ rounds of the game.

To achieve the uniformity (across the set~$S_0$ of initial states) required by the statement of the lemma, we add a fresh ``random'' state (i.e., a state in which each player has only a single action available) that branches uniformly at random to a state in~$S_0$.
Call this state~$\hat s_0$.
The value of~$\hat s_0$ is the arithmetic average of the values of the states in~$S_0$.
It follows that every $(\eps/\card{S_0})$-optimal memoryless strategy for Maximizer in~$\hat s_0$ must be $\eps$-optimal in every state in~$S_0$.
So it suffices to prove the statement of the lemma under the assumption that $S_0$ is a singleton, say $S_0 = \{s_0\}$.

Fix $\eps>0$ and let $\eps' \eqdef \eps/4$.
By Equation~\eqref{eq:MS-DiscreteGambling-7-7} there is a number $n$ such that $\valueof{\reachn{n}{\reachset}}{s_0} = \valueof{\reach{\reachset}}{s_0} - \eps'$.
Let $\sigma$ be a Maximizer strategy such that
\begin{equation}\label{eq:conc-reach-non-uniform-1}
 \inf_{\ostrat \in \ostratset} \probm_{s_0,\sigma,\ostrat}(\reachn{n}{\reachset}) \ \ge \ \valueof{\reach{\reachset}}{s_0} - 2\eps'\,.
\end{equation}
For each $m$ with
$0 \le m \le n$ we will inductively construct a finite subset
$H'_m \subseteq H_m$ of the $m$-step histories of
plays from $s_0$ that are compatible with $\sigma$ 
such that, for every Minimizer strategy $\ostrat$,
the plays in $H_m'Z^\omega$ have probability $\ge 1 - \frac{m}{n} \eps'$,
where the event $H_m'Z^\omega$ is defined as the set of
continuations of the $m$-step histories in $H'_m$.
Formally,
\begin{equation}\label{eq:conc-reach-non-uniform-2}
 \inf_{\ostrat \in \ostratset} \probm_{s_0,\sigma,\ostrat}(H_m'Z^\omega) \ \ge \ 1 - \frac{m}{n} \eps'
\end{equation}

The base case of $m=0$ is trivial.
Now we show the inductive step from $m$ to $m+1$.
For any of the finitely many histories $h \in H'_m$ ending in some state $\state$,
consider the chosen mixed actions $a \in \dist(A(s))$ and $b \in \dist(B(s))$
by Maximizer and Minimizer, respectively.
Since $B(s)$ is finite, $b$ has finite support.
However, the distribution $a$ can have infinite support.
We fix a sufficiently large finite subset $A'$ of the support of $a$
that has probability mass $\ge 1-\frac{\eps'}{2n}$.
Consider the set $\gamma(s)$ of possible successor states  of $s$.
Since the size of the support of $b$ is upper bounded by the finite
number $\card{B(s)}$ independently of $\ostrat$, we can pick a finite subset
$\gamma'(s) \subseteq \gamma(s)$ sufficiently large such that both Maximizer's
chosen action is inside $A'$ and the chosen successor state is inside
$\gamma'(s)$ with probability $\ge 1 - \frac{1}{n} \eps'$.
We then define $H'_{m+1}$ as the finitely many one-round extensions of histories in $H'_m$
with Maximizer action in $A'$ and successor state in $\gamma'(s)$.
Using the induction hypothesis and the properties above, we obtain that
\new{
\begin{align*}
  \inf_{\ostrat \in \ostratset} \probm_{s_0,\sigma,\ostrat}(H_{m+1}'Z^\omega) 
  & \ge \ \inf_{\ostrat \in \ostratset} \probm_{s_0,\sigma,\ostrat}(H_m'Z^\omega) (1-\frac{1}{n}\eps')\\
  & \ge \ (1-\frac{m}{n}\eps') (1-\frac{1}{n}\eps') \\ 
  & = \ 1 - \frac{m+1}{n}\eps' + \frac{m}{n^2} \eps'^2\\ 
  & \ge \ 1 - \frac{m+1}{n} \eps'.
\end{align*}
}
This completes the induction step, and thus we obtain \eqref{eq:conc-reach-non-uniform-2}.

For every $0 \le m \le n$ let $R_m$ be the finite set of
states that are visited during the first $m$ steps of the histories in $H_m'$.
Then $R \eqdef R_n$ is a finite set of states.
It follows that
\new{
\begin{align*}
  & \inf_{\ostrat \in \ostratset} \probm_{s_0,\sigma,\ostrat}(\reachn{R}{\reachset})\\
  &\ge
  \inf_{\ostrat \in \ostratset} \probm_{s_0,\sigma,\ostrat}(H'_n Z^\omega \cap \reachn{R}{\reachset})\\
  &\ge
  \inf_{\ostrat \in \ostratset} \probm_{s_0,\sigma,\ostrat}(H'_n Z^\omega \cap
    \reachn{n}{\reachset}) & \mbox{set incl.}\\
  & =
  \inf_{\ostrat \in \ostratset}
    (\probm_{s_0,\sigma,\ostrat}(\reachn{n}{\reachset}) -
\probm_{s_0,\sigma,\ostrat}(\overline{H'_n Z^\omega} \cap \reachn{n}{\reachset}))\\
  & \ge
  \inf_{\ostrat \in \ostratset}
    (\probm_{s_0,\sigma,\ostrat}(\reachn{n}{\reachset}) -
\probm_{s_0,\sigma,\ostrat}(\overline{H'_n Z^\omega}))\\
  & \ge
  \inf_{\ostrat \in \ostratset}
    (\probm_{s_0,\sigma,\ostrat}(\reachn{n}{\reachset})) -
 \sup_{\ostrat \in \ostratset}(\probm_{s_0,\sigma,\ostrat}(\overline{H'_n Z^\omega}))\\
  & =
  \inf_{\ostrat \in \ostratset}
    (\probm_{s_0,\sigma,\ostrat}(\reachn{n}{\reachset})) -
 (1-\inf_{\ostrat \in \ostratset} \probm_{s_0,\sigma,\ostrat}(H'_n Z^\omega))\\
  &\ge
  \inf_{\ostrat \in \ostratset}
    \probm_{s_0,\sigma,\ostrat}(\reachn{n}{\reachset}) - \eps' & \mbox{by \eqref{eq:conc-reach-non-uniform-2}}\\
  &\ge 
  \valueof{\reach{\reachset}}{s_0} - 3\eps'. & \mbox{by \eqref{eq:conc-reach-non-uniform-1}}
\end{align*}
Note that the restriction on the time horizon, $n$, has been lifted here.
In particular, \verynew{the} above implies that
\begin{equation}\label{eq:conc-reach-non-uniform-3}
\valueof{\reachn{R}{\reachset}}{s_0} \ge \valueof{\reach{\reachset}}{s_0} - 3\eps'.
\end{equation}
}
The restriction of the objective to the (finitely many) states in~$R$ means that we have effectively another reachability game.
It is known \cite[Corollary~3.9]{Secchi97} that for concurrent games with finite action sets and reachability objective, Maximizer has a memoryless $\eps'$-optimal strategy.
In turn-based games with finitely many states, he even has an MD optimal strategy~\cite[]{CONDON1992203}.
So Maximizer has a memoryless (in the turn-based case: MD) strategy $\sigma'$ such that
\new{
\begin{align*}
& \inf_{\ostrat \in \ostratset} \probm_{s_0,\sigma',\ostrat}(\reachn{R}{\reachset})\\
& \ge \valueof{\reachn{R}{\reachset}}{s_0} - \eps' & \mbox{by $\eps'$-optimality of $\sigma'$} \\
& \ge \valueof{\reach{\reachset}}{s_0} - 4\eps'    & \mbox{by \eqref{eq:conc-reach-non-uniform-3}}\\ 
&  =\valueof{\reach{\reachset}}{s_0} - \eps. 
\end{align*}
}
\end{proof}

\begin{restatable}{theorem}{concreachuniform}\label{thm:conc-reach-uniform}
For any concurrent game with finite action sets and reachability objective, for any $\eps>0$,
Maximizer has a uniformly $\eps$-optimal public 1-bit strategy.
If the game is turn-based and finitely branching, Maximizer has a deterministic such strategy.
\end{restatable}

\begin{proof}
  Denote the game by~$\hat\game$, over state space~$S$.
  \new{Let $\eps >0$. We show how to construct the uniformly $\eps$-optimal public 1-bit strategy.}
It is convenient to describe the 1-bit strategy \new{in $\hat\game$ in terms of a memoryless
strategy in a derived} game~$\game$ with state space $S \times \{0,1\}$, where the second component ($0$ or~$1$) reflects the current memory mode of Maximizer.
Accordingly, we think of the state space of~$\game$ as organized in two \emph{layers}, the ``inner'' and the ``outer'' layer, with the memory mode being $0$ and~$1$, respectively.
\new{In each state $(s,j)$ of~$\game$ (where $j \in \{0,1\}$ denotes the layer), Maximizer can choose the layer, $j' \in \{0,1\}$, of the successor state~$(s',j')$, possibly depending on~$s'$. This is exactly analogous to Maximizer using $1$~bit of memory.}
In these terms, our goal is to construct, for the layered game~$\game$, a \emph{memoryless} strategy for Maximizer.
\new{From this one can naturally extract a public 1-bit strategy for Maximizer in the original game~$\hat\game$.}
Upon reaching the target, the memory mode is irrelevant, so for notational simplicity we denote the objective as $\reach{T}$, also in the layered game~$\game$ (instead of $\reach{T \times \{0,1\}}$).
The current state of the layered game is known to both players (to Minimizer in particular); this corresponds to the (1-bit) memory being public in the original game~$\hat\game$: at each point in the game, Minimizer knows the distribution of actions that Maximizer is about to play.
Notice that the values of states $(s,0)$ and $(s,1)$ in~$\game$ are equal to the value of~$s$ in~$\hat\game$; this is because the definition of value does not impose restrictions on the memory of strategies and so the players could, in~$\hat\game$, simulate the two layers of~$\game$ in their memory if that were advantageous.

In general Maximizer does not have uniform\new{ly} $\eps$-optimal memoryless strategies in reachability games; cf.\ \cref{thm:BMI-zplus}.
So our construction will exploit the special structure in the layered game, namely, the symmetry of the two layers.
The memoryless Maximizer strategy we construct will be $\eps$-optimal from each state $(s,0)$ in the inner layer, but not necessarily from the states in the outer layer.

As building blocks we use the non-uniformly \verynew{$\eps$-optimal}
memoryless strategies that we get from \cref{lem:conc-reach-non-uniform};
in the turn-based finitely branching case they are even MD.
We combine them by ``plastering'' the state space (of the layered game).
This is inspired by the construction in \cite{Ornstein:AMS1969}; see \cite[Section~3.2]{KMSTW2020} for a recent description.\footnote{These papers consider MDPs, i.e., Minimizer is passive.
In countable MDPs, Maximizer has uniform\new{ly} $\eps$-optimal MD strategies
even without layering the system
\cite[Proposition A]{Ornstein:AMS1969}.
}

In the general concurrent case, a memoryless strategy prescribes for each state $(s,i)$ a probability distribution over Maximizer's actions.
We define a memoryless strategy by successively \emph{fixing} such distributions in more and more states.
Technically, one can \emph{fix} a state~$s$ by replacing the actions~$A(s)$ available to Maximizer by a single action which is a convex combination over~$A(s)$.
Visually, we ``plaster'' the whole state space by the fixings.
This is in general an infinite (but countable) process; it defines a memoryless strategy for Maximizer in the limit.

In the turn-based and finitely branching case, an MD strategy prescribes one outgoing transition for each Maximizer state.
Accordingly, \emph{fixing} a Maximizer state means restricting the outgoing transitions to a single such outgoing transition.
The plastering proceeds similarly as in the concurrent case; it defines an MD strategy for Maximizer in the limit.

Put the states of~$\hat\game$ in some order, i.e., $s_1, s_2, \ldots$ with $S = \{s_1, s_2, \ldots\}$.
The plastering proceeds in \emph{rounds}.
In round $i \ge 1$ we fix the states in $S_i[0] \times \{0\}$ and in $S_i[1] \times \{1\}$, where $S_1[0], S_2[0], \ldots \subseteq S$ are pairwise disjoint and $S_1[1], S_2[1], \ldots \subseteq S$ are pairwise disjoint; see~\Cref{fig:uniform-eps-two-layers} for an example of sets $S_i[0]$ and~$S_i[1]$ in a two-layer game.
Define $F_i[0] \eqdef \bigcup_{j \le i} S_j[0]$ and $F_i[1] \eqdef \bigcup_{j \le i} S_j[1]$.
So $F_i[0] \times \{0\}$ and $F_i[1] \times \{1\}$ are the states that have been fixed by the end of round~$i$.
We will keep an invariant $F_i[0] \subseteq F_i[1] \subseteq F_{i+1}[0]$.

Let $\game_i$ be the game obtained from~$\game$ after the fixings of the first $i-1$ rounds (with $\game_1 = \game$).
Define \[S_i[0] \eqdef (\{s_i\} \cup S_{i-1}[1]) \setminus F_{i-1}[0]\]
 (and $S_1[0] \eqdef \{s_1\}$), the set of states to be fixed in round~$i$.
In particular, round~$i$ guarantees that states $(s,0)$ whose ``outer sibling'' $(s,1)$ has been fixed previously are also fixed, ensuring $F_{i-1}[1] \subseteq F_i[0]$.
\new{It follows from the invariant above} that $\bigcup_{j=1}^\infty S_j[0] = \bigcup_{j=1}^\infty S_j[1] = S$.
The set $S_i[1]$ will be defined below.

\begin{figure}
\begin{center}
    \includegraphics[width=0.8\textwidth, angle=0]{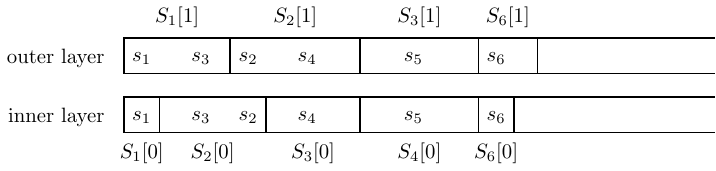}
\end{center}
\caption{Example of sets $S_i[0]$ and~$S_i[1]$ in the two layers of~$\game$, where the outer and inner layers are $S\times \{1\}$ and $S \times \{0\}$, respectively.}
\label{fig:uniform-eps-two-layers}
\end{figure}

In round~$i$ we fix the states in $(S_i[0] \times \{0\}) \cup (S_i[1] \times \{1\})$ in such a way that
\begin{enumerate}[(A)]
\item starting from any $(s,0)$ with $s \in S_i[0]$, the (infimum over all Minimizer strategies~$\ostrat$) probability of reaching~$T$ using only fixed states is not much less than the value $\valueof{\game_i,\reach{T}}{(s,0)}$; and
\item for all states $(s,0) \in S \times\{0\}$ in the inner layer, the value $\valueof{\game_{i+1},\reach{T}}{(s,0)}$ is almost as high as $\valueof{\game_{i},\reach{T}}{(s,0)}$.
\end{enumerate}
The purpose of goal~(A) is to guarantee good progress towards the target when starting from any state $(s,0)$ in~$S_i[0] \times \{0\}$.
The purpose of goal~(B) is to avoid fixings that would cause damage to the values of other states in the inner layer.

We want to define the fixings in round~$i$.
First we define an auxiliary game $\bar\game_i$ with state space $\bar S_i \eqdef (F_i[0] \times \{0,1\}) \cup (S \setminus F_i[0])$.
Game~$\bar\game_i$ is obtained from~$\game_i$ by collapsing, for all $s \in S \setminus F_i[0]$, the siblings $(s,0), (s,1)$ (neither of which have been fixed yet) to a single state~$s$.
See \Cref{fig:bar-game}.
The game~$\bar\game_i$ inherits the fixings from~$\game_i$.
The values remain equal; in particular, for $s \in S \setminus F_i[0]$, the values of $(s,0)$ and $(s,1)$ in~$\game_i$ and the value of $s$ in~$\bar\game_i$ are all equal.

\begin{figure}
\begin{center}
    \includegraphics[width=0.8\textwidth, angle=0]{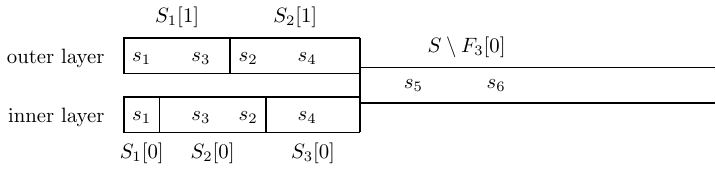}
\end{center}
\caption{Example of a game $\bar\game_3$.}
\label{fig:bar-game}
\end{figure}

Let $\eps_i > 0$.
We apply \cref{lem:conc-reach-non-uniform} to~$\bar\game_i$ with set of initial states $S_i[0] \times \{0\}$.
So Maximizer has a memoryless strategy~$\sigma_i$ for~$\bar\game_i$ and a finite set of states $R \subseteq \bar S_i$ so that for all $s \in S_i[0]$ we have $\inf_\ostrat \probm_{\bar\game_i, (s,0),\sigma_i,\ostrat}(\reachn{R}{T}) \ge \valueof{\game_i,\reach{T}}{(s,0)} - \eps_i$.

Now we carry the strategy $\sigma_i$ from $\bar\game_i$ to $\game_i$ by suitably
adapting it (see below). Then we obtain $\game_{i+1}$ from $\game_i$ by fixing
(the adapted version of) $\sigma_i$ in $\game_i$.

\new{The adaption of~$\sigma_i$ to~$\game_i$ is by treating states $s \in S \setminus F_i[0]$ in~$\bar\game_i$ as states in the \emph{outer} layer $(s,1)$ of~$\game_i$, as follows.
Every transition that in~$\bar\game_i$ goes from a state $(s,j) \in F_i[0] \times \{0,1\}$ to a state $s' \in S \setminus F_i[0]$ is redirected so that in~$\game_i$ it goes from $(s,j)$ to~$(s',1)$.
Similarly, every transition that in~$\bar\game_i$ goes from a state $s' \in S \setminus F_i[0]$ to a state $(s,j) \in F_i[0] \times \{0,1\}$ goes in~$\game_i$ from $(s',1)$ to~$(s,j)$.
Finally, every transition that in~$\bar\game_i$ goes from a state $s' \in S \setminus F_i[0]$ to another state $t' \in S \setminus F_i[0]$ goes in~$\game_i$ from $(s',1)$ to~$(t',1)$.
}

Accordingly, define $S_i[1] \eqdef (S_i[0] \setminus F_{i-1}[1]) \cup ((S \setminus F_i[0]) \cap R)$ (this ensures that $F_i[0] \subseteq F_i[1]$), and obtain~$\game_{i+1}$ from~$\game_i$ by fixing the adapted version of~$\sigma_i$ in $(S_i[0] \times \{0\}) \cup (S_i[1] \times \{1\})$.
This yields, for all $s \in S_i[0]$,

\begin{equation}
\begin{aligned}
\label{eq:invoke-non-uniform}
\inf_{\sigma,\ostrat} \probm_{\game_{i+1}, (s,0),\sigma,\ostrat} (\reachn{(F_i[0] \times \{0\}) \cup (F_i[1] \times \{1\})}{T}) \ge\valueof{\game_i,\reach{T}}{(s,0)} - \eps_i,
\end{aligned}
\end{equation}
achieving goal~(A) above.
Notice that the fixings in~$\game_{i+1}$ ``lock in'' a good attainment from $S_i[0] \times \{0\}$, regardless of the Maximizer strategy~$\sigma$.
Now we extend \eqref{eq:invoke-non-uniform} to achieve goal~(B) from above:
for all $s \in S$ we have
\begin{equation}
 \valueof{\game_{i+1},\reach{T}}{(s,0)} \ \ge \ \valueof{\game_i,\reach{T}}{(s,0)} - \eps_i\,. \label{eq:invoke-non-uniform-extended}
\end{equation}
Indeed, consider any Maximizer strategy $\sigma$ in~$\game_i$ from any~$(s,0)$.
Without loss of generality we can assume that $\sigma$ is such that the play enters the outer layer only (if at all) after having entered $F_i[0] \times \{0\}$.
Now change $\sigma$ to a strategy $\sigma'$ in~$\game_{i+1}$ so that as soon as $F_i[0] \times \{0\}$ is entered, $\sigma'$ respects the fixings (and plays arbitrarily afterwards).
By \eqref{eq:invoke-non-uniform} this decreases the (infimum over Minimizer strategies~$\ostrat$) probability by at most~$\eps_i$.
Thus,
\[
 \inf_\ostrat \probm_{\game_{i+1}, (s,0),\sigma',\ostrat}(\reach{T}) ~\ge~ \inf_\ostrat \probm_{\game_{i}, (s,0),\sigma,\ostrat}(\reach{T}) - \eps_i\,.
\]
Taking the supremum over strategies~$\sigma$ in~$\game_i$ yields \eqref{eq:invoke-non-uniform-extended}.

For any $\eps > 0$ choose $\eps_i \eqdef 2^{-i} \eps$; thus, $\sum_{i \ge 1} \eps_i = \eps$.
Let $\sigma$ be the memoryless strategy that respects all fixings in all $\game_i$.
Then, by \eqref{eq:invoke-non-uniform-extended}, for all $s \in S$ we have
\[
\inf_\ostrat \probm_{\game, (s,0),\sigma,\ostrat}(\reach{T}) \ \ge \ \valueof{\game,\reach{T}}{(s,0)} - \sum_{i=1}^\infty \eps_i\,,
\]
so $\sigma$ is $\eps$-optimal in~$\game$ from all $(s,0)$.
Hence, the corresponding public 1-bit memory strategy (with initial memory mode~$0$, corresponding to the inner layer) is uniform\new{ly} $\eps$-optimal in~$\hat\game$.
\end{proof}

\section{Uniform Strategies in Turn-based Games}\label{sec:uniturn}
\new{
\Cref{thm:BMI-zplus} in the previous section shows that Maximizer
has no \emph{uniformly} $\eps$-optimal memoryless strategies in
concurrent reachability games with finite action sets (since the concurrent Big Match game on $\N$
is a counterexample).
Here we strengthen this negative result by showing that it holds even for the
subclass of finitely branching \emph{turn-based} reachability games.
}

To this end, we define a finitely branching
turn-based game in \cref{def:turn-big-match} that is very similar to the
concurrent Big Match  on $\N$, as shown in~\Cref{fig:bigmatchNconcurr}.
The difference is that
at each $c_i$ Maximizer has to announce his mixed choice of actions first,
rather than concurrently with Minimizer.
Note that Maximizer only announces his distribution over the actions
$\{0,1\}$, not any particular action.
Since a good Maximizer strategy needs to work even if Minimizer knows
it in advance, this makes no difference with respect
to the attainment of memoryless Maximizer strategies.
Another slight difference is that Maximizer is restricted to choosing
distributions with only \emph{rational} probabilities where the probability of picking
action $1$ is of the form $1/k$ for some $k \in \N$.
However, since we know that there exist good Maximizer strategies
of this form (cf.\ \cref{thm:BMI-zplus}), it is not a significant restriction.

\begin{figure}[t]
    \centering
    \includegraphics[width=0.8\textwidth, angle=0]{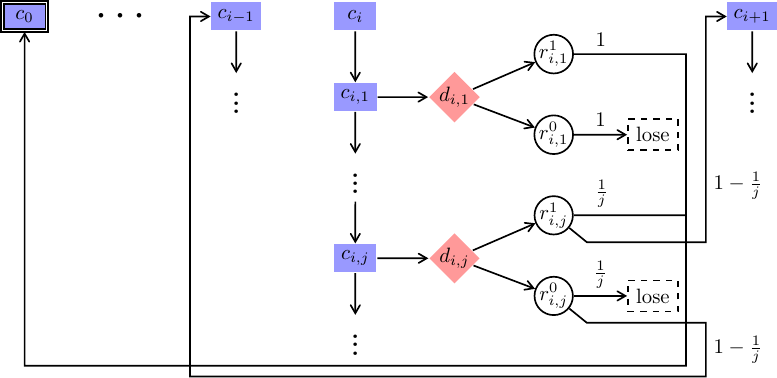}
    \smallskip
   \caption{Turn-based Big Match on~$\N$.}
 \label{fig:bigmatchturned}
\end{figure}
\begin{definition}[Turn-based Big Match on $\N$]\label{def:turn-big-match}
This game  is shown in \Cref{fig:bigmatchturned}. Maximizer controls the set $\{c_i \mid i \in \nat\} \cup \{c_{i,j} \mid i,j \in \nat\} \cup \{\lose\}$ of states, whereas Minimizer controls only the states in $\{d_{i,j} \mid i,j \in \nat\}$. The
remaining set $\{r_{i,j}^0, r_{i,j}^1 \mid i,j \in \nat\}$ of states are random. 
For all $i,j \in \nat$, there are \new{the} following transitions  
\begin{align*}c_i \transition c_{i,1} \qquad \quad c_{i,j} \transition c_{i,j+1} \qquad \quad
c_{i,j} \transition d_{i,j} \\
d_{i,j} \transition r_{i,j}^0 \qquad \quad d_{i,j} \transition r_{i,j}^1
\end{align*}
 and $\lose \transition \lose$.
Intuitively, by going from $c_i$ to $d_{i,j}$, Maximizer chooses action $1$ with
probability $1/j$ and action $0$ with probability $1-1/j$.
Minimizer chooses actions $0$ or $1$ by going from
$d_{i,j}$ to $r_{i,j}^0$ or $r_{i,j}^1$, respectively.
The probabilistic function is defined by
\begin{align*}P(r_{i,j}^0)(\lose) = 1/j \qquad & \qquad P(r_{i,j}^0)(c_{i-1}) = 1-1/j
\\
 P(r_{i,j}^1)(c_0) = 1/j\qquad & \qquad P(r_{i,j}^1)(c_{i+1}) = 1-1/j\end{align*}
where  $i,j \in \nat$. 
The objective is $\reach{\{c_0\}}$.

This finitely branching turn-based game
mimics the behavior of the game in \cref{def:concurrent-big-match-zplus}.
\end{definition}
\begin{theorem}\label{thm:TB-BMI-zplus}
Consider the turn-based Big Match game $\game$ on $\N$
from \cref{def:turn-big-match} and let \new{$x\in \N$}.

\begin{enumerate}
\item\label{thm:TB-BMI-zplus-2}
For every start state $c_x$ and $N\ge 0$,
Maximizer can win with probability $\ge N/(2N+2)$ by choosing
the transitions $c_i \transition \dots d_{i,j}$
where $j = (n+1)^2$ whenever he is in state $c_i$ with $i=x+N-n$ for some $n \ge 0$.

In particular, $\valueof{\game}{c_x} \ge 1/2$.
\item\label{thm:TB-BMI-zplus-3}
For any $\eps < 1/2$
there does not exist any uniformly $\eps$-optimal memoryless (MR) strategy for
Maximizer.

Every MR Maximizer strategy $\zstrat$ attains arbitrarily little from $c_x$ as
$x \to \infty$.
Formally,
$
\limsup_{x \to \infty}\inf_{\ostrat}\probm_{\game,c_x,\zstrat,\ostrat}(\reach{\{c_0\}})=0.
$
\end{enumerate}
\end{theorem}
\begin{proof}
Let $\game'$ be the concurrent game
from \cref{def:concurrent-big-match-zplus}.

Towards item~\ref{thm:TB-BMI-zplus-2}, consider the concurrent game
$\game'$ and the turn-based game $\game$ from
\cref{def:turn-big-match}. Let $c_x$ be our start state. 
After fixing the Maximizer strategy from \cref{thm:BMI-zplus}(2) in $\game'$,
we obtain an MDP $\mdp'$ from Minimizer's point of view.
Similarly, after fixing the strategy described above in $\game$, we obtain an
MDP $\mdp$. Then $\mdp'$ and $\mdp$ are almost isomorphic (apart from linear chains
of steps \new{$c_{i,j}$} $\transition c_{i,j+1} \dots$ in $\mdp'$), and thus the
infimum of the chance of winning, over all Minimizer strategies are the same.
Therefore the result follows from \cref{thm:BMI-zplus}(2).

Towards item~\ref{thm:TB-BMI-zplus-3}, note that every Maximizer MR strategy $\zstrat$
in $\game$ corresponds to a Maximizer MR strategy $\zstrat'$ in $\game'$.
First, forever staying in states $c_{i,j}$ is losing, since the target is
never reached. Thus, without restriction, we assume that $\zstrat$ almost
surely moves from
$c_i$ to some $d_{i,j}$ eventually. Let $p_{i,j}$ be the probability that
$\zstrat$ moves from $c_i$ to $d_{i,j}$.
Thus the corresponding strategy $\zstrat'$ in $\game'$
in $c_i$ plays action $1$ with probability $\sum_j p_{i,j}(1/j)$ and action
$0$ otherwise.
Again the MDPs resulting from fixing the respective strategies in $\game$
and $\game'$ are (almost) isomorphic, and thus the result follows from
\cref{thm:BMI-zplus}(3).
\end{proof}

\new{
In the rest of this section we briefly describe an alternative
construction of a turn-based finitely branching reachability game without uniformly
$\eps$-optimal memoryless Maximizer strategies, i.e., \Cref{thm:TB-BMI-simple}
is a different proof of the same
result as in \Cref{thm:TB-BMI-zplus}.
In the direct construction in \Cref{def:turn-big-match},
Maximizer had many alternatives in the states $c_i$ (by going to some state
$d_{i,j}$ for some $j \ge 1$).
However, \cref{thm:conc-reach-uniform} shows that a deterministic 1-bit
strategy suffices for Maximizer. Thus, it suffices for Maximizer to have just
two alternatives, corresponding to the two memory modes of the 1-bit strategy.
The following definition uses this observation to construct an alternative counterexample.
}

\begin{definition}\label{def:TB-BMI-simple}
Consider the concurrent reachability game
from \cref{def:concurrent-big-match-zplus} and let $\eps = 1/4$.
By \cref{thm:conc-reach-uniform}, Maximizer has a uniform $\eps$-optimal
1-bit strategy $\hat{\zstrat}$. Let $p_{i,0}$ (resp.\ $p_{i,1}$)
be the probability that $\hat{\zstrat}$ picks action $1$ at state $c_i$ when
in memory mode $0$ (resp.\ memory mode $1$).

We construct a turn-based reachability game $\game$ with branching
degree two where Maximizer can pick randomized actions according to these
probabilities $p_{i,0}, p_{i,1}$, but nothing else.
Let $\game=\gametuple$ where
$\zstates = \{c_i \mid i \in \nat\} \cup \{\lose\}$,
$\ostates = \{d_{i,0}, d_{i,1} \mid i \in \nat\}$,
and
$\rstates = \{r_{i,0,0}, r_{i,0,1}, r_{i,1,0}, r_{i,1,1} \mid i \in \nat\}$.
We have controlled transitions
$c_i \transition d_{i,j}$, 
$d_{i,j} \transition r_{i,j,k}$ for all $i \in \nat$ and $j,k \in \{0,1\}$ and
$\lose \transition \lose$.
Intuitively, by going from $c_i$ to $d_{i,j}$, Maximizer chooses action $1$ with
probability $p_{i,j}$ and action $0$ otherwise.
Minimizer chooses action $k$ by going from
$d_{i,j}$ to $r_{i,j,k}$.
The random transitions are defined by
$P(r_{i,j,0})(\lose) = p_{i,j}$,
$P(r_{i,j,0})(c_{i-1}) = 1-p_{i,j}$,
$P(r_{i,j,1})(c_0) = p_{i,j}$,
$P(r_{i,j,1})(c_{i+1}) = 1-p_{i,j}$.

The objective is $\reach{\{c_0\}}$.
\end{definition}

\begin{theorem}\label{thm:TB-BMI-simple}
Consider the turn-based reachability game $\game$ of branching degree two 
from \cref{def:TB-BMI-simple} and let \new{$x\in \N$}.
\begin{enumerate}
\item\label{thm:TB-BMI-simple-1}
$\valueof{\game}{c_x} \ge 1/4$.
\item\label{thm:TB-BMI-simple-2}
There does not exist any uniformly $\eps$-optimal memoryless (MR) strategy for
Maximizer.

Every MR Maximizer strategy $\zstrat$ attains arbitrarily little from $c_x$ as
$x \to \infty$.
Formally,
$
\limsup_{x \to \infty}\inf_{\ostrat}\probm_{\game,c_x,\zstrat,\ostrat}(\reach{\{c_0\}})=0.
$
\end{enumerate}
\end{theorem}
\begin{proof}
Let $\game'$ be the concurrent game from \cref{def:concurrent-big-match-zplus}.
We have $\valueof{\game'}{c_x} \ge 1/2$ by \cref{thm:BMI-zplus}.
Consider the $(1/4)$-optimal 1-bit Maximizer strategy $\hat{\zstrat}$ used in $\game'$
in \cref{def:TB-BMI-simple}.
We can define a corresponding 1-bit Maximizer strategy $\zstrat$ in $\game$.
In every state $c_i$, it picks the move $c_i \transition d_{i,j}$
whenever its memory mode is $j$, and it updates its memory in the same way
as $\hat{\zstrat}$.
Then
\begin{align*}\valueof{\game}{c_x}
&\ge \inf_\ostrat\probm_{\game,c_x,\zstrat,\ostrat}(\reach{\{c_0\}})\\
&= \inf_\ostrat\probm_{\game',c_x,\hat{\zstrat},\ostrat}(\reach{\{c_0\}})\\
&\ge \valueof{\game'}{c_x} - 1/4 \ge 1/4.
\end{align*}

For item~\ref{thm:TB-BMI-simple-2}
the argument is exactly the same as in \cref{thm:TB-BMI-zplus}(2).
\end{proof}
 
\section{No \verynew{Memoryless} Strategies for Reachability in Infinitely Branching Games}\label{sec:positional}
\begin{figure}[t]
\includegraphics[width=\textwidth, angle=0]{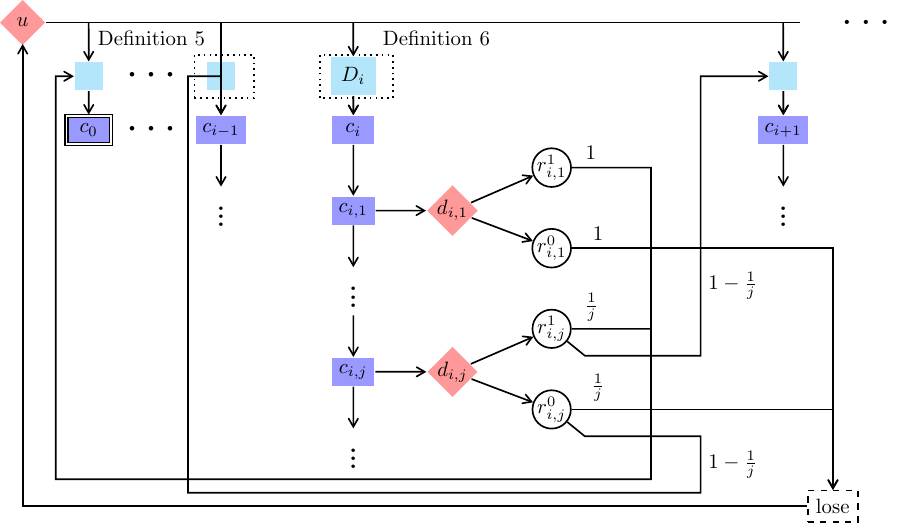}
\caption{The scheme of games defined in~\Cref{def:inf-branch-no-MR} and in~\Cref{def:inf-branch-no-Markov}, where in the former the 
  teal-colored boxes  are replaced with a line connecting $u$ directly to the $c_i$, whereas in the latter  such a  teal-colored box is replaced with a Delay gadget~$D_i$ illustrated in~\Cref{fig:Delay}.
   }
 \label{fig:TheGame-inf-barnch-Pos-Markov}
\end{figure}

In finitely branching turn-based stochastic 2-player games with reachability objectives,
Maximizer has $\eps$-optimal MD strategies (\Cref{lem:conc-reach-non-uniform}).
We go on to show that this does not carry over to infinitely branching turn-based
reachability games. In this case, there are not even $\eps$-optimal MR
strategies, i.e., good Maximizer strategies need memory.
The reason for this is the infinite branching of Minimizer.
\new{
Infinite
branching of Maximizer states and random states does not make a difference
in the case of reachability objectives.
Each infinitely branching Maximizer state can be encoded into a gadget
containing an infinite sequence
of binary branching Maximizer states, where the sequence must eventually be
left, because the target state is not on the sequence.
Similarly, each infinitely branching random state can be encoded into a gadget
containing an infinite sequence of binary branching random states, where the
sequence is left eventually almost surely.
Such an encoding is not possible for infinitely branching Minimizer states,
because Minimizer could choose to stay inside the gadget forever, and spuriously win the game.
(Strictly speaking, such encodings do not preserve path lengths. However, we
show in \Cref{sec:acyclic} that a step counter does not help Maximizer anyway.)
}

\begin{definition}\label{def:inf-branch-no-MR}
Consider the finitely branching turn-based reachability game
from \cref{def:turn-big-match}. We construct an infinitely branching game
$\game$ by adding a new Minimizer-controlled initial state $u$, Minimizer-transitions
$u \transition c_i$ for all
$i \in \N$ and $\lose \transition u$. See \Cref{fig:TheGame-inf-barnch-Pos-Markov} for a scheme of this game. 
The objective is still $\reach{\{c_0\}}$.
\end{definition}

\begin{theorem}\label{thm:inf-branch-no-MR}
Let $\game$ be the infinitely branching turn-based
reachability game from \cref{def:inf-branch-no-MR}.
\begin{enumerate}
\item\label{thm:inf-branch-no-MR-1}
All states in $\game$ are almost surely winning. I.e., for every state $s$ there
exists a Maximizer strategy $\zstrat$ such that
$\inf_{\ostrat}\probm_{\game,s,\zstrat,\ostrat}(\reach{\{c_0\}})=1$.
\item\label{thm:inf-branch-no-MR-2}
For each MR Maximizer strategy $\zstrat$ we have
\[\inf_{\ostrat}\probm_{\game,u,\zstrat,\ostrat}(\reach{\{c_0\}})=0.\]
I.e., for any $\eps < 1$ there does not exist any $\eps$-optimal MR Maximizer
strategy $\zstrat$ from state $u$.
\end{enumerate}
\end{theorem}
\begin{proof}
Towards item~\ref{thm:inf-branch-no-MR-1}, first note that by playing
\[c_i \transition c_{i,1} \transition d_{i,1}\] Maximizer can enforce that he either wins (if
Minimizer goes to $r_{i,1}^1$) or the game returns to state $u$ (via state
$\lose$ if Minimizer goes to $r_{i,1}^0$).
Thus it suffices to show that Maximizer can win almost surely from state $u$.
We construct a suitable strategy $\zstrat$ (which is not MR).
By \cref{thm:TB-BMI-zplus}, $\valueof{\game}{c_x} \ge 1/2$ for every $x$.
\new{Moreover, the subgraph between any $c_x$ and a return to $u$ (which goes via a losing state and is thus to be avoided by Maximizer) is finitely branching.}
Thus there exists a strategy $\zstrat_x$ and a finite horizon $h_x$ such that
$\inf_\ostrat \probm_{\game,c_x,\zstrat_x,\ostrat}(\reachn{h_x}{\{c_0\}}) \ge 1/4$.
Then $\zstrat$ plays from $u$ as follows. If Minimizer moves
$u \transition c_x$ then first play $\zstrat_x$ for $h_x$ steps, unless $c_0$ or $u$
are reached first. Then play to reach $u$ again, i.e., the next
time that the play reaches a state $c_i$ play
$c_i \transition c_{i,1} \transition d_{i,1}$ (thus either Maximizer wins or the play returns to $u$).
So after every visit to $u$ the Maximizer strategy $\zstrat$ wins with
probability $\ge 1/4$ before seeing $u$ again, and otherwise the play returns to $u$.
Thus $\inf_{\ostrat}\probm_{\game,s,\zstrat,\ostrat}(\reach{\{c_0\}}) \ge 1 -
(1/4)^\infty = 1$.
This proves item~\ref{thm:inf-branch-no-MR-1}.

\begin{claim}\label{thm:inf-branch-no-MR-claim-1}
Suppose that for each MR Maximizer strategy~$\sigma$ and every $\eps > 0$ there is a Minimizer strategy~$\pi(\sigma,\eps)$ such that, from~$u$, the probability of visiting~$c_0$ before revisiting~$u$ is at most~$\eps$.
Then item~\ref{thm:inf-branch-no-MR-2} is true.
\end{claim}
\begin{proof}[Proof of the claim]
Suppose the precondition of the claim.
Let $\sigma$ be an MR Maximizer strategy.
Let $\pi$ be the Minimizer strategy that, after the $i$-th visit to $u$, continues to play $\pi(\sigma,\eps \cdot 2^{-i})$ until the next visit to~$u$.
It follows that
$\probm_{\game,u,\zstrat,\ostrat}(\reach{\{c_0\}}) \le \sum_{i \ge 1} \eps \cdot 2^{-i} = \eps$.
\end{proof}
It remains to prove the precondition of the claim.
Let $\sigma$ be an MR Maximizer strategy, and let $\eps > 0$.
By \cref{thm:TB-BMI-zplus}(2), there are $i \in \N$ and a Minimizer strategy~$\pi$ such that $\probm_{\game',c_i,\zstrat,\ostrat}(\reach{\{c_0\}}) \le \eps$, where $\game'$ is the finitely branching subgame of $\game$ from \cref{def:turn-big-match}.
The Minimizer strategy $\pi(\sigma,\eps)$ that, in~$\game$, from~$u$ goes to~$c_i$ and then plays~$\pi$ has the required property.
\end{proof}

\begin{figure}[t]
    \centering
\includegraphics[width=0.8\textwidth, angle=0]{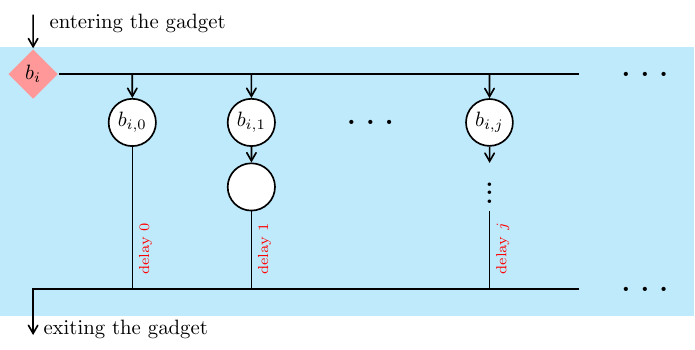}
\caption{Delay gadget~$D_i$, used in~\Cref{fig:TheGame-inf-barnch-Pos-Markov}.}
 \label{fig:Delay}
\end{figure}
 
\section{No Markov Strategies for Reachability in Infinitely Branching Games}\label{sec:acyclic}
\new{Recall that Markov strategies are strategies that use just a step counter as memory.}
We strengthen the result from the previous section by modifying the game so
that even Markov strategies are useless
for Maximizer.
The modification of the game allows Minimizer to cause an arbitrary but finite delay before any state~$c_i$ is entered.

\begin{definition}\label{def:inf-branch-no-Markov}
Consider the infinitely branching turn-based reachability game
from \cref{def:inf-branch-no-MR}. We modify it as follows.
For each $i \in \nat$ we add a Minimizer-controlled state~$b_i$ and redirect all transitions going into $c_i$ to go into $b_i$ instead.
Each~$b_i$ is infinitely branching: for each $j \in \nat$ we add a random state~$b_{i,j}$ and a transition $b_i \transition b_{i,j}$.
We add further states so that the game moves (deterministically, via a chain
of random states) from $b_{i,j}$ to~$c_i$ in exactly $j$~steps. See \Cref{fig:TheGame-inf-barnch-Pos-Markov,fig:Delay} for a depiction of this game. 
The objective is still $\reach{\{c_0\}}$.
\end{definition}

\begin{theorem}\label{thm:inf-branch-no-Markov}
Let $\game$ be the infinitely branching turn-based
reachability game from \cref{def:inf-branch-no-Markov}.
\begin{enumerate}
\item\label{thm:inf-branch-no-Markov-1}
All states in $\game$ are almost surely winning. I.e., for every state $s$ there
exists a Maximizer strategy $\zstrat$ such that
$\inf_{\ostrat}\probm_{\game,s,\zstrat,\ostrat}(\reach{\{c_0\}})=1$.
\item\label{thm:inf-branch-no-Markov-2}
For every Markov Maximizer strategy~$\zstrat$ it holds that
$\inf_{\ostrat}\probm_{\game,u,\zstrat,\ostrat}(\reach{\{c_0\}})=0$.
I.e., no Markov Maximizer strategy is $\eps$-optimal from state $u$ for any $\eps < 1$.
\end{enumerate}
\end{theorem}
\begin{proof}
Item~\ref{thm:inf-branch-no-Markov-1} follows from \cref{thm:inf-branch-no-MR}(\ref{thm:inf-branch-no-MR-1}), as the modification in \cref{def:inf-branch-no-Markov} only allows Minimizer to cause finite delays.

Towards item~\ref{thm:inf-branch-no-Markov-2}, the idea of the proof is that for every Markov Maximizer strategy~$\zstrat$, Minimizer can cause delays that make $\zstrat$ behave in the way it would after a long time.
This way, Minimizer turns~$\zstrat$ approximately to an MR-strategy, which is useless by \cref{thm:inf-branch-no-MR}(\ref{thm:inf-branch-no-MR-2}).

In more detail, fix any Markov Maximizer strategy~$\zstrat$.
As in the proof of \cref{thm:TB-BMI-zplus}(\ref{thm:TB-BMI-zplus-3}), we can assume that whenever the game is in~$c_i$, the strategy~$\zstrat$ almost surely moves eventually to some~$d_{i,j}$.
Let $p_{i,j,t}$ be the probability that strategy~$\zstrat$, when it is in~$c_i$ at time~$t$, moves to~$d_{i,j}$.
Thus, a corresponding Maximizer strategy in the (concurrent) Big Match, when it is in~$c_i$ at time~$t$, picks action~$1$ with probability $f(i,t) \eqdef \sum_j p_{i,j,t}(1/j)$; cf.\ the proof of \cref{thm:TB-BMI-zplus}(\ref{thm:TB-BMI-zplus-3}).
For each $i \in \nat$, let $f(i)$ be an accumulation point of $f(i,1), f(i,2), \ldots$; e.g., take $f(i) \eqdef \liminf_t f(i,t)$.
We have that
\verynew{
\begin{align}
\forall {i\in \nat}\ \forall {t_0 \in \nat}\ \forall {\eps > 0}\ \exists {t\ge t_0} : & \quad f(i,t) \le f(i)+\eps\label{eq:accu-1} \\
\forall {i\in \nat}\ \forall {t_0 \in \nat}\ \forall {\eps > 0}\ \exists {t\ge t_0} : & \quad f(i,t) \ge f(i)-\eps \label{eq:accu-2}
\end{align}
}
Similarly to the proof of \cref{thm:inf-branch-no-MR}(\ref{thm:inf-branch-no-MR-2}) (see \cref{thm:inf-branch-no-MR-claim-1} therein), it suffices to show that after each visit to~$u$, Minimizer can make the probability of visiting $c_0$ before seeing $u$ again arbitrarily small.
Let $\eps > 0$.
We show that Minimizer has a strategy~$\pi$ to make this probability at most~$\eps$.

Consider the first case where $\sum_{i \ge 1} f(i) < \infty$.
Then there is $i_0 \in \nat$ such that $\sum_{i \ge i_0} f(i) \le \eps/2$.
In~$u$, strategy~$\pi$ moves to $b_{i_0}$.
Whenever the game is in some~$b_i$, strategy~$\pi$ moves to some~$b_{i,j}$ so that the game will arrive in~$c_i$ at a time~$t$ that satisfies $f(i,t) \le f(i) + 2^{-i} \cdot \eps/2$; such $t$ exists due to \eqref{eq:accu-1}.
In~$c_i$ Maximizer (using~$\zstrat)$ moves (eventually) to some~$d_{i,j}$.
Then $\pi$ always chooses ``action~$1$''; i.e., $\pi$ moves to~$r_{i,j}^1$.
In this way, the play, restricted to states~$c_i$, is either of the form
$c_{i_0}, c_{i_0+1}, \ldots$ (Maximizer loses)
or of the form $c_{i_0}, c_{i_0+1}, \ldots, c_{i_0+k}, c_0$ (Maximizer wins).
The probability of the latter can be bounded similarly to the proof of \cref{thm:BMI-zplus}(\ref{thm:BMI-zplus-3}); i.e., we have
\begin{equation*}
\probm_{\game,u,\zstrat,\ostrat}(\reach{\{c_0\}})
~\le~
\sum_{i=i_0}^\infty f(i) + 2^{-i} \cdot \frac\eps2
~\le~
\frac\eps2 + \frac\eps2
~=~
\eps.
\end{equation*}

Now consider the second case where $\sum_{i \ge 1} f(i) = \infty$.
Then there is $i_0 \in \nat$ such that $\sum_{i=1}^{i_0} f(i) \ge \frac1\eps$.
In~$u$, strategy~$\pi$ moves to $b_{i_0}$.
Whenever the game is in some~$b_i$, strategy~$\pi$ moves to some~$b_{i,j}$ so that the game will arrive in~$c_i$ at a time~$t$ that satisfies $f(i,t) \ge f(i) - 2^{-i}$; such $t$ exists due to \eqref{eq:accu-2}.
In~$c_i$ Maximizer (using~$\zstrat)$ moves (eventually) to some~$d_{i,j}$.
Then $\pi$ always chooses ``action~$0$''; i.e., $\pi$ moves to~$r_{i,j}^0$.
In this way, the play, restricted to states~$c_i$, is either of the form
$c_{i_0}, c_{i_0-1}, \ldots, c_0$ (Maximizer wins) or of the form
$c_{i_0}, c_{i_0-1}, \ldots, c_{i_0-k}$ (for some $k < i_0$), followed by $\lose, u$ (Maximizer does not reach~$c_0$ before revisiting~$u$).
The probability of the former can be bounded similarly to the proof of
\cref{thm:BMI-zplus}(\ref{thm:BMI-zplus-3}); i.e., the probability that the
play reaches $c_0$ before~$u$ is
upper-bounded by
\begin{align*}
    \prod_{i=1}^{i_0}(1 - \max\{f(i) - 2^{-i}, 0\})
    \quad
    &\le \quad \frac{1}{1+\sum_{i=1}^{i_0} (f(i) - 2^{-i})} &\text{by \cref{prop:Weierstrass}}\\ 
    &\le \quad \frac{1}{\sum_{i=1}^{i_0} f(i)} \; \le \ \eps.
\end{align*}
\end{proof}

\section{Good Strategies for Reachability Require Infinite Memory}\label{sec:infinite}
\begin{figure}[t]
\includegraphics[width=\textwidth, angle=0]{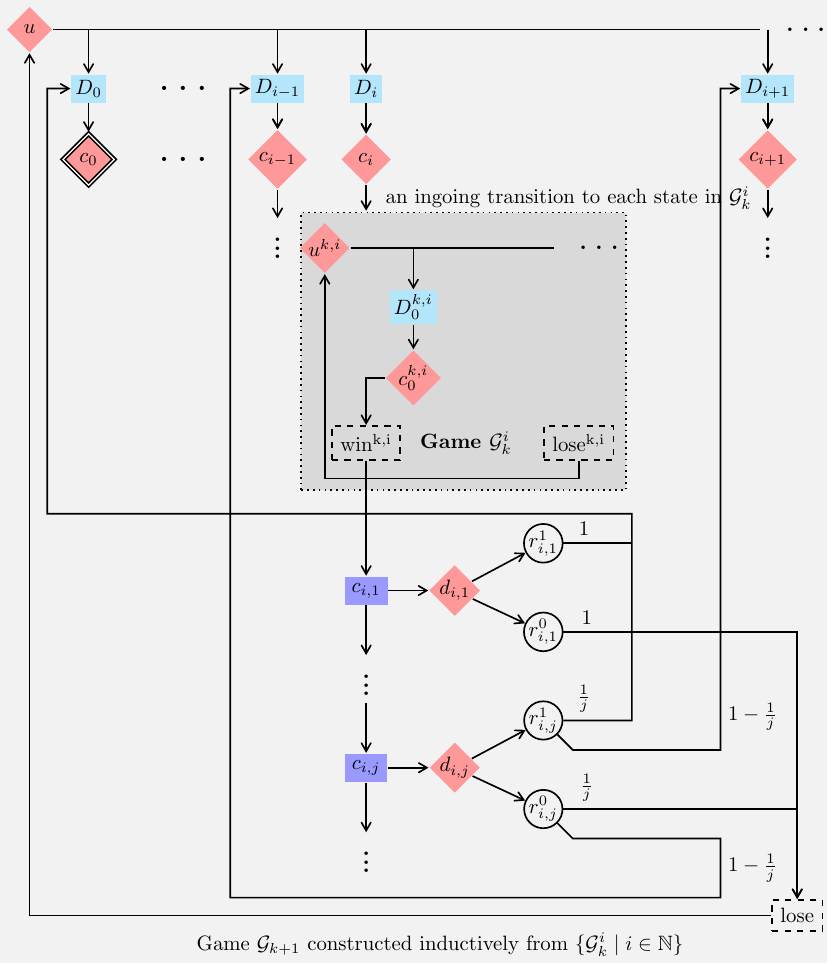}
\vspace{-.5cm}
   \caption{The scheme of the nested construction in~\Cref{def:nested-BMI}. }
 \label{fig:TheGame-inf}
\end{figure}

We show that even finite private memory, in addition to a step counter, is useless for
Maximizer in infinitely branching reachability games.
To this end, we define a nested version of the game
of \cref{def:inf-branch-no-Markov}, where the memory requirements increase
unboundedly with the nesting depth.

\begin{definition}\label{def:nested-BMI}
Let $\game$ be the game from \cref{def:inf-branch-no-Markov}.
We inductively define the \emph{$k$-nested} game $\game_k$ as follows
(see \cref{fig:TheGame-inf}).
For the base case, let $\game_1 \eqdef \game$.

For every $i\ge 1$ let $\game_k^i$ be a fresh copy of $\game_k$ and let
$u^{k,i}$ (resp.\ $c_0^{k,i}$) be the initial state $u$
(resp.\ the target state $c_0$) in $\game_k^i$.
For every $k \ge 1$ we construct $\game_{k+1}$ by modifying $\game$ as
follows.
The idea is that at every state $c_i$ Maximizer first needs to win the subgame $\game_k^i$
before continuing in the game $\game_{k+1}$, but Minimizer can choose at which
state $\state$ in $\game_k^i$ the subgame is entered.

We make the state $c_i$ Minimizer-controlled and
replace the previous Maximizer transition $c_i \transition c_{i,1}$
by Minimizer transitions $c_i \transition \state$ for every state $\state$ in
$\game_k^i$.
Moreover, we add the transitions 
$c_0^{k,i} \transition \win^{k,i} \transition c_{i,1}$.
(The new state $\win^{k,i}$ is not strictly needed. It just indicates that Maximizer
has won and exited the subgame $\game_k^i$.)
Note that also in $\game_{k+1}$ Minimizer can introduce arbitrary delays
between states $b_i$ and $c_i$.

The objective in $\game_{k+1}$ is still $\reach{\{c_0\}}$.
\end{definition}

\begin{restatable}{lemma}{lemBMInoSCplusF}\label{lem:BMI-no-SC-plus-F}
For any $k \ge 1$ let $\game_k$ be the infinitely branching turn-based
reachability game from \cref{def:nested-BMI}.
\begin{enumerate}
\item\label{lem:BMI-no-SC-plus-F-1}
All states in $\game_k$ are almost surely winning. I.e., for every state $\state$ there
exists a Maximizer strategy $\zstrat$ such that
\[
\inf_{\ostrat}\probm_{\game_k,\state,\zstrat,\ostrat}(\reach{\{c_0\}})=1.
\]
\item\label{lem:BMI-no-SC-plus-F-2}
For each Maximizer strategy~$\zstrat$ with a step counter plus a private
finite memory with $\le k$ modes
\[\inf_{\ostrat}\probm_{\game_k,u,\zstrat,\ostrat}(\reach{\{c_0\}})=0.\]
I.e., for any $\eps < 1$ there does not exist any $\eps$-optimal step counter plus
$k$ memory mode Maximizer strategy $\zstrat$ from state $u$ in $\game_k$. 
\end{enumerate}
\end{restatable}

\begin{proof}
We show \cref{lem:BMI-no-SC-plus-F-1} by induction on $k$.

In the base case of $k=1$ we have $\game_1 = \game$ from
\cref{def:inf-branch-no-Markov}, and thus the result holds by
\cref{thm:inf-branch-no-Markov}(\cref{thm:inf-branch-no-Markov-1}).

Induction step $k \transition k+1$.
For every state $\state$ in $\game_{k+1}$ outside of any subgame,
let $\zstrat'(\state)$ be the almost surely winning Maximizer strategy from $\state$ in the
non-nested game $\game_1$, obtained from above.
By the induction hypothesis, for any state $\state$ in a subgame $\game_k^i$
there exists a Maximizer strategy $\zstrat_k^i(\state)$
from $\state$ that almost surely wins this subgame $\game_k^i$.

We now construct a Maximizer strategy $\zstrat(\state)$ from any state
$\state$ in $\game_{k+1}$.
If $\state$ is not in any strict subgame then $\zstrat(\state)$ plays like
$\zstrat'(\state)$ outside of the subgames. Whenever a subgame $\game_k^i$ is
entered at some state $x$ then it plays like $\zstrat_k^i(x)$ until $\win^{k,i}$ is
reached (which happens eventually almost surely by the definition of $\zstrat_k^i(x)$)
and the play exits the subgame, and then it continues with the outer strategy $\zstrat'(\state)$.
Similarly, if the start state $\state$ is in some subgame $\game_k^i$ then it first
plays $\zstrat_k^i(\state)$ until $\win^{k,i}$ is reached
(which happens eventually almost surely by the definition of
$\zstrat_k^i(\state)$) and the play exits the subgame, 
and then it continues like the strategy $\zstrat'(c_{i,1})$ described above.
Then $\zstrat(\state)$ wins almost surely, since the strategies $\zstrat'$ and
$\zstrat_k^i$ win almost surely.

\smallskip
Towards \cref{lem:BMI-no-SC-plus-F-2}, we show, by induction on $k$,
\new{the following slightly stronger property}.
For each Maximizer strategy~$\zstrat$ in $\game_k$
with a step counter plus a private
finite memory with $\le k$ modes,
for every $\delta >0$ there exists a Minimizer strategy
$\ostrat$ that upper-bounds Maximizer's attainment by $\delta$
regardless of Maximizer's initial memory
mode and the starting time.
Formally,
\begin{equation}\label{eq:stronger-induction-claim}
\forall \delta >0\,\exists\ostrat\,
\forall \memconf\, \forall
t\ \probm_{\game_{k},u,\zstrat[\memconf](t),\ostrat}(\reach{\{c_0\}}) \le \delta
\end{equation}
For the base case $k=1$ we have 
$\game_1 = \game$ from \cref{def:inf-branch-no-Markov}.
Since $k=1$, Maximizer has only one memory mode.
Moreover, Minimizer's strategy from the proof of
\cref{thm:inf-branch-no-Markov}(\ref{thm:inf-branch-no-Markov-2})
works regardless of the starting time $t$ (since it just chooses
delays in the states $b_i$ to
satisfy \eqref{eq:accu-1} and \eqref{eq:accu-2}).
Thus we obtain \eqref{eq:stronger-induction-claim}.

Induction step $k \transition k+1$.
Consider the game $\game_{k+1}$ and a fixed Maximizer 
strategy $\zstrat$ with a step counter plus $(k+1)$ private memory modes
from state $u$.
Let $\{0,1,\dots,k\}$ denote the $k+1$ private memory modes of $\zstrat$.

From every state $c_i$ in $\game_{k+1}$ we enter the subgame $\game_k^i$, at
some state chosen by Minimizer.
When (and if) Maximizer wins this subgame then we are in state $\win^{k,i}$
and the game $\game_{k+1}$ continues with Maximizer's choice at $c_{i,1}$, etc.

Consider the state $c_i$, visited at some time $t$
with some Maximizer memory mode $\memconf$.
Let $\ostrat'$ be some Minimizer strategy.
Then let $\alpha(i,\memconf,t,\ostrat')$ be the probability
that Maximizer will play action
``1'' (w.r.t.\ the encoded concurrent game)
in the next round in $\game_{k+1}$ after winning the subgame $\game_k^i$ (i.e.,
after reaching $\win^{k,i}$), or loses the subgame (never reaches $\win^{k,i}$).
So $\alpha(i,\memconf,t,\ostrat')$ is
the probability of losing the subgame $\game_k^i$ plus
$\sum_j (1/j)\cdot p_j$, where $p_j$
is the probability of winning the subgame and then directly going to $d_{i,j}$
(i.e., in the same round, without seeing any other state $c_i$ before).
\new{To formally capture this notion of the ``same round'', we let
$C \eqdef \{c_i \mid i \in \N\}$ and define}
\begin{align*}
    \alpha(i,\memconf,t,\ostrat') ~\eqdef~& \probm_{\game_{k+1},c_i,\zstrat[\memconf](t),\ostrat'}(\neg\reach{\{\win^{k,i}\}})\\
&\quad + \sum_j (1/j)
\probm_{\game_{k+1},c_i,\zstrat[\memconf](t),\ostrat'}(d_{i,j}\ \before\ C)
\end{align*}
\new{
where $(d_{i,j}\ \before\ C)$ denotes the set of plays that visit state
$d_{i,j}$ before visiting any state in $C$ \emph{again}, i.e., any visit to $C$ at
the start state (here $c_i$) does not count.}
The probability $\alpha(i,\memconf,t,\ostrat')$ depends on $i$ (since we are looking at
state $c_i$), on Maximizer's private memory mode
$\memconf \in \{0,1,\dots,k\}$ at state $c_i$, at the time $t\in \N$ when we are at
$c_i$ and on Minimizer's strategy $\ostrat'$.
Let $\alpha(i,\memconf,t) \eqdef \sup_{\ostrat'} \alpha(i,\memconf,t,\ostrat')$ be the supremum
over all Minimizer strategies.
Let $\alpha(i,t) \eqdef \min_{\memconf \in \{0,\dots,k\}} \alpha(i,\memconf,t)$ the
minimum over all memory modes.
Intuitively, when entering the subgame $\game_k^i$ from $c_i$ at time $t$, for each
Maximizer memory mode $\memconf$, for each $\eps>0$,
Minimizer has a strategy to make the probability
of Maximizer playing action ``1'' after winning the subgame (or else losing
the subgame) at least $\alpha(i,t)-\eps$.
Let $\alpha(i)$ be the maximal accumulation point of the
infinite sequence $\alpha(i,1), \alpha(i,2), \dots$, i.e.,
$\alpha(i) \eqdef \limsup_t \alpha(i,t)$. 
We have:
\begin{align}
\forall\,i\in \nat \ \forall\,\eps > 0\ \forall\,t_0 \in \nat \ \exists\,t\ge t_0 : & \quad \alpha(i,t) \ge \alpha(i)-\eps \label{eq:accum-2}\\
\forall\,i\in \nat \ \forall\,\eps > 0\ \exists\,t_0 \in \nat
\ \forall\,t\ge t_0 : & \quad \alpha(i,t) \le \alpha(i)+\eps\label{eq:accum-1}
\end{align}
Now there are two cases.

In the first case, $\sum_i \alpha(i)$ diverges.
Intuitively, since we have $\alpha(i) = \limsup_t \alpha(i,t)$, Minimizer chooses the
delays at the states $b_i$ in order to make $\alpha(i,t)$ ``large'', i.e.,
close to $\alpha(i)$, by using \eqref{eq:accum-2}.

Analogously to \cref{thm:inf-branch-no-MR-claim-1}, it suffices, for every
$\eps >0$ to construct a Minimizer strategy $\ostrat$ in $\game_{k+1}$ from state $u$
that makes the probability of visiting $c_0$ before revisiting~$u$
(denoted as the event ``$\czerobeforeu$'') at most~$\eps$.

We construct a Minimizer strategy $\ostrat$ that plays as follows.
First $\ostrat$ picks the transition $u \transition b_{i_0}$ for a sufficiently high
$i_0 \in \nat$, to be determined.
At every state $d_{i,j}$, outside of the subgames, $\ostrat$ always plays
action ``0'', i.e., $d_{i,j} \transition r_{i,j}^0$.
This implies that the state $c_i$ is not visited again, unless the state $u$
is re-visited first.
(If Maximizer plays action ``1'' then he loses this round
and the game goes back to $u$.
If Maximizer plays ``0'' then the game goes down to $c_{i-1}$.)

At every state $b_i$ at each time $t'$ the strategy $\ostrat$ picks a delay such that the game arrives
at $c_i$ at a time $t$ such that
\begin{equation}\label{eq:ait-ai}
\alpha(i,t) \ge \alpha(i)- \frac{1}{4} 2^{-i}.
\end{equation}
This is possible by \cref{eq:accum-2}. Let $T_i$ be the set of the times $t$ that
satisfy \eqref{eq:ait-ai}. Minimizer's strategy ensures that the states $c_i$
are only reached at times $t \in T_i$.

From state $c_i$, at each time $t$, for each memory mode $\memconf$ there exists a
Minimizer strategy $\ostrat(i,\memconf,t)$ such that
we have
$\alpha(i,\memconf,t,\ostrat(i,\memconf,t)) \ge \alpha(i,\memconf,t) - \frac{1}{4} 2^{-i}$,
by the definition of
$\alpha(i,\memconf,t) \eqdef \sup_{\ostrat'} \alpha(i,\memconf,t,\ostrat')$.

Since the Maximizer memory mode $\memconf$ is private, Minimizer does not know it.
So our Minimizer strategy $\ostrat$
hedges her bets over all possible $\memconf \in \{0,\dots,k\}$ and plays
each strategy $\ostrat(i,\memconf,t)$ with equal probability $\frac{1}{k+1}$.
It follows that for every $i$
and time $t \in T_i$ chosen according to \cref{eq:ait-ai},
after (and if) winning the subgame $\game_k^i$, Maximizer
plays action ``1'' with probability
\begin{equation}\label{eq:case-one-alpha-large}
    \begin{aligned}
        \alpha(i,\memconf,t,\ostrat) &~\ge~ \frac{1}{k+1} \left(\alpha(i,t) - \frac{1}{4} 2^{-i}\right)
\\
                                     &~\ge~
\frac{1}{k+1} \left(\alpha(i) - \frac{1}{4} 2^{-i} - \frac{1}{4} 2^{-i}\right)
= \frac{1}{k+1} \left(\alpha(i) - \frac{1}{2} 2^{-i}\right).
    \end{aligned}
\end{equation}
Since, outside of subgames, Minimizer always plays action ``0'' and
each state $c_i$ is visited at most once
(unless state $u$ is re-visited),
we get for every starting memory mode
$\memconf'$ and starting time $t' \in T_{i_0}$ that
the probability of visiting $c_0$ before revisiting $u$ is
upper-bounded, i.e.,
\begin{align*}
\probm_{\game_{k+1},c_{i_0},\zstrat[\memconf'](t'),\ostrat}(\czerobeforeu) 
~\le~
\prod_{i=1}^{i_0} (1 - \min_{\memconf}\inf_{t \in T_i} \alpha(i,\memconf,t,\ostrat)).
\end{align*}
Since $\sum_i \alpha(i)$ diverges, it follows
from \eqref{eq:case-one-alpha-large} that
\[
\sum_{i=1}^\infty \min_{\memconf}\inf_{t \in T_i} \alpha(i,\memconf,t,\ostrat)
\ge \left(\frac{1}{k+1}\sum_i \alpha(i)\right) - \frac{1}{2k+2}
\]
also diverges.
From \cref{prop:product-sum} we obtain
\[
\prod_{i=1}^\infty (1 - \min_{\memconf}\inf_{t \in T_i} \alpha(i,\memconf,t,\ostrat)) = 0
\]
and hence
\[
\lim_{i_0 \to \infty} \prod_{i=1}^{i_0} (1 - \min_{\memconf}\inf_{t \in T_i}
\alpha(i,\memconf,t,\ostrat))= 0.
\]
Thus, for every $\eps >0$,
there exists a sufficiently large $i_0$ such that
for all $t' \in T_{i_0}$ and all \new{memory modes} $\memconf'$
\[
\probm_{\game_{k+1},c_{i_0},\zstrat[\memconf'](t'),\ostrat}(\czerobeforeu)
\le \eps.
\]
\new{
Recall that at state $u$ the Minimizer strategy
$\ostrat$ picks the transition $u \transition b_{i_0}$,
i.e., the number $i_0$ is Minimizer's choice.
Moreover, the delay gadget $D_{i_0}$ ensures that state $c_{i_0}$
is entered at some time $t' \in T_{i_0}$, regardless of the time $t$ at state $u$.
}
Hence, for all $t \in \N$ and all \new{memory modes} $\memconf$ we have
\[
\probm_{\game_{k+1},u,\zstrat[\memconf](t),\ostrat}(\czerobeforeu)
\le \eps.
\]
Analogously to \cref{thm:inf-branch-no-MR-claim-1},
we can pick $i_0^l$ and $\eps^l = \delta \cdot 2^{-l}$ after the $l$-th visit to $u$
such that for all $t \in \N$ and all $\memconf$ we have
\[
\probm_{\game_{k+1},u,\zstrat[\memconf](t),\ostrat}(\reach{\{c_0\}}) \le
\sum_{l=1}^\infty \eps_l = \delta
\]
as required.

\bigskip
Now we consider the second case where $\sum_i \alpha(i)$ converges.
We construct a Minimizer strategy $\ostrat$ in $\game_{k+1}$ that plays as follows.
First $\ostrat$ picks the transition $u \transition b_{i_0}$ for a sufficiently high
$i_0 \in \nat$, to be determined.
At every state $d_{i,j}$, outside of the subgames, $\ostrat$ always plays
action ``1'', i.e., $d_{i,j} \transition r_{i,j}^1$.
This implies that each state $c_i$ with $i \ge i_0$ is visited at most once,
and states $c_i$ with $0< i < i_0$ are never visited.
(If Maximizer plays action ``1'' then he immediately wins at $c_0$ and if he plays ``0'' then the
games goes up to $c_{i+1}$.)

By \cref{eq:accum-1}, for every $i \ge 1$ there exists a time $t_i \in \nat$ such that 
\begin{equation}\label{eq:ait-ai-2}
\forall t \ge t_i.\ \alpha(i,t) \le \alpha(i) + \frac{1}{4} 2^{-i}.
\end{equation}
Let $T_i = \{t \in \nat \mid t \ge t_i\}$ be the set of the times $t$ that
satisfy \eqref{eq:ait-ai-2}.
At each state $b_i$ Minimizer's strategy $\ostrat$ delays sufficiently long
such that $c_i$ is reached at a time $t \in T_i$. This is possible for every
arrival time $t'$ at $b_i$.

Consider a state $\state$ in the subgame $\game_k^i$
that is reached at some time $t$ when Maximizer's strategy is in some
memory mode $\memconf$, and let $\ostrat'$ be a Minimizer strategy.
Let $\beta(\state,i,\memconf,t,\ostrat')$ be the probability 
that Maximizer will play action
``1'' (w.r.t.\ the encoded concurrent game)
in the next round in $\game_{k+1}$ after winning the subgame
$\game_k^i$ (i.e., after reaching $\win^{k,i}$), or loses the subgame (never reaches $\win^{k,i}$).
So $\beta(\state,i,\memconf,t,\ostrat')$ is
the probability, from state $\state$, of losing the subgame $\game_k^i$ plus
$\sum_j (1/j)\cdot p_j$, where $p_j$
is the probability of winning the subgame and then directly going to $d_{i,j}$ 
(without visiting any other state $c_i$ in between).
\new{Recall that $C = \{c_i \mid i \in \N\}$.} We let
\begin{align*}
    \beta(\state,i,\memconf,t,\ostrat') ~\eqdef~& \probm_{\game_{k+1},\state,\zstrat[\memconf](t),\ostrat'}(\neg\reach{\{\win^{k,i}\}})\\
&\quad + \sum_j (1/j)
\probm_{\game_{k+1},\state,\zstrat[\memconf](t),\ostrat'}(d_{i,j}\ \before\ C)
\end{align*}
Let $\beta(\state,i,\memconf,t) \eqdef \sup_{\ostrat'} \beta(\state,i,\memconf,t,\ostrat')$ be the supremum
over all Minimizer strategies.
Let
\[
\beta(\state,i,t) \eqdef \min_{\memconf \in \{0,\dots,k\}} \beta(\state,i,\memconf,t)
\]
be the minimum over all memory modes.

\begin{claim}\label{lem:BMI-no-SC-plus-F-claim-1}
For all states $\state$ in $\game_k^i$ and times $t$ we have $\beta(\state,i,t+1) \le (k+1)\alpha(i,t)$.
\end{claim}
\begin{proof}
Consider the situation where we are at state $c_i$ at time $t$ and
the Maximizer strategy $\zstrat$ is in some memory mode $\memconf$ (unknown to
Minimizer). Then some particular Minimizer strategy $\hat{\ostrat}$ could play
$c_i \transition \state$ to arrive at $\state$ in one
step at time $t+1$. 
Meanwhile, $\zstrat$ can update its memory to some other mode $\memconf'$ (or
a distribution over memory modes), still unknown to Minimizer.
Then $\hat{\ostrat}$ hedges her bets by guessing Maximizer's memory mode
$\memconf'$.
For each of the $k+1$ possible modes $\memconf'$, the Minimizer
strategy $\hat{\ostrat}$ plays, with probability
$\frac{1}{k+1}$, an $\eps$-optimal strategy to maximize the probability that
Maximizer plays action ``1'' after winning $\game_k^i$ (or loses the subgame).
Thus 
\begin{align*}\alpha(i,t) & = \min_\memconf\sup_{\ostrat'} \alpha(i,\memconf,t,\ostrat')\\
&\ge
\min_\memconf \alpha(i,\memconf,t,\hat{\ostrat})\\
&\ge
\min_{\memconf'} \frac{1}{k+1} \beta(\state,i,\memconf',t+1,\hat{\ostrat})\\
&\ge 
\frac{1}{k+1} (\min_{\memconf'} \sup_{\ostrat'} \beta(\state,i,\memconf',t+1,\ostrat') - \eps)\\
&=
\frac{1}{k+1} (\beta(\state,i,t+1) - \eps).
\end{align*}
Since this holds for every $\eps >0$, the claim follows.
\end{proof}

By \cref{lem:BMI-no-SC-plus-F-claim-1}, for every $\state$
in $\game_k^i$ and time $t+1$ there exists at least
one memory mode $\memconf(\state,t+1)$
(a mode $\memconf$ where the minimum $\beta(\state,i,t+1) = \min_{\memconf \in \{0,\dots,k\}} \beta(\state,i,\memconf,t+1)$
is realized) such that
if $\zstrat$ enters $\state$ at time $t+1$ in mode $\memconf(\state,t+1)$
then after winning $\game_k^i$ (if at all) Maximizer plays action ``1''
with a ``small'' probability $\le (k+1)\alpha(i,t)$.
Crucially, this property holds for the $\sup$ over the Minimizer strategies and thus for
\emph{every} Minimizer strategy inside the subgame $\game_k^i$.
In particular it holds for the
Minimizer strategy $\ostrat$ that we will construct.
We call $\memconf(\state,t+1)$ the \emph{forbidden} memory mode
for state $\state$ at time $t+1$.

Above we have defined our Minimizer strategy $\ostrat$ so that it adds sufficient
delays in the states $b_i$ such that $c_i$ is only visited at times $t \ge t_i$.
This implies that states $\state$ in $\game_k^i$ are only visited at times
$t+1$ where $t \ge t_i$.
Since for these times \cref{eq:ait-ai-2} is satisfied, we obtain
\begin{equation}\label{eq:b_is_small}
\forall t \ge t_i.\ \beta(\state,i,\memconf(\state,t+1),t+1) \le (k+1)(\alpha(i) + \frac{1}{4}2^{-i}).
\end{equation}

Let $\zstrat'$ be a restriction of $\zstrat$ that, inside the subgame
$\game_k^i$, is never in the forbidden memory mode $\memconf(\state,t+1)$ at state $\state$
at time $t+1$, or else concedes defeat.

\begin{claim}\label{claim:avoiding-forbidden-less-memory}
Consider a step counter plus $(k+1)$ mode Maximizer strategy $\zstrat'$ in $\game_k^i$
that is never in the forbidden memory mode $\memconf(\state,t+1)$ at state $\state$
at time $t+1$.
Then there exists a step counter plus $k$ mode Maximizer strategy $\zstrat''$ in
$\game_k^i$ that performs equally well as $\zstrat'$ against any Minimizer
strategy.
\end{claim}
\begin{proof}
\new{
The strategy $\zstrat'$ has $k+1$ memory modes $\{0,\dots,k\}$, plus the step counter.
We will construct the strategy $\zstrat''$ to only have $k$ 
memory modes $\{0,\dots,k-1\}$, plus the step counter.
The strategy $\zstrat''$ can directly imitate the behavior of $\zstrat'$ as
follows.
Suppose that $\zstrat'$ enters memory mode $k$ at some state $\state$ and
time $t+1$. From our assumption that $\zstrat'$ never enters the forbidden memory
mode it follows that $k \neq \memconf(\state,t+1)$.
In this situation $\zstrat''$ enters memory mode $\memconf(\state,t+1)$ instead.
Whenever $\zstrat''$ is in memory mode $\memconf(\state,t+1)$ at some state $\state$
and time $t+1$ then it plays like $\zstrat'$ at state $\state$ in memory mode
$k$. By the condition on the behavior of $\zstrat'$ there is no confusion,
$\zstrat''$ just uses the memory modes $\{0,\dots,k-1\}$ and it still imitates the
behavior of $\zstrat$.}
\end{proof}

By \cref{claim:avoiding-forbidden-less-memory}, the Maximizer strategy
$\zstrat'$ is equivalent to a strategy with just a step counter and $k$ memory
modes.
By the induction hypothesis \eqref{eq:stronger-induction-claim},
for this restricted Maximizer strategy $\zstrat'$
there exists a Minimizer strategy $\ostrat_i$ in $\game_k^i$ such that
\begin{equation}\label{eq:k-mode-strat}
\forall \memconf\ \forall t\ \probm_{\game_k^i,u^{k,i},\zstrat'[\memconf](t),\ostrat_i}(\reach{\{c_0^{k,i}\}}) \le \delta \cdot 2^{-(i+1)}.
\end{equation}
We are now ready to construct Minimizer's strategy $\ostrat$ in $\game_{k+1}$.
At every state $d_{i,j}$, outside of the subgames, $\ostrat$ always plays
action ``1'', i.e., $d_{i,j} \transition r_{i,j}^1$.
This implies that each state $c_i$ is visited at most once.
At the states $b_i$, Minimizer chooses the delays
such that $c_i$ is reached at a time $t \ge t_i$, as described above.
From each $c_i$ Minimizer goes to state $u^{k,i}$ of the subgame $\game_k^i$.
Inside each subgame $\game_k^i$ Minimizer plays like $\ostrat_i$.
By \eqref{eq:k-mode-strat}, $\ostrat_i$
performs well in $\game_k^i$ (regardless of the initial memory mode and time)
if Maximizer limits himself to $\zstrat'$.

Now we show that $\ostrat$ performs well in $\game_{k+1}$.
Since $\ostrat$ first picks the transition $u \transition b_{i_0}$
and then always plays action ``1'' outside of the subgames,
it follows that each subgame $\game_k^i$ with $i \ge i_0$ is played
at most once, and subgames $\game_k^i$ with $i < i_0$ are never played.
For each subgame $\game_k^i$, let $\mathrm{Forb}_i$ be the set of plays where Maximizer
enters a forbidden memory mode (for the current state and time) at least once.
\footnote{Strictly speaking, $\mathrm{Forb}_i$ is not an event in $\game_k^i$,
since it refers to the memory mode of Maximizer's strategy $\zstrat$.
However, since we fix $\zstrat$ first,
we can consider the MDP that is induced by fixing $\zstrat$ in
$\game_k^i$. Then Maximizer's memory mode $\memconf$ and the step counter $t$
are encoded into the states, which are of the form $(\state,\memconf,t)$.
In this MDP, $\mathrm{Forb}_i$ is a measurable event, actually an open set.
However, since Maximizer's memory is private, Minimizer has only partial
observation in this MDP, i.e., she cannot distinguish between states $(\state,\memconf,t)$ and
$(\state,\memconf',t)$.
Indeed the Minimizer strategy that we construct does not assume any knowledge
of the memory mode, but instead hedges her bets.
}

From Equation~\eqref{eq:k-mode-strat} we obtain that Maximizer loses $\game_k^i$ (and
thus $\game_{k+1}$) with high probability if he \emph{never} enters a forbidden
memory mode.
\begin{equation}\label{eq:max-good}
\begin{aligned}
\max_{\memconf}\sup_{t} &\; \probm_{\game_{k+1},c_i,\zstrat[\memconf](t),\ostrat}(\reach{\{c_0\}}  \cap \overline{\mathrm{Forb}_i})\\
&\le 
\max_{\memconf}\sup_{t} \probm_{\game_k^i,u^{k,i},\zstrat'[\memconf](t),\ostrat_i}(\reach{\{c_0^{k,i}\}})\\
& \le
\delta \cdot 2^{-(i+1)}.
\end{aligned}
\end{equation}
On the other hand, we can show that if Maximizer does enter a forbidden memory mode
(for the current state and time) in
$\game_k^i$ then his chance of playing action ``1'' (and thus winning
$\game_{k+1}$ in that round) after (and if) winning the subgame $\game_k^i$ is
small.
This holds for \emph{every} Minimizer's strategy inside $\game_k^i$ and thus
in particular this holds for our chosen Minimizer strategy $\ostrat_i$.

Recall that $\ostrat$ ensures that states in $\game_k^i$ are only reached at
times $t+1$ where $t \ge t_i$, and thus \cref{eq:b_is_small} applies.
Hence, at $c_i$, Maximizer's chance of satisfying $\mathrm{Forb}_i$ and still winning the game
in this round (without going to $c_{i+1}$ and the next subgame $G_k^{i+1}$)
is upper bounded by $(k+1)(\alpha(i) + \frac{1}{4}2^{-i})$.
For all $\memconf$ and all $t \ge t_i$ we have
\begin{equation}\label{eq:use-bad}
\begin{aligned}
\probm_{\game_{k+1},c_i,\zstrat[\memconf](t),\ostrat}(\reach{\{c_0\}} \cap \mathrm{Forb}_i
\cap \neg\reach{\{c_{i+1}\}})\\
 \le (k+1)(\alpha(i) + \frac{1}{4}2^{-i})
\end{aligned}
\end{equation}
Since in our current case $\sum_i \alpha(i)$ converges, it follows that
$\sum_i (k+1)(\alpha(i) + \frac{1}{4}2^{-i})$ also converges, and thus
there exists a sufficiently large
$i_0 \in \nat$ such that
\begin{equation}\label{eq:bound-sum-ai}
\sum_{i \ge i_0} (k+1)(\alpha(i) + \frac{1}{4}2^{-i}) \le \delta/2
\end{equation}
Let
\[
{\it nFo}(i,\memconf',t') ~\eqdef~ \probm_{\game_{k+1},c_i,\zstrat[\memconf'](t'),\ostrat}(\reach{\{c_0\}} \cap \overline{\mathrm{Forb}_i})
\]
and
\[
{\it Fo}(i,\memconf',t') ~\eqdef~
\probm_{\game_{k+1},c_i,\zstrat[\memconf'](t'),\ostrat}(\reach{\{c_0\}} \cap \mathrm{Forb}_i \cap \neg\reach{\{c_{i+1}\}})
\]
Then from \eqref{eq:max-good}, \eqref{eq:use-bad} and \eqref{eq:bound-sum-ai}
we obtain that for every initial memory mode $\memconf$ and time $t$,
\begin{align*}
\probm_{\game_{k+1},u,\zstrat[\memconf](t),\ostrat}(\reach{\{c_0\}}
&\le
\sum_{i \ge i_0}
\max_{\memconf'}\sup_{t'\ge t_i} {\it nFo}(i,\memconf',t')
+
 \sum_{i \ge i_0}
 \max_{\memconf'}\sup_{t' \ge t_i} {\it Fo}(i,\memconf',t')
\\
&\le
\sum_{i \ge i_0} \delta \cdot 2^{-(i+1)} +
\sum_{i \ge i_0} (k+1)(\alpha(i) + \frac{1}{4}2^{-i})\\
&\le
\delta/2 
+ \delta/2
=
\delta
\
\end{align*}
\end{proof}

In order to show that Maximizer needs infinite memory, in addition to a step
counter, we combine all the nested games $\game_k$ into a single game.

\begin{definition}\label{def:combine-nested}
For all $k \ge 1$ consider the nested games $\game_k$ from \cref{def:nested-BMI}
with initial state $u^k$ and target state $c_0^k$, respectively.
We construct a game $\game$ with initial state $\state_0$,
target state $f$, 
Minimizer-controlled transitions $\state_0 \transition u_k$ for all $k$,
and Maximizer controlled transitions $c_0^k \transition f$ for all $k$.
The objective in $\game$ is $\reach{\{f\}}$.
\end{definition}

\new{The following theorem is the formal version of \cref{thm:story-main} from the introduction.}

\begin{theorem}\label{thm:no-sc-plus-finite}
Let $\game$ be the infinitely branching turn-based reachability game from \cref{def:combine-nested}.
\begin{enumerate}
\item\label{thm:no-sc-plus-finite-1}
All states in $\game$ are almost surely winning. I.e., for every state $\state$ there
exists a Maximizer strategy $\zstrat$ such that
$\inf_{\ostrat}\probm_{\game,\state,\zstrat,\ostrat}(\reach{\{f\}})=1$.
\item\label{thm:no-sc-plus-finite-2}
For each Maximizer strategy~$\zstrat$ with a step counter plus a private
finite memory we have
\[
\inf_{\ostrat}\probm_{\game,\state_0,\zstrat,\ostrat}(\reach{\{f\}})=0.
\]
I.e., for any $\eps < 1$ there does not exist any $\eps$-optimal step counter plus
finite private memory Maximizer strategy $\zstrat$ from state $\state_0$ in $\game$. 
\end{enumerate}
\end{theorem}
\begin{proof}
Towards \cref{thm:no-sc-plus-finite-1}, every state in $\game_k$ is almost
surely winning by \cref{lem:BMI-no-SC-plus-F}(\ref{lem:BMI-no-SC-plus-F-1}).
Thus, after the first step $\state_0 \transition u^k$ into some game $\game_k$, Maximizer just needs to play the
respective almost surely winning strategy in $\game_k$.

Towards \cref{thm:no-sc-plus-finite-2}, consider a Maximizer strategy
$\zstrat$ with a step counter and a finite memory with some number of
modes $k \in \nat$.
Then, by \cref{lem:BMI-no-SC-plus-F}(\ref{lem:BMI-no-SC-plus-F-2}),
for every $\delta >0$, Minimizer can choose a first step
$\state_0 \transition u^k$ into $\game_k$ 
and a strategy in $\game_k$ that upper-bounds Maximizer's attainment to $\le \delta$.
\end{proof}

\begin{remark}\label{rem:liminf}
  \rm
\cref{thm:no-sc-plus-finite} has implications
even for games with finite action sets (resp.\ finitely branching turn-based
games).

Consider a finitely branching game where the states are labeled with rewards in $\{-1,1\}$.
The objective of Maximizer is to ensure that the $\liminf$ of the seen rewards
is $\ge 0$. (Equivalently that states with reward $-1$ are visited only
finitely often. This is also called a co-B\"uchi objective in \cite{KMSW2017}).
The infinitely branching reachability game of \cref{thm:no-sc-plus-finite}
can be encoded into a finitely branching $\liminf$ game, and thus the lower
bound of \cref{thm:no-sc-plus-finite} carries over.
One just replaces every infinite Minimizer branching $\state \to \state_i$ for
$i \in \N$ by a Minimizer-controlled gadget $\state \to \state'_1 \to
\state'_2 \dots$ and $\state'_i \to \state_i$ with new states $\state'_i$ that
have reward $1$.
Minimizer cannot stay in states $\state'_i$ forever, since
their rewards of $1$ makes this winning for Maximizer. Thus the finitely
branching gadgets faithfully encode Minimizer's original infinitely branching
choice.
Finally, the target state $c_0$ is given reward $1$ and a self-loop,
and all other states are given reward $-1$.
Thus the $\liminf \ge 0$ objective in the new game corresponds to the reachability
objective to reach state $c_0$ in the original game.
The only problem with this construction is that the new gadgets incur extra
steps, i.e., the step counters in the two games do not coincide.
Thus the property that a step counter does not help Maximizer does not
follow immediately from \cref{thm:no-sc-plus-finite} if taken as a black box.
However, the delay gadgets $D_i$ in \cref{def:nested-BMI} (and their
finitely branching encoding in the new game)
still ensure that the step counter does not help Maximizer.
I.e., the proof of the lower bound for the finitely branching $\liminf \ge 0$ game
is nearly identical to the proof of \cref{thm:no-sc-plus-finite}.

In the new finitely branching game above, the objective to reach $c_0$ 
also coincides with the objective to attain a high \emph{expected}
$\liminf$ of the rewards. Thus $\eps$-optimal (resp.\ optimal) strategies to maximize
the expected $\liminf$ also require infinite memory.

Finally, if one flips the signs of the rewards of all transitions in the
game above, then the objective to reach $c_0$ coincides with the
objective to minimize the expected $\limsup$ of the rewards.
Thus $\eps$-optimal (resp.\ optimal) strategies to minimize
the expected $\limsup$ also require infinite memory.
This solves the open question in Section 5 of \cite{Secchi:1998}.
\end{remark}
 
\section{Infinitely Branching but only Finitely Often}\label{sec:restrictedmin}
\new{
In \cref{thm:no-sc-plus-finite}
we showed that
already for turn-based reachability games,
$\eps$-optimal strategies for Maximizer require infinite memory.
The lower bound construction crucially uses that 
Minimizer's action set is infinite.
If Minimizer has only finite action sets then Maximizer's $\eps$-optimal strategies
can be simpler, namely MR for concurrent games and MD for turn-based games (cf.~\cref{maxtable}).
}

However, the connection between Minimizer's infinite action sets and
Maximizer's need for infinite memory is not as direct as it might seem.
In this section we consider the restricted setting 
of concurrent games where \emph{Maximizer has finite action sets
and Minimizer can use an infinite action set at most finitely often in any play}
(see also \cref{def:inf-branch-finite} below).

We show in \cref{thm:infinity-bound-upper} that in this case, Maximizer still has uniformly $\eps$-optimal 1-bit strategies.
Moreover, we show that this upper bound is tight in the sense that
even if Minimizer can use an infinite action set \emph{only once},
$\eps$-optimal Maximizer strategies still require $1$ bit of memory.

We start with the lower bound.
The following theorem shows that, even in turn-based reachability games,
if Minimizer can use infinite branching just once,
MR strategies cannot be $\eps$-optimal
for Maximizer for any $\eps < 1/2$.

\begin{definition}\label{def:def:turn-big-match-u}
Consider the
\new{
Turn-based Big Match on $\N$
(\cref{def:turn-big-match} on page~\pageref{def:turn-big-match}),
}
and add a new
infinitely branching Minimizer-controlled initial state $u$ and Minimizer-transitions
$u \transition c_x$ for all $x \in \N$.
\end{definition}

\begin{theorem}\label{thm:infinity-bound-lower}
There exists a turn-based game $\game$ as in \cref{def:def:turn-big-match-u}
with initial state $u$
and reachability objective $\reach{\{c_0\}}$
such that
\begin{enumerate}
\item
  All states except $u$ are finitely branching,
  and all plays from $u$ use infinite branching exactly once.
\item
$\valueof{\game}{u} \ge 1/2$.
\item
$\inf_{\ostrat}\probm_{\game,u,\zstrat,\ostrat}(\reach{\{c_0\}})=0$
holds for every Maximizer MR strategy~$\zstrat$\new{.}
\end{enumerate}
\end{theorem}
\begin{proof}
Item 1 holds by the construction in \cref{def:def:turn-big-match-u}, since $u$
is the only state with infinite branching and plays cannot return to $u$.
  
Towards Item 2., \cref{thm:TB-BMI-zplus}(Item~\ref{thm:TB-BMI-zplus-2})
yields $\valueof{\game}{c_x} \ge 1/2$ for every $x \in \N$, and thus $\valueof{\game}{u} \ge 1/2$.

Towards Item 3., by
\cref{thm:TB-BMI-zplus}(Item~\ref{thm:TB-BMI-zplus-3}),
for every MR Maximizer strategy $\zstrat$
we have
$\limsup_{x \to \infty}\inf_{\ostrat}\probm_{\game,c_x,\zstrat,\ostrat}(\reach{\{c_0\}})=0$.
Since in our game the Minimizer strategy $\ostrat$ gets to pick the transition
$u \to c_x$ for an arbitrary $x \in \N$, 
$\inf_{\ostrat}\probm_{\game,u,\zstrat,\ostrat}(\reach{\{c_0\}})=0$.
\end{proof}

Towards the upper bound, we first define the case where Minimizer can use
infinite action sets only finitely often in any play.

\begin{definition}\label{def:inf-branch-finite}
Let $\game$ be a concurrent game on a countable set of states
$\states$ and
$\states^\infty \eqdef \{\state \in \states \mid \card{B(\state)}=\infty\}$
the states with an infinite Minimizer action set.
Let $\states' \subseteq \states$ be the subset of states
$\state$ such that every play from $\state$ (under any strategies)
visits $\states^\infty$ only finitely often.
\end{definition}

Different plays from the same start state can have
different numbers of visits to $\states^\infty$.
Even if this number is finite for every play, there is no uniform finite upper bound.
Thus the condition of \cref{def:inf-branch-finite}
on plays does not imply a finite bound for the start state.
However, we show that an ordinal bound exists.

\new{
We introduce a
ranking function $I: \states' \to \ord$
so that $I(x)$ is an upper bound on the number of possible visits
to $\states^\infty$, including the current state $x$.
This is based on a classic result on well-founded relations.
Recall that a binary relation $E\subseteq \states\x\states$ is \emph{well-founded} if every non-empty subset $X\subseteq\states$ has a minimal element w.r.t.~$E$.
}
\begin{theorem}[{\cite{Jech:2002} Theorem 2.27}]\label{thm:Jech}
  If $E\subseteq \states\x\states$ is well-founded then there exists
  a unique function $\rho: V \to \ord$ such that for all $x \in V$
  \[
    \rho(x) = \sup\{\rho(y)+1 \mid y E x\}.
  \]
  In particular, $y E x$ implies $\rho(y) < \rho(x)$.
  Moreover, if $V$ is countable then $\sup \rho(V)$ is a countable ordinal.
\end{theorem}

\begin{definition}[Ranking function $I$]\label{def:ordinal-index}
\new{
Let $\game$ be a concurrent reachability game on a countable set of states
$\states$, and let $\states', \states^\infty \subseteq \states$ be as in \cref{def:inf-branch-finite}.
Let $\mathord{\to} \subseteq \states \times \states$ be the induced game graph,
i.e., $x \mathord{\to} y \iff \exists a,b.\, y \in \support(p(x,a,b))$.
}

\new{
Let $V \eqdef \states' \cap \states^\infty$ and
$E \subseteq V \times V$ be the reversal of the closure of $\to$ over states in
$\states' \setminus \states^\infty$, i.e.,
$(y,x) \in E \iff \exists k\ge 0, z_1,\dots,z_k \in \states'
\setminus \states^\infty\ : x \to z_1 \to \dots \to z_k \to y$.
}

\new{
From the definition of $\states'$ we obtain that $E$ is a well-founded
relation on $V$.
By \cref{thm:Jech}, there exists a unique function
$\rho: V \to \ord$ such that for all $x \in V$
\begin{equation}\label{eq:rho}
\rho(x) = \sup\{\rho(y)+1 \mid y E x\}.
\end{equation}
By convention, $\sup \emptyset =0$.
We first define our ranking function $I: V \to \ord$
only on the set $V$ by $I(x) \eqdef \rho(x)+1$.
Intuitively, $I(x)$ is the upper bound on the number of visits
to $\states^\infty$, including the current state $x$. 
We then extend the function $I$ from $V$ to $\states'$ as follows.
For every $x \in \states' \setminus V$ let
\[
I(x) \eqdef \sup\{\rho(y) \mid y \in V \wedge x \to^+ y\},
\]
where $\to^+$ is the transitive closure of $\to$.
Since we assume $\sup \emptyset =0$, 
the states $x$ that cannot reach
$\states^\infty$ satisfy $I(x)=0$.
}
\end{definition}

\begin{lemma}\label{lem:ordinal-index}
\new{
The ranking function $I:\states'\to \ord$ satisfies the following properties.
}
\begin{equation}\label{eq:index-noninc}
x \mathord{\to} y \ \mbox{implies}\ I(x) \ge I(y)
\end{equation}  
\begin{equation}\label{eq:index-dec}
  x \in \states^\infty\cap \states' \ \wedge\ x \mathord{\to} y \ \mbox{implies}\ I(x) > I(y)
\end{equation}
and $\gamma(\game) \eqdef \sup I(\states')$ is a countable ordinal.
\end{lemma}
\begin{proof}
\new{
\cref{eq:index-noninc} and \cref{eq:index-dec} follow 
directly from the definition of function $I$ in \Cref{def:ordinal-index}.
}
Since $\states$ and $\states'$ are countable, $\sup I(\states')$ is a
countable ordinal by \cref{thm:Jech}.
\end{proof}

It follows from \cref{lem:ordinal-index} that
states in $\states'$ can be part of cycles,
but not part of any cycle that contains a state from $\states^\infty$.
E.g., in the game in \cref{fig:TheGame-inf} we have $\states' = \{c_0\}$,
i.e., none of the states are in $\states'$ except for the target.

Now we show that 1 bit of public memory is sufficient for Maximizer,
provided that Minimizer can use infinite action sets only finitely often in
any play.
I.e., for every $\eps >0$,
Maximizer has a public 1-bit strategy for reachability that is uniformly
$\eps$-optimal from $\states'$.

First we need a slight generalization of the reachability objective.

\begin{definition}[Weighted reachability]\label{def:weighted-reach}
Let $\game$ be a concurrent game on $\states$ and $\reachset \subseteq \states$.
Let $f: \reachset \to [0,1]$ be a reward function. We lift $f$ to plays
$f: Z^\omega \to [0,1]$ as follows. If a play $h \in Z^\omega$
never visits $\reachset$ then $f(h) \eqdef 0$. Otherwise,
let $f(h) \eqdef f(\state)$ where $\state$ is the first state in $\reachset$
that is visited by $h$.
Let $\weighted{f}$ denote the weighted reachability objective, i.e., to
maximize the expected payoff w.r.t.\ function $f$.
\end{definition}

Weighted reachability generalizes reachability (just let $f(\state)=1$ for all
$\state \in \reachset$).
Now we generalize \cref{thm:conc-reach-uniform} to weighted reachability.

\begin{restatable}{theorem}{weightedconcreachuniform}\label{thm:weighted-conc-reach-uniform}
For any concurrent game with finite action sets and weighted reachability objective, for any $\eps>0$,
Maximizer has a uniformly $\eps$-optimal public 1-bit strategy.
If the game is turn-based and finitely branching, Maximizer has a deterministic such strategy.
\end{restatable}
\begin{proof}
We can encode weighted reachability into ordinary reachability.
Given a concurrent game $\game$ with target set $\reachset$ and weighted
reachability objective $\weighted{f}$, we construct a modified game $\game'$
with target set $\{t\}$, where $t$ is a new state, as follows.
From every state $\state \in \reachset$, regardless of the chosen actions,
the game goes to $t$ with probability $f(\state)$ and to a special
new sink state $\bot$ with probability $1-f(\state)$.
Then $\weighted{f}$ in $\game$ coincides with $\reach{\{t\}}$ in
$\game'$, i.e.,
$
\expectval_{\game,\state,\zstrat,\ostrat}(f)=
\probm_{\game',\state,\zstrat,\ostrat}(\reach{\{t\}})$.
The result follows from \cref{thm:conc-reach-uniform}, since the
1-bit strategy can be carried from $\game'$ to $\game$.
\end{proof}

\begin{theorem}\label{thm:infinity-bound-upper}
Let $\game$ be a concurrent game \new{with finite Maximizer action sets}
on a countable set of states $\states$ 
with reachability objective $\reach{\reachset}$
and $\states' \subseteq \states$ as in \cref{def:inf-branch-finite}.

For every $\eps >0$, Maximizer has a public 1-bit strategy that is uniformly
$\eps$-optimal from every state in $\states'$.
If the game is turn-based then Maximizer has a deterministic such strategy.
\end{theorem}
\begin{proof}
  Let $I: \states' \to \ord$ be the \new{ranking} function from
  \new{\cref{def:ordinal-index}}.
For every ordinal $\alpha \in \ord$ let
$\states_\alpha \eqdef \{\state \in \states' \mid I(\state)=\alpha\}$.
We have $\states' = \bigcup_{\alpha \le \gamma(\game)} \states_\alpha$
for the countable ordinal $\gamma(\game)$ by \cref{lem:ordinal-index}.
Let
$\states_{< \alpha} \eqdef \bigcup_{\beta < \alpha} \states_\beta$
and
$\states_{\le \alpha} \eqdef \bigcup_{\beta \le \alpha} \states_\beta$.
We can assume without restriction that
the states in $\reachset$ are absorbing and thus $\reachset \subseteq \states_0$.

Since $\gamma(\game)$ is a countable ordinal, 
the set $\{\alpha \in \ord \mid \alpha \le \gamma(\game)\}$ is countable
and thus we can pick an injection
$g: \{\alpha \in \ord \mid \alpha \le \gamma(\game)\} \to \N$.
Let $\eps_\alpha \eqdef \eps \cdot 2^{-g(\alpha)}$ for every $\alpha \le \gamma(\game)$.

For every ordinal $\alpha \le \gamma(\game)$ we consider a restricted subgame
$\game_\alpha$ of $\game$ that is played on the subspace $\states_{\le \alpha}$.
\new{
The objective of $\game_\alpha$ is a weighted reachability objective, defined
relative to a reward function $f_\alpha$ like in \cref{def:weighted-reach}.
Let $\reachset_\alpha \eqdef \states_{< \alpha} \cup \reachset$ be a target
set. We consider the weighted reachability objective
$\weighted{f_\alpha}$ where $f_\alpha: \reachset_\alpha \to [0,1]$
with $f_\alpha(\state) \eqdef \valueof{\game,\reach{\reachset}}{\state}$.
}

For every $\alpha \le \gamma(\game)$ and 
$\state \in \states_\alpha$ we show that
\begin{equation}\label{eq:val-reach-equal-val-weight}
\valueof{\game_\alpha,\weighted{f_\alpha}}{\state} = \valueof{\game,\reach{\reachset}}{\state}
\end{equation}
If $\alpha=0$ then the equality \eqref{eq:val-reach-equal-val-weight} holds trivially, since $\reachset_0 = \reachset$
and $f_0(\state)=1$ for every $\state \in \reachset$.

Now we consider the case of ${\alpha>0}$.
For the $\le$ inequality of \eqref{eq:val-reach-equal-val-weight}, first assume towards a contradiction that
$\valueof{\game_\alpha,\weighted{f_\alpha}}{\state} >
\valueof{\game,\reach{\reachset}}{\state}$
for some state $\state \in \states_\alpha$.
Let \[\eps' \eqdef (\valueof{\game_\alpha,\weighted{f_\alpha}}{\state} -
\valueof{\game,\reach{\reachset}}{\state})/3 > 0\]
and $\zstrat$ an $\eps'$-optimal Maximizer strategy from $\state$ for $\weighted{f_\alpha}$
in $\game_\alpha$. We construct a Maximizer strategy $\zstrat'$ in $\game$
from $\state$ as follows. 
Initially, $\zstrat'$ plays like $\zstrat$.
Then upon reaching some state $\state'$ in \new{$\reachset_\alpha$} it
switches to an $\eps'$-optimal strategy for $\reach{\reachset}$ from $\state'$.
Hence, we get that $\probm_{\game,\state,\zstrat',\ostrat}(\reach{\reachset}) \ge
{\valueof{\game_\alpha,\weighted{f_\alpha}}{\state} - 2\eps'}
> \valueof{\game,\reach{\reachset}}{\state}$, a contradiction.
Therefore $\valueof{\game_\alpha,\weighted{f_\alpha}}{\state} \le \valueof{\game,\reach{\reachset}}{\state}$.

Towards the $\ge$ inequality of \eqref{eq:val-reach-equal-val-weight}, consider an
$\eps'$-optimal strategy $\zstrat$ from $\state$ for
$\reach{\reachset}$ in $\game$ and apply it in $\game_\alpha$.
For any Minimizer strategy $\ostrat$ from $\state$,
\new{let $\playset^{\state,\zstrat,\ostrat}$ be the set of plays from $\state$
consistent with $\zstrat,\ostrat$. 
Let 
$\playset^{\state,\zstrat,\ostrat}_{\state'} \subseteq \playset^{\state,\zstrat,\ostrat}$
be the subset of plays
where $\state'$ is the first visited state with $I(\state') < \alpha$.
Since $\state \in \states_\alpha$ and
$\alpha >0$ but $\reachset \subseteq \states_0$,
every play from $\state$ that reaches $\reachset$ must first visit some state
$\state' \in \states_{< \alpha}$.
Thus these subsets $\playset^{\state,\zstrat,\ostrat}_{\state'}$
are a disjoint partition of $\playset^{\state,\zstrat,\ostrat}$, i.e.,
\begin{equation}\label{eq:infinity-bound-upper-partition}
\playset^{\state,\zstrat,\ostrat} = \biguplus_{\state' \in \states_{< \alpha}} \playset^{\state,\zstrat,\ostrat}_{\state'}
\end{equation}
}
Then
\new{
\begin{align*}
& \valueof{\game_\alpha,\weighted{f_\alpha}}{\state}
\\  
& \ge
       \inf_\ostrat \expectval_{\game_\alpha,\state,\zstrat,\ostrat}(f_\alpha) &
\mbox{def.\ of value} 
\\
& =
\inf_\ostrat \sum_{\state' \in \states_{< \alpha}}
\probm_{\game,\state,\zstrat,\ostrat}(\playset^{\state,\zstrat,\ostrat}_{\state'})\cdot f_\alpha(\state')
& \mbox{by \eqref{eq:infinity-bound-upper-partition}}
\\
& =
\inf_\ostrat \sum_{\state' \in \states_{< \alpha}}
\probm_{\game,\state,\zstrat,\ostrat}(\playset^{\state,\zstrat,\ostrat}_{\state'})\cdot \valueof{\game,\reach{\reachset}}{\state'}
& \mbox{def.\ of $f_\alpha$}
\\
  &\ge
\inf_\ostrat \sum_{\state' \in \states_{< \alpha}}
\probm_{\game,\state,\zstrat,\ostrat}(\playset^{\state,\zstrat,\ostrat}_{\state'} \cap \reach{\reachset})
& \mbox{$\ostrat$ can restart at $\state'$}
  \\
&=
\inf_\ostrat \probm_{\game,\state,\zstrat,\ostrat}(\reach{\reachset})
& \mbox{by \eqref{eq:infinity-bound-upper-partition}} 
  \\
&\ge
\valueof{\game,\reach{\reachset}}{\state} - \eps'.
& \mbox{def.\ of $\zstrat$}       
\end{align*}
}
Since the above holds for every $\eps'>0$, it follows that
$\valueof{\game_\alpha,\weighted{f_\alpha}}{\state} \ge \valueof{\game,\reach{\reachset}}{\state}$
and we obtain \eqref{eq:val-reach-equal-val-weight}.

\smallskip
We now define Maximizer's public 1-bit strategy $\zstrat$ on $\states'$ in $\game$.
\new{
It uses two memory modes $\{0,1\}$ and $\zstrat[\memconf]$ denotes
$\zstrat$ with current memory mode $\memconf$.
\verynew{The strategy} $\zstrat$ starts in memory mode $0$, i.e., $\zstrat = \zstrat[0]$
(cf. ``Memory-based Strategies'' in \Cref{sec:prelim}).
}

\new{
First we consider a slightly modified weighted reachability 
objective
$\weighted{f'_\alpha}$ on $\game_\alpha$ where 
$f'_\alpha : \reachset_\alpha \cup \states^\infty \to [0,1]$
and $f'_\alpha(\state) \eqdef \valueof{\game,\reach{\reachset}}{\state}$.
This game effectively ends
when a state in $\states^\infty$ (with an infinite Minimizer action set)
is visited, unlike for the
$\weighted{f_\alpha}$ objective where the game only ends in the following step when it
inevitably (by definition of the ranking function) visits a state in $\states_{< \alpha}$.
Thus effectively the game $\game_\alpha$ with objective $\weighted{f'_\alpha}$
has only finite action sets, since it stops before infinite action sets can be
used.
Therefore, by
\cref{thm:weighted-conc-reach-uniform},
there exists a uniformly
$(\epsilon_\alpha/2)$-optimal public 1-bit strategy $\zstrat'_\alpha$ for Maximizer on
$\game_\alpha$ with objective $\weighted{f'_\alpha}$.
We now extend $\zstrat'_\alpha$ to a uniformly
$\epsilon_\alpha$-optimal public 1-bit strategy $\zstrat_\alpha$ for Maximizer on
$\game_\alpha$ with objective $\weighted{f_\alpha}$.
It suffices for Maximizer to play $(\epsilon_\alpha/2)$-optimal
in all states $\states_\alpha \cap \states^\infty$
w.r.t.\ the one-shot game with reward function $f_\alpha$,
regardless of the current memory mode.
(These one-shot games with infinite Minimizer action sets and finite Maximizer
action sets have a value by \cite[Theorem 3]{Flesch-Predtetchinski-Sudderth:2020},
and thus Maximizer can play $\epsilon_\alpha/2$-optimally.)
After this one-shot game,
the index of any successor state will always be $<\alpha$, by definition
of the ranking function.
}

\new{
In the special case of turn-based games, $\zstrat'_\alpha$ can be chosen as deterministic
by \cref{thm:weighted-conc-reach-uniform}.  
Moreover, Maximizer is then passive in the one-shot games from states in
$\states^\infty$, since these states belong to Minimizer who has an infinitely branching 
choice there. Thus, in turn-based games, $\zstrat_\alpha$ is deterministic as well.
}

The Maximizer strategy $\zstrat =\zstrat[0]$ starts with memory mode $0$.
In every state $\state$ with $I(\state) = \alpha$
the strategy $\zstrat$ plays like $\zstrat_\alpha$. Whenever we make a step
$\state \to \state'$ with $I(\state') < I(\state)$ then its sets the memory mode
to $0$ again.
(It is impossible that $I(\state') > I(\state)$ by the definition of $I$.)
Since all the $\zstrat_\alpha$ are public 1-bit strategies, so is $\zstrat$.
In the special case of turn-based games, the $\zstrat_\alpha$ are 
deterministic and thus also $\zstrat$ is deterministic.

We now show by induction on $\alpha$
(for every $\alpha \le \gamma(\game)$)
that $\zstrat$
is uniformly $\eps'_\alpha$-optimal in $\game$ for objective $\reach{\reachset}$
from every state $\state \in \states_{\le \alpha}$,
where $\eps'_\alpha \eqdef \sum_{\beta \le \alpha} \eps_\beta$.

In the base case of $\alpha=0$ we have $\game=\game_0$ on $\states_0$,
$\eps'_0 = \eps_0$ and $\zstrat = \zstrat_0$.
Since $\reachset_0 = \reachset \subseteq \states_0$, the objectives
$\weighted{f_0}$ and $\reach{\reachset}$ coincide.
\new{
Formally, for any $\zstrat',\ostrat'$, we have
\begin{equation}\label{eq:restrictedmin-basecase}
\expectval_{\game_0,\state,\zstrat',\ostrat'}(f_0)=
\probm_{\game_0,\state,\zstrat',\ostrat'}(\reach{\reachset})
\end{equation}
By our construction above, $\zstrat_0$ is
a uniformly $\epsilon_0$-optimal public 1-bit strategy for Maximizer on
$\game_0$ with objective $\weighted{f_0}$.
Thus, for every $\state \in \states_0$ we have
\begin{align*}
& \inf_\ostrat \probm_{\game,\state,\zstrat,\ostrat}(\reach{\reachset})
\\
  & = \inf_\ostrat \probm_{\game_0,\state,\zstrat_0,\ostrat}(\reach{\reachset})
  & \mbox{$\game=\game_0$ and $\zstrat=\zstrat_0$}  
\\
  & = \inf_\ostrat \expectval_{\game_0,\state,\zstrat_0,\ostrat}(f_0)
  & \mbox{by \eqref{eq:restrictedmin-basecase}}
\\
  & \ge \valueof{\game_0,\weighted{f_0}}{\state} - \eps_0
  & \mbox{$\eps_0$-optimality of $\zstrat_0$}
\\
  & = \valueof{\game,\reach{\reachset}}{\state} - \eps'_0
  & \mbox{by \eqref{eq:val-reach-equal-val-weight} and $\eps'_0 = \eps_0$} 
\end{align*}
}
For the induction step let $\alpha >0$.
If $\state \in \states_{< \alpha}$ then the claim holds by induction
hypothesis.
Now let $\state \in \states_\alpha$
and $\ostrat$ be an arbitrary Minimizer strategy.
Let $\playset$ be the set of induced plays from $\state$ under $\zstrat$
and $\ostrat$ and
$\playset_{\state'} \subseteq \playset$ be the subset of plays
where $\state'$ is the first visited state with $I(\state') < \alpha$.
\new{
Recall that $\zstrat$ is a 1-bit strategy with two memory modes $\{0,1\}$.
For $\memconf \in \{0,1\}$, $\zstrat[\memconf]$ denotes the strategy $\zstrat$
with current memory mode $\memconf$. Therefore, $\zstrat[\memconf]$
can be applied to start at any state, since it does not depend on the history.
The initial memory mode is $0$, i.e., $\zstrat=\zstrat[0]$.
}
\begin{align*}
& \probm_{\game,\state,\zstrat,\ostrat}(\reach{\reachset})\\
& \new{= \probm_{\game,\state,\zstrat[0],\ostrat}(\reach{\reachset})}\\  
& \ge
\sum_{\state' \in \states_{< \alpha}}
\probm_{\game,\state,\zstrat[0],\ostrat}(\playset_{\state'})\cdot \inf_{\ostrat'}\probm_{\game,\state',\new{\zstrat[0]},\ostrat'}(\reach{\reachset}) \\
& \new{\mbox{(the memory mode of $\zstrat$ is set to $0$ at $\state'$,
since $I(\state')<\alpha$)}}
\\
& \ge
\sum_{\state' \in \states_{< \alpha}}
\probm_{\game,\state,\zstrat[0],\ostrat}(\playset_{\state'})\cdot 
(\valueof{\game,\reach{\reachset}}{\state'} - \eps'_{I(\state')}) & \mbox{ind.~hyp.}
\\
&  =
\sum_{\state' \in \states_{< \alpha}}
\probm_{\game,\state,\zstrat[0],\ostrat}(\playset_{\state'})\cdot f_{\alpha}(\state')
-
\sum_{\state' \in \states_{< \alpha}}
\probm_{\game,\state,\zstrat[0],\ostrat}(\playset_{\state'})\cdot \eps'_{I(\state')}  
& \mbox{def.~$f_\alpha$}
\\
&  \ge
\expectval_{\game_\alpha,\state,\zstrat_\alpha,\ostrat}(f_\alpha)
-
\sup_{\state' \in \states_{< \alpha}} \eps'_{I(\state')}
\\
&  \ge
\valueof{\game_\alpha,\weighted{f_\alpha}}{\state} - \eps_\alpha
-
\sup_{\state' \in \states_{< \alpha}} \eps'_{I(\state')} & \mbox{$\zstrat_\alpha$ is $\eps_\alpha$-optimal}
\\
&  =
\valueof{\game,\reach{\reachset}}{\state} - \eps_\alpha
-
\sup_{\state' \in \states_{< \alpha}} \eps'_{I(\state')} & \mbox{by \eqref{eq:val-reach-equal-val-weight}}
\\
&  =
\valueof{\game,\reach{\reachset}}{\state}
-
(\eps_\alpha + \sup_{\state' \in \states_{< \alpha}} \eps'_{I(\state')})
\\
&  \ge
\valueof{\game,\reach{\reachset}}{\state}
-
\eps'_\alpha
\end{align*}
Therefore, for every $\state \in \states' = \states_{\le \gamma(\game)}$ our
strategy $\zstrat$ is $\eps'_{\gamma(\game)}$-optimal.
Moreover, 
$\eps'_{\gamma(\game)} = \sum_{\beta \le \gamma(\game)} \eps_\beta =
\sum_{\beta \le \gamma(\game)} \eps \cdot 2^{-g(\beta)} \le \eps$,
since $g: \{\alpha \in \ord \mid \alpha \le \gamma(\game)\} \to \N$
is injective.
Thus $\zstrat$ is uniformly $\eps$-optimal from $\states'$ in $\game$.
\end{proof}
 
\section{Optimal Maximizer Strategies}\label{sec:optimalmax}
\new{
In finite turn-based reachability games,
there always exist optimal Maximizer strategies, and even optimal memoryless deterministic
ones \cite{CONDON1992203}, \cite[Proposition 5.6.c, Proposition 5.7.c]{kucera_2011}.
This does not carry over to finite concurrent reachability games.
E.g., in the \emph{snowball game} (aka \emph{Hide-or-Run} game)
described in 
\cite[Example 1]{Everett1957} and
\cite{KumarShiau,AlfaroHK98},
Maximizer does not have any optimal strategy.
However, it was recently shown by
\cite{BordaisB022} that, in finite concurrent games with finite action sets,
optimal Maximizer strategies, if they exist, can be chosen as memoryless randomized.
}

In countably infinite reachability games,
optimal strategies for Maximizer need not exist in general
even if the game is turn-based, in fact not even in countably
infinite MDPs that are finitely branching \cite{Ornstein:AMS1969,KMSW2017}.

In this section we study the memory requirements of optimal Maximizer
strategies under the condition that such an optimal strategy exists.

If we allow infinite action sets for Minimizer (resp.\ infinite Minimizer
branching in turn-based games) then optimal (and even almost surely winning)
Maximizer strategies require infinite memory by \cref{thm:no-sc-plus-finite}.
Thus, in the rest of this section, we consider games with finite action sets
(resp.\ turn-based games where the players are finitely branching).

\subsection{Turn-Based Games}\label{subsec:optmax-turn}

\new{
Here we consider turn-based reachability games where the players have only
finitely many choices at each controlled state (i.e., finite action sets).
It turns out that the memory requirements of optimal Maximizer strategies, if
they exist, also depend on whether random states are infinitely branching or
finitely branching, i.e., on whether these distributions have finite support.
}

\new{
If we allow infinite branching at random states, then optimal Maximizer
strategies require infinite memory, even with a step counter, by the following example.
(A weaker result, without considering the step counter, was shown in 
\cite[Prop.~5.7.b]{kucera_2011}.)
}

\new{
\begin{definition}\label{def:conc-optmax}
Let $\game$ be the following turn-based reachability game 
depicted in \cref{fig:optimalmax-0},
where Maximizer and Minimizer have only finite branching (i.e., finite action sets),
with initial state $\state_0$ and target state $t$.
State~$\state_0$ is a random state and
the distribution $p(\state_0)$ over its infinitely many successor
states is defined as $p(\state_0)(\state_i') = \frac{1}{2^i}$ for all $i \ge 1$.
Further, for every $i\ge 1$ there is a Minimizer-controlled state $\state_i'$ and a  
Maximizer-controlled state $\state_i''$.
In $\state_i'$ Minimizer chooses between moving to state $\state_1''$
or (via a random state) to the target with probability $1 - \frac{1}{2^i}$ and to a losing sink with probability $\frac{1}{2^i}$.
At $\state_i''$ Maximizer chooses between moving to state $\state_{i+1}''$
or (via a random state) to target $t$ with probability $1 - \frac{1}{2^i}$ and
to a losing sink with probability $\frac{1}{2^i}$.
\end{definition}
}

\begin{figure}
\begin{center}
\includegraphics[width=0.75\textwidth, angle=0]{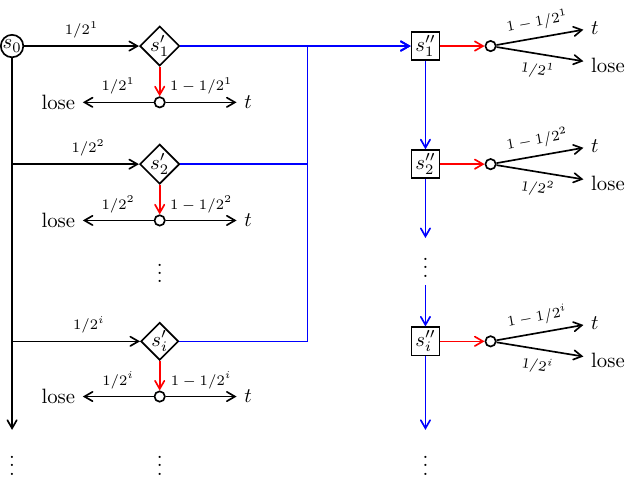}
\end{center}
\caption{The game $\game$ from \cref{def:conc-optmax}. Choices depicted in red immediately end the game after one round, their alternative choices are drawn in blue.}
\label{fig:optimalmax-0}
\end{figure}

\begin{proposition}\label{prop:conc-optmax}
\new{There exists a turn-based reachability game $\game$ where Maximizer and Minimizer
have only finite branching (i.e., finite action sets)
with initial state $\state_0$ and objective $\reach{\{t\}}$
as in \Cref{def:conc-optmax}, such that}
\begin{enumerate}
\item\label{prop:conc-optmax-1}
Maximizer has an optimal strategy from $\state_0$.
\item\label{prop:conc-optmax-2}
Every randomized Maximizer strategy from $\state_0$ that uses only a step
counter and finite private memory is not optimal.
\end{enumerate}
\end{proposition}
\begin{proof}
We have $\valueof{\game}{\state_i''} = 1$ for all~$i$, and so $\valueof{\game}{\state_i'} = 1 - \frac{1}{2^i}$ for all~$i$, and so $\valueof{\game}{\state_0} = \sum_{i=1}^\infty \frac{1}{2^i} \cdot (1 - \frac{1}{2^i})$.
(The latter series equals~$\frac23$, but that will not be needed.)
It follows that the only optimal Minimizer strategy 
is the one that chooses the red option
at any state $\state_i'$ (where $i \ge 1$).
Note that Maximizer does not make any choices if Minimizer plays her optimal strategy.

Towards \Cref{prop:conc-optmax-1}, Maximizer's optimal strategy $\zstrat$ from $\state_0$ is defined as follows.
In plays where the state $\state_1''$ is not reached, Maximizer does not make any decisions.
If $\state_1''$ is reached, Maximizer considers the history of this play:
If Minimizer chose the move from $\state_i'$ to $\state_1''$ for some $i\ge 1$, then Maximizer chooses moves (via states $\state_2'', \ldots, \state_{i-1}''$) to state~$\state_i''$ for the same~$i$, and at state~$\state_i''$ he chooses the red option (end the game and win with probability $1-2^{-i}$. 
Note that in this way Maximizer takes the action that Minimizer refused to take (although it would have been optimal for her) at~$\state_i'$.
With this Maximizer strategy~$\zstrat$, for every Minimizer strategy
$\ostrat$,
the probability to reach~$t$ equals 
\verynew{
\[
    \probm_{\game,\state_0,\zstrat,\ostrat}(\reach{\{t\}})
    =
    \sum_{i=1}^\infty \frac{1}{2^i} \cdot \left(1 - \frac{1}{2^i}\right)
    = \valueof{\game}{\state_0}
  \]
}
meaning that $\zstrat$ is optimal.

Towards \Cref{prop:conc-optmax-2}, we note that the step counter from $\state_0$ is implicit in the states of $\game$ (except in the target~$t$ and the losing sink state), and thus superfluous for Maximizer strategies.
Hence it suffices to prove the property for Maximizer strategies with finite memory.
Let $\zstrat$ be an FR Maximizer strategy with finitely many memory modes $\{1,\dots,k\}$.
At state $\state_1''$ this strategy $\zstrat$ can base its decision only on the current memory mode $\memconf \in \{1,\dots,k\}$.
Let
\verynew{$X(\memconf) \eqdef \inf_\ostrat\probm_{\game,\state_1'',\zstrat[\memconf],\ostrat}(\reach{\{t\}})$}
be the probability of reaching the target if $\zstrat$ is in mode~$\memconf$ at state~$\state_1''$.
(From state~$\state_1''$ only Maximizer plays, thus Minimizer has no influence.)
Since $X(\memconf) < 1$ and the memory is finite, we have $Y \eqdef \max_{\memconf \in \{1,\dots,k\}} X(\memconf) < 1$.
There exists a number $i$ sufficiently large such that $Y < 1 - \frac{1}{2^i}$.
Let $\ostrat$ be a Minimizer strategy from $\state_0$ that takes the blue option from $\state_i'$ to~$\state_1''$, but chooses the red option in all states $\state_j'$ with $j \ne i$.
Then we have
\verynew{
\begin{align*}
\probm_{\game,\state_0,\zstrat,\ostrat}(\reach{\{t\}})
\ &\le \ \frac{1}{2^i} Y + \sum_{j \ne i} \frac{1}{2^j} \cdot \left(1 - \frac{1}{2^j}\right) \\
\ &< \ \sum_{j=1}^\infty \frac{1}{2^j} \cdot \left(1 - \frac{1}{2^j}\right)
\ = \ \valueof{\game}{\state_0}
\end{align*}
}
and thus $\zstrat$ is not optimal.
\end{proof}

\new{
Note that the counterexample in \Cref{def:conc-optmax,prop:conc-optmax}
has some particular properties. Even though the players have finite action
sets, the random state $\state_0$ is infinitely branching.
Moreover, while $\state_0$ admits an optimal Maximizer strategy,
the same does \emph{not} hold for all states in the game, e.g., the states $\state_i''$ have
value $1$, but do not admit any optimal Maximizer strategy.
}

\new{
The following theorem shows that if we impose any such extra condition on the game
(i.e., even all random states are finitely branching, or all states admit an
optimal Maximizer strategy) then the memory requirements of optimal Maximizer
strategies are somewhat lower.  
In these cases, just a step counter and 1 bit of public memory are sufficient.
}

\begin{theorem}\label{thm:turn-fb-optmax-upper}
Let $\game$ be a turn-based reachability game \new{with finite action sets}
with initial state~$s_0$ and objective $\reach{\{t\}}$
\new{such that at least one of the following two conditions is satisfied:
\begin{description}
\item[(A)]
$\game$ is finitely branching (at all states, including the random states), or
\item[(B)]
  Every state in $\game$ admits an optimal Maximizer strategy.
\end{description}
}
Suppose that Maximizer has an optimal strategy~$\sigma$, i.e.,
$\probm_{\game,s_0,\sigma,\pi}(\reach{\{t\}}) \ge \valueof{\game}{s_0}$ holds
for all Minimizer strategies~$\pi$.
Then Maximizer also has a deterministic such strategy that uses 1 bit of public memory and a step counter.
\end{theorem}

\new{
In the proof of \Cref{thm:turn-fb-optmax-upper} we will use the following version of the optional stopping theorem.
\begin{theorem}[Optional Stopping Theorem] \label{thm:optional-stopping}
Suppose $X_0, X_1, \ldots$ is a submartingale adapted to a filtration $\mathcal{F}_0, \mathcal{F}_1, \ldots$; i.e., $X_n \le \expectation(X_{n+1} \mid \mathcal{F}_n)$ for all~$n$.
Suppose further that there is $c \in \mathbb{R}$ with $\abs{X_n} \le c$ almost surely for all~$n$.
Then the limit $X_\infty \eqdef \lim_{n \to \infty} X_n$ exists almost surely.
Let $\tau_1, \tau_2$ be stopping times with $\tau_1 \le \tau_2$ almost surely (where $\tau_1 = \infty$ and $\tau_2 = \infty$ may have positive probability).
Then we have $X_{\tau_1} \le \expectation(X_{\tau_2} \mid \mathcal{F}_{\tau_1})$ almost surely.
\end{theorem}
\begin{proof}
The proof is immediate from \cite[Proposition~IV-5-24, Corollary~IV-2-25]{Neveu:1975}.
\end{proof}
Notice that if, in addition to the other preconditions of
\cref{thm:optional-stopping}, the submartingale $X_0, X_1, \ldots$ is a martingale, i.e., $X_n = \expectation(X_{n+1} \mid \mathcal{F}_n)$ for all~$n$, then it follows, by considering $Y_n \eqdef -X_n$ for all~$n$, that we have $X_{\tau_1} = \expectation(X_{\tau_2} \mid \mathcal{F}_{\tau_1})$ almost surely.
}

For the proof of \cref{thm:turn-fb-optmax-upper} we also
use \cite{KieferMSW17a} Theorem~5(2), slightly generalized as the following
lemma.

\begin{lemma}\label{lem:no-val-inc}
Let $\game$ be a turn-based reachability game,
\new{such that Minimizer has finite action sets and}
Minimizer does not have any value-increasing transitions;
i.e., for all transitions $s \transition{} s'$ with $s \in \ostates$ we have $\valueof{\game}{s} = \valueof{\game}{s'}$.

Then there exists some MD Maximizer strategy that is optimal from every state that admits an optimal strategy.
\end{lemma}
\begin{proof}
\new{First we consider the special case where $\game$ is finitely branching (i.e.,
at every state, not just at the Minimizer-controlled states).}
The statement then
follows from \cite[Thm.~5(2)]{KieferMSW17a}, but there it is stated only for a single initial state that admits an optimal strategy.
Therefore, denote by $\statesopt \subseteq S$ the set of states that admit an optimal strategy.
Add a fresh random state, say $s_0$, such that the support of $\probp(s_0)$ equals~$\statesopt$.
This might require infinite branching, but one can easily encode infinite
branching of random states into finite branching in the case of reachability
objectives, using a ``ladder'' gadget of fresh intermediate finitely branching random states.
Since every state in~$\statesopt$ admits an optimal strategy, the new state~$s_0$ admits an optimal strategy.
The mentioned result \cite[Thm.~5(2)]{KieferMSW17a} applied to~$s_0$ gives an
MD Maximizer strategy~$\sigma$ that is optimal starting from~$s_0$.
But then $\sigma$ must be optimal from every state in~$\statesopt$.

\new{
The above result can be generalized to allow infinitely branching random
states by the same encoding as above, using a ``ladder'' of fresh intermediate
finitely branching random states.
Similarly, infinitely branching Maximizer states can also be encoded into
a ``ladder'' of fresh intermediate finitely branching Maximizer states.
This encoding gives Maximizer the additional option to remain on the ladder forever,
but this is not a problem. Since the target is not on the ladder, staying on
the ladder forever would be losing for Maximizer.
Finally, since we are dealing with MD strategies, the strategies can be
carried back
\verynew{from the finitely branching game that uses the ``ladder'' gadget encoding
to the infinitely branching original game}.
(The same would not hold for Markov strategies in general, since the encoding
does not preserve path lengths. Also it is not possible to encode infinite
Minimizer branching in this way, because Minimizer could spuriously win by
staying on the ladder gadget forever.)}
\end{proof}

\begin{proof}[Proof of \cref{thm:turn-fb-optmax-upper}.]
Denote by $\bar\game$ the game obtained from~$\game$ by deleting Minimizer's value-increasing transitions, i.e., those transitions $s \transition{} s'$ with $s \in \ostates$ and $\valueof{\game}{s} < \valueof{\game}{s'}$.
As $\game$ \new{has finite Minimizer action sets}, each Minimizer state still
has at least one outgoing transition in $\bar\game$, and all states have the same value in $\game$ and~$\bar\game$, and all states that admit an optimal Maximizer strategy in~$\game$ admit an optimal Maximizer strategy in~$\bar\game$ and vice versa.
Denote by $\statesopt \subseteq \states$ the set of states that admit an
optimal Maximizer strategy.
By \cref{lem:no-val-inc}, there exists an MD Maximizer strategy $\bar\sigma$ in~$\bar\game$ that is optimal from every state in $\statesopt$.
Thus, for any $s \in \zstates \cap \statesopt$, the transition $s \transition{} s'$ that $\bar\sigma$ prescribes preserves the value, i.e., $\valueof{}{s} = \valueof{}{s'}$, and $s' \in \statesopt$.
By assumption, $s_0 \in \statesopt$.
Thus, $\bar\sigma$ is optimal from~$s_0$ in~$\bar\game$.
Hence, $\bar\sigma$ is optimal from~$s_0$ in~$\game$ if Minimizer never chooses a value-increasing transition $s \transition{} s'$ with $s \in \ostates$ and $\valueof{}{s} < \valueof{}{s'}$.
We view such a transition as a \emph{gift} from Minimizer of size
$\valueof{}{s'} - \valueof{}{s} >0$.

Let us now sketch a first draft of a Maximizer strategy that is optimal
from~$s_0$ in $\game$.
\begin{itemize}
\item
Play the strategy~$\bar\sigma$ until Minimizer gives a gift of, say, $\eps >0$. 
Use the 1 bit of public memory to record the fact that a gift has been given.
\item
Then play an $\eps$-optimal MD strategy, which exists by \cref{lem:conc-reach-non-uniform}.
\end{itemize}
The problem with this draft strategy is that storing the \emph{size} of
Minimizer's gift $\eps$ appears to require infinite memory, not just 1
bit, because $\eps$ could be arbitrarily small.
Moreover, the different $\eps$-optimal MD strategies might prescribe
different choices for different $\eps$.

Therefore, we use the step counter to deduce a lower bound on any
nonzero gift that Minimizer may have given up to that point in time.

\new{
Let $R(i)$ be the set of states that could be reached from $s_0$
with nonzero probability under any pair of strategies within $\le i$ steps.
}

\new{Under condition (A), $R(i)$ is \emph{finite} for every $i \in \N$,
because $\game$ is finitely branching.
Here we just define the finite set $\states(i) \eqdef R(i)$.
(Under condition (B), $\states(i)$ will be a subset instead.)
}

\new{
Under condition (B), even though both players have finite action sets,
random states are still allowed to be infinitely branching.
Thus $R(i)$ could be infinite.
However, since the players have finite action sets, for every time $i \ge 0$
and $\delta >0$, there exists a \emph{finite} subset of states
$\states(i,\delta) \subseteq R(i)$
such that 
under any pair of strategies $\zstrat,\ostrat$ from $s_0$,
the probability of ever being outside $\states(i,\delta)$
at any time $t \le i$ 
is upper-bounded by $\delta$.
I.e.,
\begin{equation}\label{eq:unlikely-outside-si}
  \forall\zstrat,\ostrat 
  \ \probm_{\game,\state_0,\zstrat,\ostrat}(\reachn{i}{\states\setminus\states(i,\delta)}) \le \delta.
\end{equation}
Since only random states can be infinitely branching,
$\states(i,\delta)$ can easily be defined by cutting infinite tails off
distributions, e.g., losing $\le \delta \cdot 2^{-(t+1)}$ in the $t$-th round.
Additionally, we can define these sets such that they are monotone
increasing in $i$. That is, $\states(i,\delta) \subseteq \states(i+1,\delta)$
for all $i \in \N$.
We then define $\states(i) \eqdef \states(i,2^{-i})$, and these sets are
also monotone increasing in $i$.
}

\new{
Let $\eps_i >0$ denote the size of the smallest nonzero gift that
Minimizer can give from any state inside $\states(i)$,
i.e.,
\[
  \eps_i \eqdef \min\{(\valueof{}{\state'} - \valueof{}{\state}) > 0 \mid
  \state \to \state'\ \wedge\ \state \in \states(i) \cap \ostates\}.
\]
We have $\eps_i >0$, because $\states(i)$ is finite and
Minimizer has finite action sets.
Moreover, the $\eps_i$ are monotone decreasing in $i$, because the sets
$\states(i)$ are monotone increasing.
}

\new{
Under condition (A), $\eps_i$ is
a lower bound on \emph{all} possible Minimizer gifts until time $i$,
while under condition (B) it is only a lower bound on
\emph{most} of Minimizer's possible
gifts until time $i$ (namely on those originating from a state in the
subset $\states(i)$).
}

\new{
However, we will show that, under condition (B), it is safe for Maximizer to
ignore gifts from Minimizer if gifts are given only finitely often.
(This does not hold under condition (A).)
So our Maximizer strategy will ignore gifts from Minimizer at time $i$ if
the gift originates from a state \emph{outside} $\states(i)$.
Indeed, except for a nullset of plays, Minimizer cannot give a gift at infinitely many
times $i$ at states outside of $\states(i)$, because it is so unlikely to be
outside $\states(i)$ at time $i$.
Consider an arbitrary pair of strategies $\zstrat, \ostrat$ and
let $\playset \subseteq \state_0\states^\omega$ be the set of plays
$s_0 s_1 \cdots s_i s_{i+1} \dots$ from $\state_0$ where
$\state_i \notin \states(i)$ for infinitely many $i \in \N$.
}

\new{
\begin{claim}\label{claim:only-fin-non-ignored-gifts}
$\forall\zstrat,\ostrat\ \probm_{\game,\state_0,\zstrat,\ostrat}(\playset) = 0$.
\end{claim}
\begin{proof}
Consider a number $k \in \N$. For every play $s_0 s_1 \cdots s_i s_{i+1} \dots$
in $\playset$ there exists a number $k' > k$ such that $\state_{k'} \notin \states(k')$.
Let $\playset_{k'} \subseteq \playset$ be the subset of plays where $k'$
is the smallest number $> k$ where $\state_{k'} \notin \states(k')$.
Then $\playset$ can be partitioned as
$\playset = \uplus_{k' > k} \playset_{k'}$.
However, by $\states(k') = \states(k',2^{-k'})$ and
\eqref{eq:unlikely-outside-si}, we have
$\probm_{\game,\state_0,\zstrat,\ostrat}(\playset_{k'}) \le 2^{-k'}$
and thus
$\probm_{\game,\state_0,\zstrat,\ostrat}(\playset) \le \sum_{k' >k} 2^{-k'}
\le 2^{-k}$.
Since this holds for every $k \in \N$, the result follows.
\end{proof}
}

Another problem with the draft strategy is that the $\varepsilon$-optimal MD strategy from \cref{lem:conc-reach-non-uniform} is $\varepsilon$-optimal only from a finite set of initial states (uniformly $\varepsilon$-optimal memoryless strategies do not always exist; see \cref{thm:TB-BMI-zplus}).

Therefore, we partition time into infinitely many finite \emph{phases} $\Phi_1 = \{1, \ldots, t_1\}, \Phi_2 = \{t_1+1, \ldots, t_2\}, \Phi_3 = \{t_2+1,\ldots, t_3\}$, etc., and refer to $\Phi_1, \Phi_3, \ldots$ as \emph{odd} and to $\Phi_2, \Phi_4, \ldots$ as \emph{even} phases.
The length of the phases is determined inductively; see below.
\new{
Let $S_i$ with $S(t_{i-1}+1) \subseteq S_i \subseteq R(t_{i-1}+1)$
be a sufficiently large finite subset of the states that could be reached by
the beginning of phase $i$ such that the following condition holds:
Under any pair of strategies, for
any $s \in S(t_{i-1})$, conditioned under the event that $s$ has been visited
at some time $t \le t_{i-1}$, the probability of being inside $S_i$
at time $t_{i-1}+1$ is $\ge 1-(\eps_{t_{i-1}}/2)$.
Under condition (A), we can simply take $S_i \eqdef R(t_{i-1}+1)$, since that
is finite.
Under condition (B), since $S(t_{i-1})$ is finite and the players have finite
action sets, we obtain a suitable $S_i$ by cutting suitably small tails off
distributions.
Note that $S_i$ does not depend on the pair of strategies.
}  
The lengths of the phases $\Phi_1, \Phi_2, \ldots$ are determined as follows.
\begin{enumerate}
\item[(LO)]
  Each odd phase~$\Phi_i$ is long enough
  (i.e., $t_i$ is chosen large enough)
  so that we have $\inf_\pi \probm_{\bar\game,s,\bar\sigma,\pi}(\reachn{\Phi_i}{\{t\}}) \ge \frac{\valueof{}{s}}{2}$ for all $s \in S_i \cap \statesopt$, where we write $\reachn{\Phi_i}{\{t\}}$ for the event that $t$ is reached within $t_i - t_{i-1}$ steps, i.e., the length of phase~$\Phi_i$.
That is, if $\Phi_i$~begins at a state $s \in S_i \cap \statesopt$ and if Minimizer does not give a gift during~$\Phi_i$, the Maximizer strategy~$\bar\sigma$ realizes at least half of the value of~$s$ already within~$\Phi_i$.
\new{The length of $\Phi_i$ can be chosen finite, because $S_i$ is finite and Minimizer
has finite action sets.}
\item[(LE)]
\new{
For each even phase~$\Phi_i$, by \cref{lem:conc-reach-non-uniform}, there is
an MD Maximizer strategy~$\sigma_i$ so that $\Phi_i$~can be made long enough
so that we have $\inf_\pi \probm_{\game,s,\sigma_i,\pi}(\reachn{\Phi_i}{\{t\}}) \ge \valueof{}{s} -
(\varepsilon_{t_{i-1}}/2)$ for all $s \in S_i$.
Again the length of $\Phi_i$ can be chosen finite, because $S_i$ is finite and Minimizer
has finite action sets.
If Minimizer has given a gift
from a state $s \in \states(t_{i-1})$
in the previous phase, then this gift will be $\ge \varepsilon_{t_{i-1}}$.
Moreover, by the definition of $S_i$,
we will then be in a state in $S_i$ at the beginning of phase $\Phi_i$ with 
very high conditional probability $\ge 1-(\eps_{t_{i-1}}/2)$ (or even surely under
condition (A)).
Thus, in the phase $\Phi_i$,
the Maximizer strategy~$\sigma_i$ can undercut
Minimizer's gift and realizes most of the value of~$s$ already within~$\Phi_i$.}
\end{enumerate}
We now define a deterministic Maximizer strategy~$\sigma$ from $\state_0$
that uses a step counter and 1 bit of public memory.
Later we show that $\sigma$ is optimal from $\state_0$.
Strategy~$\sigma$ uses two memory modes, $\memconf_0$ and $\memconf_1$, where $\memconf_0$~is the initial mode.
Strategy~$\sigma$ \emph{updates} the mode as follows.
\begin{enumerate}
\item[(U1)] While in~$\memconf_0$ and in an odd phase $\Phi_i$:
  if Minimizer gives a gift
  \new{from a state $s \in \states(t_i)$}
  switch to~$\memconf_1$.
  I.e., Maximizer uses the bit to remember that Minimizer has given a gift
  and will undercut it in the next even phase.
  \new{The size of the gift is lower-bounded by $\eps_{t_i} >0$.
  Note that Maximizer ignores all Minimizer gifts from states outside $\states(t_i)$
  (which can only happen under condition (B)).}
\item[(U2)] While in~$\memconf_0$ and upon entering an odd phase: if the new state does not admit an optimal strategy, switch to~$\memconf_1$.
    This can only happen if Minimizer has
    given a gift in some previous \emph{even} phase
    \new{(and not at all under condition (B)).}
    If Minimizer had given a gift in some previous odd phase then the memory
    mode would already be $\memconf_1$. If Minimizer has never given a gift then the
    current state would still admit an optimal strategy, since Maximizer never voluntarily
    leaves $\statesopt$.
 \end{enumerate}
Note that once the mode has been switched to~$\memconf_1$ it is never switched back to~$\memconf_0$.
Strategy~$\sigma$ \emph{plays} as follows.
\begin{enumerate}
\item[(P1)]
While in~$\memconf_0$ and in~$\statesopt$: play~$\bar\sigma$.
This keeps the game in~$\statesopt$, at least until possibly Minimizer gives a
gift.
\new{(Under condition (A), the play might leave $\statesopt$ after a Minimizer gift.
Under condition (B), all plays stay inside $\statesopt$.)} 
\item[(P2)]
While in~$\memconf_0$ and in a state $s \in \zstates \setminus \statesopt$: choose a value-preserving transition, i.e., $s \transition{} s'$ with $\valueof{}{s} = \valueof{}{s'}$.
Such a transition must exist, due to the finite Maximizer branching in~$\game$.
\item[(P3)]
While in~$\memconf_1$ during an odd phase: choose a value-preserving transition, i.e., $s \transition{} s'$ with $\valueof{}{s} = \valueof{}{s'}$.
Such a transition must exist, due to the finite Maximizer branching in~$\game$.
\new{Intuitively, Maximizer has recorded the fact that Minimizer has given a
  non-ignored gift, but waits until the next even phase to capitalize on it.}
\item[(P4)]
While in~$\memconf_1$ during an even phase~$\Phi_i$: play the MD strategy
$\sigma_i$ from the
definition (LE) of the even phase~$\Phi_i$.
It follows from (U1) and~(U2) that $\sigma_i$~has been played from the
beginning of~$\Phi_i$.
(Here Maximizer undercuts Minimizer's previous gift.)
\end{enumerate}
Note that not all possible gifts by Minimizer are detected, i.e., result in a
switch to memory mode $\memconf_1$.
\new{First, gifts in phase $\Phi_i$ from states outside $\states(t_i)$ are ignored.}
Moreover, Minimizer could give a gift in an even phase while staying in $\statesopt$,
or the game might just temporarily leave $\statesopt$ but move back to
$\statesopt$ before the next odd phase,
thus avoiding rule~(U2).
However, this is not \new{a} problem for Maximizer:
Since the game returns to $\statesopt$ before the next odd phase, Maximizer
is fine to just continue playing $\bar\sigma$ by~(P1), because he will realize at
least half of the value
\new{(at least of most states; those in $S_i$)}
during the next odd phase.

To show that $\sigma$~is optimal from~$s_0$, fix an arbitrary Minimizer strategy~$\pi$ for the rest of this proof, and assume that the target~$t$ is a sink.
Let us write $\probm$ for $\probm_{\game,\state_0,\zstrat,\ostrat}$ and $\expectation$ for the associated expectation.
We need to show that $\valueof{}{s_0} \le \probm(\reach{\{t\}})$.

For a play $s_0 s_1 \cdots \in \{s_0\} S^\omega$, define a random variable
$\tau_1$, taking values in $\N \cup \{\infty\}$, such that $\tau_1 = \infty$
if Minimizer never gives a gift
\new{or all Minimizer gifts are ignored},
and $\tau_1 = j < \infty$ if $s_j \transition{} s_{j+1}$ is the first
\new{non-ignored} Minimizer gift.
Also define a random variable $\tau_2$, taking values in $\N \cup \{\infty\}$, such that $\tau_2 = \infty$ if no mode switch from $\memconf_0$ to~$\memconf_1$ ever occurs, and $\tau_2 = k < \infty$ if $k$~is the beginning of the even phase following a mode switch from $\memconf_0$ to~$\memconf_1$.
The random variables $\tau_1+1, \tau_2$ are both stopping times.
As long as Minimizer does not give any \new{non-ignored}
gift, Maximizer plays~$\bar\sigma$ and, by~(P1), keeps the game in~$\statesopt$, and thus, by (U1) and~(U2), the mode remains~$\memconf_0$.
Hence, $\tau_1 < \tau_2 \le \infty$ or $\tau_1 = \infty = \tau_2$.

Further, define random variables $V_0, V_1, \ldots$ and $W_0, W_1, \ldots$, taking values in~$[0,1]$, by 
$V_i \eqdef \valueof{}{s_i}$ for all $i \le \tau_1$, and $V_i = V_{\tau_1}$ for all $i \ge \tau_1$, and
$W_i \eqdef \valueof{}{s_i}$ for all $i \le \tau_2$, and $W_i = W_{\tau_2}$ for all $i \ge \tau_2$.
By (P1), (P2) and~(P3), Maximizer preserves the value in each of his transitions, at least until~$\tau_2$.
Thus, $W_0, W_1, \ldots$ is a submartingale.
\new{
Minimizer cannot decrease the value, but might increase it when giving a gift.
Under condition (B), Minimizer might give ignored gifts before $\tau_1$.
Thus, $V_0, V_1, \ldots$ is a submartingale.
Under condition (A), Minimizer gifts are never ignored, and
thus $V_0, V_1, \ldots$ is even a martingale.
}
By \new{\cref{thm:optional-stopping}}, $V_0, V_1, \ldots$ and $W_0, W_1, \ldots$ converge almost surely to random variables, which we may call, without risk of confusion, $V_{\tau_1}$ (which equals $V_{\tau_1+1}$) and~$W_{\tau_2}$, respectively.
\new{Again by \cref{thm:optional-stopping}}, we have
\begin{equation} \label{eq:optimalmax-1}
 \valueof{}{s_0} \ = \ \expectation V_0 \ \new{\le} \ \expectation V_{\tau_1+1} \ = \ \expectation V_{\tau_1}\,. 
\end{equation}

Now consider the event~$\tau_2=\infty$. By (U1) and~(U2), Minimizer does not give
\new{any non-ignored} gift in any odd phase and the state at the beginning of every odd phase admits an optimal strategy for Maximizer.
\new{
This property is ensured by each of the conditions (A) and (B).
Under condition (A), Minimizer gifts are never ignored.
Condition (B) ensures that the game is always in a state that admits
an optimal strategy for Maximizer, and thus in particular at the beginning of
every odd phase.
}

\new{
Under condition (A), Minimizer will not give any gift at all in any odd phase,
and thus by (P1) Maximizer plays the strategy $\bar\sigma$ undisturbed in every odd
phase. Since (A) implies that $S_i \eqdef R(t_{i-1}+1)$, definition (LO) ensures
that in every odd phase Maximizer realizes at least half of the value of the
state at the beginning of the phase.
}

\new{
Under condition (B), Minimizer might still give ignored gifts in odd phases,
which could disrupt the attainment of Maximizer's strategy $\bar\sigma$.
However, by \cref{claim:only-fin-non-ignored-gifts}, except in a
nullset of plays, there are only finitely many ignored Minimizer gifts in a play.
I.e., almost every play is eventually undisturbed by ignored Minimizer gifts
in odd phases.
Moreover, no part of the value is lost, since $V_0, V_1, \ldots$ is a
submartingale, and every state admits an optimal strategy for Maximizer.
Hence, except in a nullset, by (LO), Maximizer \emph{eventually} realizes at least half
of the value of each state $s \in \states_i$ that it is in at the beginning of every odd phase $\Phi_i$.
Finally, since $S(t_{i-1}+1) \subseteq S_i$, the probability of being in a
state $s \in \states_i$ at the beginning of an odd phase $\Phi_i$ converges to $1$
as $i \to \infty$.
}

Therefore, \new{under either condition (A) or (B)}, we have that
\begin{equation} \label{eq:optimalmax-realize}
 \probm(\{\tau_2 = \infty\} \cap \{W_\infty > 0\} \setminus \reach{\{t\}}) \ =\ 0\,.
\end{equation}
Thus,
\[
\begin{aligned}
\expectation(V_\infty \mid \tau_1 = \infty) \ & \le \ \probm(V_\infty > 0 \mid \tau_1 = \infty) \\
 & = \ \probm(W_\infty > 0 \mid \tau_1 = \infty) \ \le \ \probm(\reach{\{t\}} \mid \tau_1 = \infty)\,.
\end{aligned}
\]
Hence, continuing~\eqref{eq:optimalmax-1},
\begin{equation} \label{eq:optimalmax-2}
\begin{aligned}
 \valueof{}{s_0} \ & \new{\le} \ \probm(\tau_1 = \infty) \cdot \expectation(V_\infty
 \mid \tau_1 = \infty) + \sum_{0 \le j < \infty} \probm(\tau_1 = j) \cdot \expectation(V_j \mid \tau_1 = j) \\
                   &\le \ \probm(\reach{\{t\}}, \tau_1 = \infty) + \sum_{0 \le j < \infty} \probm(\tau_1 = j) \cdot \expectation(V_j \mid \tau_1 = j)\,,
\end{aligned}
\end{equation}
\new{where here and henceforth, to avoid clutter, we may write ``,'' for the intersection of events.}
Let $j \in \N$.
It follows from the definitions of $\tau_1$ and~$\eps_j$ that on $\tau_1=j$ we have $W_{j+1} \ge W_j + \eps_j$.
Thus,
\begin{equation} \label{eq:optimalmax-3}
\begin{aligned}
 &\expectation(V_j \mid \tau_1=j) \\
 &=\  \expectation(W_j \mid \tau_1=j) && \text{by def.\ of $V_j, W_j$} \\
 &\le\ -\eps_j + \expectation(W_{j+1} \mid \tau_1 = j) && \text{as explained above}\\
 &\le\ -\eps_j + \expectation(W_{\tau_2} \mid \tau_1 = j) && \text{\new{\cref{thm:optional-stopping}}} \\
 &=\   -\eps_j + \probm(\tau_2=\infty \mid \tau_1=j) \cdot \expectation(W_\infty \mid \tau_1=j, \tau_2=\infty) + \mbox{} \\
 & \quad \mbox{}+ \! \sum_{j+1 \le k < \infty} \probm(\tau_2=k \mid \tau_1=j) \cdot \expectation(W_k \mid \tau_1=j, \tau_2=k)\,.
\end{aligned}
\end{equation}
Concerning the first expectation, we have
\begin{equation} \label{eq:optimalmax-4}
\begin{aligned}
\expectation(W_\infty \mid \tau_1=j, \tau_2=\infty) \ &\le \ \probm(W_\infty>0 \mid \tau_1=j, \tau_2=\infty) \\
                                                    \ &\le \ \probm(\reach{\{t\}} \mid \tau_1=j, \tau_2=\infty) &&\text{by~\eqref{eq:optimalmax-realize}.}
\end{aligned}
\end{equation}
\new{Concerning expectations under the sum,}
let $k > j$, and denote by $H(j,k)$ the set of histories $s_0 \cdots s_k \in \{s_0\} S^k$ such that for some (hence, all) extension(s) $r = s_0 \cdots s_k s_{k+1} \cdots$ we have $\tau_1(r) = j$ and $\tau_2(r) = k$.
Then we have
\begin{align*}
& \probm(\tau_1=j, \tau_2=k) \cdot \expectation(W_k \mid \tau_1=j, \tau_2=k) \\
                 &\quad =\   \sum_{h = s_0 \cdots s_k \in H(j,k)} \probm(\{h\} S^\omega) \cdot \valueof{}{s_k} && \text{by the defs.} \\
&\quad\le\ \sum_{h \in H(j,k)} \probm(\{h\} S^\omega) \cdot \big(\probm(\reach{\{t\}} \mid \{h\}S^\omega)+ \new{\eps_{k-1}/2 + \eps_{k-1}/2}\big)
                 && \text{by~(P4),\new{(LE)}} \\
& \new{\text{Minimizer's gift at time $j$ happens in some odd phase $\Phi_i$ and is $\ge \eps_{t_i}$.}}\\
  & \new{\text{The next even phase begins at time $k=t_i +1$.}}\\
  & \new{\text{Each of the two errors in this even phase are $\le \eps_{t_i}/2=\eps_{k-1}/2$.}}\\ 
&\quad\le\ \sum_{h \in H(j,k)} \probm(\{h\} S^\omega) \cdot
                                                   \big(\probm(\reach{\{t\}}
                                                   \mid \{h\} S^\omega) +
                                                   \eps_j\big) &&
                                                                  \text{$\eps_{k-1}
                                                                  \le \eps_j$
                                                                  } \\
&\quad=\ \probm(\tau_1=j, \tau_2=k) \cdot \big( \probm(\reach{\{t\}} \mid \tau_1=j, \tau_2=k) + \eps_j \big)\,.
\end{align*}
Thus,
\begin{equation*}
\begin{aligned}
    \sum_{j+1 \le k < \infty}& \probm(\tau_2=k \mid \tau_1=j) \cdot \expectation(W_k \mid \tau_1=j, \tau_2=k) \\
& \le\ \sum_{j+1 \le k < \infty} \probm(\tau_2=k \mid \tau_1=j) \cdot \big( \probm(\reach{\{t\}} \mid \tau_1=j, \tau_2=k) + \eps_j \big) \\
& \le\ \probm(\reach{\{t\}}\mid\tau_2 < \infty,\tau_1=j) + \eps_j\,.
\end{aligned}
\end{equation*}
Combined with \cref{eq:optimalmax-3,eq:optimalmax-4} this gives
\begin{equation*}
\begin{aligned}
\expectation(V_j \mid \tau_1=j)
&~\le~\probm(\reach{\{t\}} \mid\tau_2=\infty, \tau_1=j)  + \probm(\reach{\{t\}}\mid \tau_2 < \infty, \tau_1=j) \\
& \le \ \probm(\reach{\{t\}} \mid \tau_1=j)\,.
\end{aligned}
\end{equation*}
Combined with~\eqref{eq:optimalmax-2}, we obtain
\[
\valueof{}{s_0} \ \le \ \probm(\reach{\{t\}}, \tau_1=\infty) + \probm(\reach{\{t\}}, \tau_1<\infty) \ = \ \probm(\reach{\{t\}})\,,
\]
as required.
\end{proof}

\new{
The following example shows a corresponding lower bound to
\Cref{thm:turn-fb-optmax-upper},
i.e., even if \emph{both} conditions (A) and (B) hold, just a step counter does \emph{not} suffice for
optimal Maximizer strategies.
}

\begin{figure}
\begin{center}
\includegraphics[width=\textwidth, angle=0]{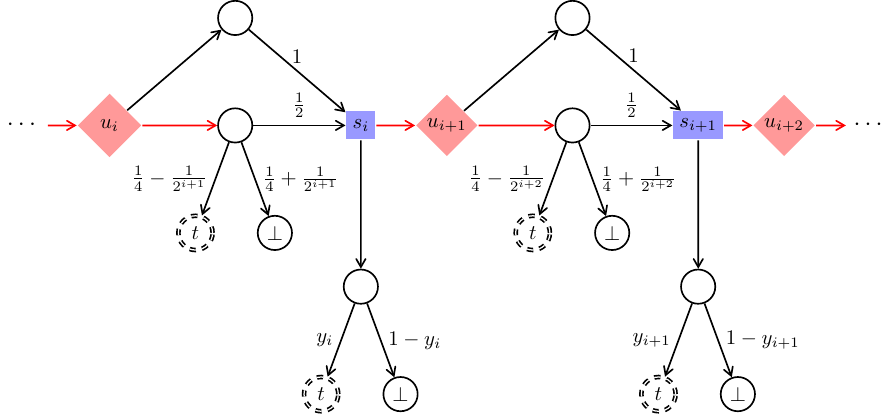}
\end{center}
\caption{Finitely branching turn-based reachability game $\game$, where optimal Maximizer strategies cannot be Markov. For clarity, we have drawn several copies of the target state $t$.
The number $y_i$ is defined to be $\frac12 - \frac1{2^{i+1}}$; see \cref{prop:turn-fb-optmax-lower}.}
\label{fig:turn-fb-optmax-lower}
\end{figure}

\begin{proposition}\label{prop:turn-fb-optmax-lower}
There exists a finitely branching turn-based reachability game $\game$ \new{with} initial state $u_1$ and objective $\reach{\{t\}}$,
as shown in \cref{fig:turn-fb-optmax-lower}, such that
\begin{enumerate}
\item\label{prop:turn-fb-optmax-lower-1}
\new{From every state in $\game$, Maximizer has an optimal strategy.} 
\item\label{prop:turn-fb-optmax-lower-2}
Every randomized Maximizer strategy from $u_1$ that uses only a step counter and no memory is not optimal.
\end{enumerate}
\end{proposition}
\begin{proof}
Consider first a version of~$\game$, say $\game'$, in which the only outgoing transition from the states $u_i$ is the horizontal one, shown in red in \cref{fig:turn-fb-optmax-lower}.
I.e., in $\game'$ Minimizer does not have any choice and thus $\game'$ can be
regarded as a maximizing MDP.
Let $\zstrat$ be the MD Maximizer strategy that chooses at all states $s_i$ the \emph{horizontal} outgoing transition, shown in red in \cref{fig:turn-fb-optmax-lower}.
Then we have
\[
\probm_{\game',u_i,\zstrat}(\reach{\{t\}}) \ = \ \sum_{j=0}^\infty \frac{1}{2^j} \cdot \left(\frac14 - \frac{1}{2^{i+j+1}}\right) \ = \ \frac12 - \frac13 \cdot \frac{1}{2^{i-1}}\,.
\]
In~$\game'$, strategy~$\sigma$ is optimal for Maximizer everywhere.
\new{
Indeed the only alternative is to take the \emph{vertical} outgoing
transition at some state $s_i$, which is suboptimal by the following.
Consider a strategy~$\sigma'$ that chooses at state~$s_i$ the \emph{vertical}
outgoing transition.
Then we have
\begin{equation} \label{eq:turn-fb-optmax-lower}
\begin{split}
\probm_{\game',s_i,\sigma'}(\reach{\{t\}}) & =  \frac12 - \frac{1}{2^{i+1}}\\
                                           & <  \frac12 - \frac13 \cdot \frac{1}{2^i} \\
                                           & =  \probm_{\game',u_{i+1},\sigma}(\reach{\{t\}})\\
                                           & =  \probm_{\game',s_{i},\sigma}(\reach{\{t\}})\\
                                           & \le  \valueof{\game'}{s_i} \,.
\end{split}
\end{equation}
}
Consider now the original game~$\game$ as shown in \cref{fig:turn-fb-optmax-lower}.
Since Minimizer has additional options, the value at each state is not larger than at the corresponding state in~$\game'$.

\new{
However, we show that, in $\game$, Maximizer still has an optimal strategy 
$\hat\sigma$ from every state $\state$.
It suffices to show this property for states $\state = u_k$ for any $k \ge 1$.
At states $\state = s_k$, the optimal move is always to go right to $u_{k+1}$,
because the vertical transition is suboptimal by
\eqref{eq:turn-fb-optmax-lower}, and at random states no decision can be made
until the next step (or ever).
}

\new{
We show that, starting from~$u_k$, strategy $\hat\sigma$ attains the same
value ($\frac12 - \frac13 \cdot \frac{1}{2^{k-1}}$) in $\game$ as in~$\game'$.
}
Namely, define $\hat\sigma$ so that as long as Minimizer chooses the
horizontal (red) outgoing transitions at~$u_i$, Maximizer chooses the horizontal (red) outgoing transition at~$s_i$; once Minimizer deviates and chooses the non-horizontal outgoing transition at, say, $u_i$, then Maximizer responds by choosing the vertical outgoing transition at~$s_i$.
(The strategy~$\hat\sigma$ is a deterministic public 1-bit strategy, but we do not need that here.)

Intuitively, for Minimizer a ``deviation'', i.e., choosing a non-horizontal outgoing transition, is value-increasing and thus suboptimal.
But she may try to lay a trap for Maximizer and trick him into visiting all states $u_i, s_i$.
To stop this from happening, Maximizer, using~$\hat\sigma$, responds to a Minimizer deviation by also deviating, i.e., by choosing a vertical outgoing transition.
Such a deviation is suboptimal for him, but the game is constructed so that a Maximizer deviation decreases the value less than Minimizer has previously increased it by her deviation.
In effect, with $\hat\sigma$, Maximizer attains as much as in~$\game'$ if Minimizer never deviates; if Minimizer deviates, Maximizer attains slightly more than in~$\game'$.
Thus, $\hat\sigma$ is optimal.

Formally, let $\pi$ be any Minimizer strategy.
Denote by $D_i$ the event that Minimizer deviates at~$u_i$ \new{(for some $i \ge k$)},
i.e., chooses the non-horizontal outgoing transition at~$u_i$.
Since the Maximizer strategy~$\hat\sigma$ responds by choosing the vertical outgoing transition at~$s_i$, we have
\[
 \probm_{\game,\new{u_k},\hat\sigma,\pi}(\reach{\{t\}} \mid D_i) \ = \ \frac12 - \frac{1}{2^{i+1}} \ > \ \frac12 - \frac{1}{3} \cdot \frac{1}{2^{i-1}} \ = \ \probm_{\game',u_i,\zstrat}(\reach{\{t\}}) \,,
\]
i.e., by deviating at~$u_i$,
Minimizer increases the probability of reaching~$t$ compared to her not deviating at~$u_i$ or thereafter (which corresponds to playing in~$\game'$).
We have already argued that $\probm_{\game',u_i,\zstrat}(\reach{\{t\}}) = \valueof{\game'}{u_i} \ge \valueof{\game}{u_i}$. It follows that $\hat\sigma$ is optimal,
which concludes the proof of \cref{prop:turn-fb-optmax-lower-1}.

Towards \cref{prop:turn-fb-optmax-lower-2}, note that in~$\game$
the step counter from $u_1$ is implicit in the current state.
In particular, starting from~$u_1$, if a state~$s_i$ is visited then it is visited as the $3i$-th state.
It follows that a step counter is not useful for Maximizer strategies.
Thus, it suffices to show that no memoryless strategy for Maximizer is optimal.
Let $\sigma$ be any memoryless Maximizer strategy.
If $\sigma$ chooses at every~$s_i$ the horizontal outgoing transition, the probability of reaching~$t$ is zero if Minimizer never chooses the horizontal outgoing transition at any~$u_i$; thus, $\sigma$ is not optimal.
Hence, we can assume that there is a state~$s_i$ at which $\sigma$ chooses with a positive probability the vertical outgoing transition.
Denote by $E_i$ the event that Maximizer chooses the vertical outgoing transition at~$s_i$.
Let $\pi$ be the Minimizer strategy that at all~$u_j$ chooses the horizontal outgoing transition.
Recall that $\pi$ is optimal for Minimizer everywhere.
Similarly to~\eqref{eq:turn-fb-optmax-lower} above, we have
\begin{equation*}
\probm_{\game,u_1,\sigma,\pi}(\reach{\{t\}} \mid E_i) \ = \ \frac12 - \frac{1}{2^{i+1}} \ < \ \frac12 - \frac13 \cdot \frac{1}{2^i} \ = \ \valueof{\game}{u_{i+1}} \ = \ \valueof{\game}{s_i}\,.
\end{equation*}
Thus, $\sigma$ is not optimal.
As $\sigma$ was chosen arbitrarily, Maximizer does not have an optimal memoryless strategy.
This proves \cref{prop:turn-fb-optmax-lower-2}.
\end{proof}

\new{
In the example in \cref{fig:turn-fb-optmax-lower},
subgame-perfect Maximizer strategies cannot guarantee any positive probability
of reaching the target state, because they would always choose the step
$\state_i \to u_{i+1}$ for all $i \in \N$.
Thus an optimal Maximizer strategy may need to take steps that are locally
sub-optimal in subgames.
}

\new{
However, in those turn-based reachability games with finite Minimizer action sets
where optimal subgame-perfect
Maximizer strategies do exist, there also exist such strategies that are
memoryless and deterministic by \cite[Theorem 5]{KieferMSW17a}.
}

\subsection{Concurrent Games}\label{subsec:optmax-concurrent}

\new{
The lower bounds for turn-based games from \Cref{subsec:optmax-turn}
immediately carry over to concurrent games.
It is an open question whether the upper bounds carry over.
We conjecture that a suitably adapted version of \Cref{thm:turn-fb-optmax-upper}
might hold for concurrent games (e.g., condition (A) might be generalized by
requiring that all probability distributions have finite support).
However, such a generalization faces several obstacles.
In concurrent games, it is more difficult to define what it means for
Minimizer to ``give a gift'', and how to define a restricted version of the game
where such gift-giving is forbidden.
Also one would need a suitably generalized version of \Cref{lem:no-val-inc}.
}

\new{
A special case of optimal Maximizer strategies are those that
win almost surely.
Here no memory is needed at all, and these strategies can even be made uniform.
The following upper bound for concurrent games trivially carries over to
turn-based games (with finite action sets).
}

\begin{restatable}{theorem}{concreachas}\label{thm:conc-reach-as}
Given a concurrent game with finite action sets and a
reachability objective,
there exists some randomized memoryless Maximizer strategy
that is almost surely winning from every state that admits an almost
surely winning strategy \verynew{(i.e., the same strategy works from
  all these states)}.
\end{restatable}
\begin{proof}
Let $\game$ be a concurrent game with state space $\states$,
and let $\reach{\reachset}$ be a reachability objective. 

Without restriction, we can assume that all states in $\states$ admit
an almost surely winning strategy.
Otherwise, we consider the subgame $\game'$ obtained by restricting
$\game$ to $\states'$, \new{ie the game on the subgraph induced by $S'$,}
where $\states' \subseteq \states$ is the subset of states that admit an almost
surely winning strategy in $\game$.
Then all states in $\states'$ admit an almost surely winning strategy in
$\game'$. (Note that this construction of $\game'$ would not work if we
replaced the ``almost surely winning'' condition by the weaker condition of ``having value $1$''.)

In order to construct a memoryless Maximizer
strategy~$\hat{\zstrat}$ that wins almost surely from every state,
we inductively define a sequence of modified games $\game_i$ in which the strategy
of Maximizer is already fixed on a finite subset of the state space,
and where all states in $\game_i$ still admit an
almost surely winning strategy.
Fix an enumeration $\state_1, \state_2, \ldots$ of $\states$
in which very state $\state$ appears \emph{infinitely often}.

For the base case we have $\game_0 \eqdef \game$ and the property holds by
our assumption on $\game$.

Given~$\game_i$, we construct~$\game_{i+1}$ as follows.
We use \cref{lem:conc-reach-non-uniform} to get a memoryless strategy
$\zstrat_i$ and a finite subset of states $R_i$
s.t.\ $\inf_\ostrat\probm_{\game_i,\state_i,\zstrat_i,\ostrat}(\reachn{R_i}{\reachset}) \ge
\valueof{\game_i}{s_i} - 2^{-i} = 1 - 2^{-i}$.

Let $\game_i'$ be the subgame of $\game_i$ that is restricted to $R_i$
and further let 
\[
    R_i' \eqdef \{\state \in R_i \mid \inf_\ostrat\probm_{\game_i',\state,\zstrat_i,\ostrat}(\reachn{R_i}{\reachset}) >0\}
\]
be the subset of states in $R_i$ where $\zstrat_i$ has strictly positive
attainment in $\game_i'$.
In particular, we have $\state_i \in R_i'$ for all $i \ge 1$.
Since $R_i'$ is finite,
we have
$$\lambda_i \eqdef \min_{\state \in R_i'}
\inf_\ostrat\probm_{\game_i',\state,\zstrat_i,\ostrat}(\reachn{R_i}{\reachset})> 0.$$

We now construct $\game_{i+1}$ by modifying $\game_i$ as follows.
For every state $\state \in R_i'$ we fix Maximizer's (randomized) action
according to $\zstrat_i$.
Then
$\inf_\ostrat\probm_{\game_{i+1},\state_i,\zstrat,\ostrat}(\reach{\reachset}) \ge 1 - 2^{-i}$
and
$\inf_\ostrat\probm_{\game_{i+1},\state,\zstrat,\ostrat}(\reachn{R_i'}{\reachset})
\ge \lambda_i$ for all
$\state \in R_i'$ and all $\zstrat \in \zstratset_{\game_{i+1}}$
(and thus in particular for the strategy $\hat{\zstrat}$ that we will construct).

Now we show that in $\game_{i+1}$ all states $\state$ 
still have an almost surely winning strategy.

Let $\zstrat$ be an a.s.\ winning Maximizer strategy from
$\state$ in $\game_i$, which exists by the induction hypothesis.
We now define an a.s.\ winning Maximizer strategy $\zstrat'$ from
$\state$ in $\game_{i+1}$.

If the game does not enter $R_i'$
then $\zstrat'$ plays exactly as $\zstrat$ (which is possible since outside
$R_i'$ no Maximizer actions have been fixed).
If the game enters $R_i'$ then it will reach the target 
within $R_i'$
\new{(i.e., before exiting $R_i'$, if ever)}
with probability $\ge \lambda_i >0$.
Plays that do not stay inside $R_i'$ then exit $R_i'$
at some state $\state' \notin R_i'$.
Then, from $\state'$, $\zstrat'$ plays
an a.s.\ winning strategy w.r.t.\ $\game_i$
(which exists by the induction hypothesis).

Now we show that $\zstrat'$ wins almost surely from $\state$ in $\game_{i+1}$.
\new{The plays from $\state$ can be partitioned into the following three subsets.
The first set of plays visit $R_i'$ only finitely often and
eventually forever follow an a.s.\ winning strategy outside of $R_i'$
and thus (except for a nullset) eventually reach the target.
The second set of plays enter $R_i'$ infinitely often
and the third set of plays eventually forever remain in $R_i'$.
For plays in both the second and third sets, the probability of reaching the
target from the current state does not converge to zero, since $\lambda_i >0$.
Hence, by L\'{e}vy's 0-1 law, the probability of reaching the target must
converge to $1$, and thus (except for a nullset) the plays in the second and
third set also reach the target.
Therefore $\zstrat'$ almost surely wins from $\state$ in $\game_{i+1}$.
}

Finally, we can construct the memoryless Maximizer strategy $\hat{\zstrat}$.
\new{
Since our enumeration of the states $\state_1, \state_2, \dots$ contains
every state $\state \in \states$ infinitely often,
in particular it contains every state in $\states$.
Moreover, $\state_i \in R_i'$ for every $i \ge 1$.}
Thus, in the limit of the games
$\game_\infty$, all Maximizer choices are fixed.
The memoryless Maximizer strategy $\hat{\zstrat}$ plays according to these
fixed choices, i.e., it plays like $\zstrat_i$ at state $\state_i$ for all
$i \in \N$.
\new{Note that if $\state_i \in R_i'$ then, for all $j > i$,
the mixed action of $\zstrat_j$ at $\state_i$ coincides with the mixed action
of $\zstrat_i$ at $\state_i$, because $\zstrat_j$ is defined in a game where
Maximizer's mixed action in $\state_i$ is already fixed.}

\new{Since $\hat{\zstrat}$ plays like $\zstrat_i$ inside $R_i'$, we obtain
$\inf_\ostrat\probm_{\game,\state_i,\hat{\zstrat},\ostrat}(\reach{\reachset})
\ge
\inf_\ostrat\probm_{\game_i,\state_i,\zstrat_i,\ostrat}(\reachn{R_i}{\reachset})
\ge
1 - 2^{-i}$ for all $i \in \N$.}
Let $s \in \states$.
Since \new{our enumeration of the states contains
every state infinitely often},
$s=s_i$ holds for infinitely many~$i$, and thus we obtain
$\inf_\ostrat\probm_{\game,\state,\hat{\zstrat},\ostrat}(\reach{\reachset})
= 1$ as required. 
\end{proof}

\section{Minimizer Strategies}\label{sec:minimizer}
In the previous sections we have considered the strategy complexity
of Maximizer's strategies.
\new{
In this section we complete the picture of the strategy complexity of
Minimizer.
In reachability games, Minimizer strategies are generally simpler than
Maximizer strategies, because they do not need to make progress towards the target.
By~\cite[Thm.~1]{Raghavan-Nowak:1991}, we already know that 
Minimizer always has optimal (and thus $\eps$-optimal) MR strategies in 
concurrent reachability games with finite action sets. 
In~\cite[Thm.~3.1]{BBKO:IC2011},  this result is strengthened in the context of 
 finitely branching \verynew{turn-based} games, where it  is shown 
 that Minimizer always has  MD such strategies.
In the sequel, as depicted in~\cref{mintable},
we close the remaining gaps in the theory by studying the
strategy complexity of Minimizer in infinitely branching \verynew{turn-based} reachability games.
We prove that $\eps$-optimal
Minimizer strategies in 
infinitely branching \verynew{turn-based} reachability games can be chosen as 
deterministic and Markov (\Cref{thm:min-eps-optimal}).
In contrast, \emph{optimal} Minimizer strategies need not always exist in
infinitely branching \verynew{turn-based} reachability games.
However, even if optimal Minimizer strategies do exist,
a step counter plus finite private memory is \emph{not} sufficient in general
(\Cref{prop:infbranch-optmin}).  
}

We begin by considering games on acyclic graphs.
Memoryless strategies in acyclic games yield Markov strategies in general
games, since an encoded step counter makes the graph acyclic.
In fact, the following result about acyclic games is slightly more general,
since not all acyclic graphs yield an implicit step counter, i.e., the same
state might be reached via paths of different lengths.

\begin{restatable}{lemma}{thmEpsOptSafety}
\label{thm:acyclic-min-eps-optimal}
\new{
For every acyclic turn-based reachability game 
$\game=\gametuple$, reachability target
$\reachset \subseteq \states$
and every $0<\eps<1$
there exists an MD Minimizer strategy $\ostrat$ which
satisfies, for every state $\state_0\in\states$
and every Maximizer strategy $\zstrat$, that
$
\probm_{\game,\state_0,\zstrat,\ostrat}(\reach{\reachset})
\le \valueof{\game,\reach{\reachset}}{\state_0} (1+ \eps)
$.
Hence, acyclic turn-based reachability games admit uniformly $\eps$-optimal MD strategies for Minimizer.
}
\end{restatable}
\begin{proof}
Let us shortly write $\valueof{}{s}=\valueof{\game,\reach{\reachset}}{s}$ for the value of a state $\state$
and let $\iota:\states\to\N\setminus\{0\}$ be an enumeration of the state space
starting at 1. Define $\ostrat$ as the MD \new{Minimizer} strategy that, at any state $\state\in\ostates$, picks a
successor $\state'$ such that
$$\valueof{}{\state'}\quad\le\quad \valueof{}{\state} (1+ \ln(1+\eps) 2^{-\iota(\state)}).$$

To show that this strategy $\ostrat$ satisfies the claim
we (over)estimate the error by
$\ferr{s} \eqdef\prod_{s'\in\poststar{s}}(1+\ln(1+\eps)2^{-\iota(s')})$
where $\poststar{s}\subseteq \states$ is the set of states reachable from state $s\in\states$ (under any pair of strategies).
Notice that this guarantees that \begin{align}
    \nonumber
    1 < \ferr{s} &\le \prod_{i>0}\left(1+\ln(1+\eps)2^{-i}\right)\\
                 &\le \exp\left(\sum_{i>0}\ln(1+\eps)2^{-i}\right) \label{eq:error-ferr}\\
                 &\le \exp(\ln(1+\eps)) = 1+\eps\nonumber
\end{align}
where the third inequality uses that $1+x \le \exp(x)$.

Let $\zstrat$ be an arbitrary Maximizer strategy. 
For this pair $\zstrat,\ostrat$ of strategies
let's consider plays $(\NthState{i})_{i\ge 0}$ that start in $\state_0\in\states$
and proceed according to $\zstrat,\ostrat$, and let $\expectation[\game,\state_0,\zstrat,\ostrat]$
be the expectation with respect to $\probm_{\game,\state_0,\zstrat,\ostrat}$.

An induction on $n$ using our choice of strategy gives, for every initial state $\state_0\in S$, that
\begin{equation}
    \label{eq:exi-1}
    \expectation[\game,\state_0,\zstrat,\ostrat](\valueof{}{\NthState{n}})
    \le \valueof{}{s_0}\ferr{s_0}.
\end{equation}
Indeed, this trivially holds for $n=0$ as
$\expectation[\game,\state_0,\zstrat,\ostrat](\valueof{}{\NthState{0}})=\valueof{}{\state_0}$
and $\ferr{s_0} > 1$.
For the induction step there are three cases.

\smallskip
\emph{Case 1: $s_0\in \ostates$ and $\ostrat(\state_0)=s$.} 
Let $\zstrat{[s_0\to s]}$ denote the Maximizer strategy from $s$ that behaves just like
$\zstrat$ does after observing the first step, i.e., satisfies $\zstrat{[s_0\to s]}(sh) = \zstrat(s_0sh)$ for all suffix histories $h\in\states^*$.
Then

\begin{align*}
    \expectation[\game,\state_0,\zstrat,\ostrat](\valueof{}{\NthState{n+1}})
    &=
    \expectation[\game,\state,\zstrat{[s_0\to s]},\ostrat](\valueof{}{\NthState{n}})\\
&\le
    \valueof{}{s}\ferr{s}
    &\text{ind. hyp.}\\
    &\le
    \valueof{}{\state_0}
    \left(1+\ln(1+\eps)2^{-\iota(\state_0)}\right)
    \ferr{s}
    &\text{def. of $\ostrat$}\\
    &\le
    \valueof{}{s_0}\ferr{s_0}
    &\text{acyclicity; def.~of } \ferr{s_0}.
\end{align*}

\smallskip
\emph{Case 2: $s_0\in \zstates$.} 
Again, for any state $s$ let $\zstrat{[s_0\to s]}$ denote the suffix strategy consistent with $\zstrat$
after the first step.
Then
\begin{align*}
    \expectation[\game,\state_0,\zstrat,\ostrat](\valueof{}{\NthState{n+1}})
    &=
    \sum_{s \in \states}\zstrat(s_0)(s) \cdot 
    \expectation[\game,\state,\zstrat{[s_0\to s]},\ostrat](\valueof{}{\NthState{n}})
    \\&\le
    \sum_{s \in \states}\zstrat(s_0)(s) \cdot \valueof{}{\state}\ferr{s}
    \\
&\le
    \sum_{s \in \states}\zstrat(s_0)(s) \cdot \valueof{}{\state}
    \left(1+\ln(1+\eps)2^{-\iota(\state_0)}\right)
    \ferr{s}
    \\
&\le
    \sum_{s\in\states}\zstrat(s_0)(s) \cdot \valueof{}{s}
    \ferr{\state_0}
    \\
&\le
    \valueof{}{s_0}\ferr{s_0},
\end{align*}

where the first inequality holds by induction hypothesis, the second holds because
$1<{(1+\ln(1+\eps)2^{-\iota(\state_0)})}$,
and the third is acyclicity and the definition of 
$\ferr{\state_0}$.

\smallskip
\emph{Case 3: $s_0\in \rstates$} is analogous to case 2,
with the only difference that the initial successor distribution is $\probp(s_0)$,
the one fixed by the game, instead of $\sigma(s_0)$ and the last inequality becomes an equality.

\medskip
\noindent
Together with the observation (\cref{eq:error-ferr}) that $\ferr{s_0}\le (1+\eps)$ for every $\state_0$,
we derive that
    \begin{equation}
        \label{eq:exi-lim}
        \liminf_{n\to\infty} \expectation[\game,\state_0,\zstrat, \ostrat](\valueof{}{\NthState{n}})
        \le
        \valueof{}{\state_0}(1+\eps).
    \end{equation}

Finally, to show the claim, 
let $\RVInT{n} : \states^\omega \to \{0,1\}$ be the random variable that indicates that the $n$th state is in $\reachset$.
Note that $\RVInT{n} \le \valueof{}{\NthState{n}}$ because target states have value $1$.
\new{
Recall that $\reachn{n}{\reachset}$ denotes the objective of visiting~$\reachset$
within at most $n$ rounds of the game.
}
We conclude that
\begin{align*}
    \probm_{\game,\state_0,\zstrat, \ostrat}(\reach{\reachset})
    &=\quad\probm_{\game,\state_0,\zstrat, \ostrat}\left(\bigcup_{i=0}^\infty{\reachn{i}{T}}\right)
    && \text{semantics of~$\reach{\reachset}$} \\
& =\quad\smashoperator{\lim_{n\to\infty}} \probm_{\game,\state_0,\zstrat, \ostrat}\left(\bigcup_{i=0}^n\reachn{i}{T} \right)
    && \text{continuity of measures} \\ & =\quad \smashoperator{\lim_{n\to\infty}} \probm_{\game,\state_0,\zstrat, \ostrat}\left(\reachn{n}{T}\right)
 && \text{$\reachset$ is a sink} \\
 & =\quad\smashoperator{\lim_{n\to\infty}} \expectation[\game,\state_0,\zstrat, \ostrat](\RVInT{n})
 && \text{definition of $\RVInT{n}$} \\
 & \le\quad \liminf_{n\to\infty} \expectation[\game,\state_0,\zstrat, \ostrat](\valueof{}{\NthState{n}})
 && \text{as $\RVInT{n} \le \valueof{}{\NthState{n}}$}\\
 & \le\quad \valueof{}{\state_0}(1+\eps)
 && \text{by \cref{eq:exi-lim}.}
 \
\end{align*}
\end{proof}

\begin{theorem}
\label{thm:min-eps-optimal}
Turn-based games, even infinitely branching ones, with reachability objective admit uniformly $\eps$-optimal
strategies for Minimizer that are deterministic and Markov.
\end{theorem}
\begin{proof}
For a given game $\game=\gametuple$ and reachability target $\reachset \subseteq \states$,
one can construct the acyclic game that encodes the stage (clock value) into the states:
$\mathcal{G}'=\tuple{\states,(\zstates',\ostates',\rstates'),\transition',\probp'}$
where $\states'=\states\x\N$,
$\zstates'=\zstates\x\N$,
$\ostates'=\ostates\x\N$
$\rstates'=\rstates\x\N$, 
and for all $i\in\N$,
$(s,i)\transition'(t,i+1) \iff s\transition t$ and
$\probp((s,i))((t,i+1)) = \probp(s)(t)$.

Every Markov strategy in $\?G$ uniquely gives rise to a memoryless strategy in $\?G'$ and vice versa.
The claim now follows from \cref{thm:acyclic-min-eps-optimal}.
\end{proof}

In infinitely branching turn-based reachability games, optimal Minimizer
strategies need not exist \cite[]{KMSW2017}.
When they do exist, they may need infinite memory.
We slightly improve the result of \cite[Proposition 5.6.(a)]{kucera_2011}
by showing that even a step counter does not help Minimizer.

\begin{figure}
\begin{center}
\includegraphics[width=0.75\textwidth, angle=0]{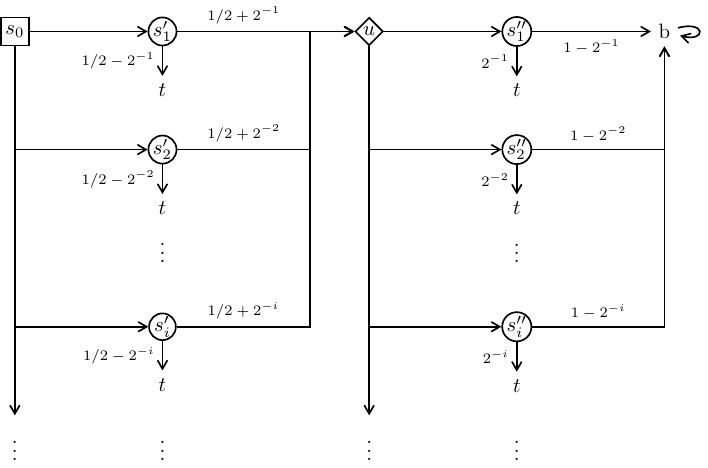}
\end{center}
\caption{The game $\game$ from \cref{def:infbranch-optmin}.}
\label{fig:infbranch-optmin}
\end{figure}

\begin{definition}\label{def:infbranch-optmin}
We define an infinitely branching turn-based reachability game $\game$ with
initial state $\state_0$ and target state $t$.
\new{See \cref{fig:infbranch-optmin} for a depiction.}
Let $\state_0$ be Maximizer-controlled. We have transitions
$\state_0 \to \state_i'$ for all $i \ge 1$.
All states $\state_i'$ are random states with 
$P(\state_i')(t) = 1/2 - 2^{-i}$ and 
$P(\state_i')(u) = 1/2 + 2^{-i}$.
The state $u$ is Minimizer-controlled with transitions
$u \to \state_i''$ for all $i \ge 1$.
All states $\state_i''$ are random states with 
$P(\state_i'')(t) = 2^{-i}$ 
and $P(\state_i'')(b) = 1-2^{-i}$
for a losing sink state $b$.
\end{definition}

\begin{proposition}\label{prop:infbranch-optmin}
There exists an infinitely branching turn-based reachability game $\game$
with initial state $\state_0$ and objective $\reach{\{t\}}$
as in \Cref{def:infbranch-optmin}, such that
\begin{enumerate}
\item\label{prop:infbranch-optmin-1}
Minimizer has an optimal strategy from $\state_0$.
\item\label{prop:infbranch-optmin-2}
Every randomized Minimizer strategy from $\state_0$ that uses only a step
counter and finite private memory is not optimal.
\end{enumerate}
\end{proposition}
\begin{proof}
Towards \Cref{prop:infbranch-optmin-1}, we note that $\valueof{\game}{u}=0$
and thus $\valueof{\game}{\state_0}=1/2$.
Minimizer's optimal strategy $\ostrat$ from $\state_0$ is defined as follows.
In plays where the state $u$ is not reached, Minimizer does not make any
decisions.
If state $u$ is reached, Minimizer considers the history of this play:
If Maximizer made the step $\state_0 \to \state_i'$ for some $i\ge 1$, then
Minimizer plays $u \to \state_i''$ for the same $i$.
Now we show that $\ostrat$ is optimal for Minimizer from $\state_0$.
Let $\zstrat$ be an arbitrary Maximizer strategy from $\state_0$
and let $x_i$ be the probability that $\zstrat$ chooses the step 
$\state_0 \to \state_i'$. This must be a distribution,
i.e., $\sum_{i\ge 1} x_i = 1$.
Then we have
\begin{equation}
\begin{aligned}
\probm_{\game,\state_0,\zstrat,\ostrat}(\reach{\{t\}})
&=
\sum_{i\ge 1} x_i((1/2 - 2^{-i}) + (1/2 + 2^{-i})2^{-i})\\
&\le
\sum_{i\ge 1} x_i(1/2)
=
1/2 = \valueof{\game}{\state_0}
\end{aligned}
\end{equation}
as required.

Towards \Cref{prop:infbranch-optmin-2}, we note that the step counter from
$\state_0$ is implicit in the states of $\game$, and thus superfluous for
Minimizer strategies. Hence it suffices to prove the property for Minimizer
strategies with finite memory.
Let $\ostrat$ be an FR Minimizer strategy with finitely many memory modes $\{1,\dots,k\}$.
In state $u$ this strategy $\ostrat$ can base its decision only on the current
memory mode $\memconf \in \{1,\dots,k\}$.
Let $X(\memconf) \eqdef \probm_{\game,u,\zstrat,\ostrat[\memconf]}(\reach{\{t\}})$
be the probability of reaching the target if $\ostrat$ is in mode $\memconf$
at state $u$. (From state $u$ only Minimizer plays, thus Maximizer has no influence.)
Since $X(\memconf) >0$ and the memory is finite, we have
$Y \eqdef \min_{\memconf \in \{1,\dots,k\}} X(\memconf) > 0$.
There exists a number $i$ sufficiently large such that $2^{-i} < Y/2$.
Let $\zstrat$ be a Maximizer strategy from $\state_0$ that chooses the
transition $\state_0 \to \state_i'$.
Then we have
\[
\probm_{\game,\state_0,\zstrat,\ostrat}(\reach{\{t\}})
\ge
(1/2 - 2^{-i}) + (1/2 + 2^{-i})Y
>
1/2 = \valueof{\game}{\state_0}
\]
and thus $\ostrat$ is not optimal.
\end{proof}
 
\section{Conclusion and Outlook}\label{sec:conclusion}
\new{
Our results closed many gaps about the strategy complexity of reachability
games; cf.~\Cref{maxtable} and \Cref{mintable}.
To summarize our main contributions, we return to the open questions raised
in \Cref{sec:intro}, which are now answered.
}

\new{
\begin{description}
\item[Q1.]
The negative result of \cite{Raghavan-Nowak:1991} can be strengthened.
There are no \emph{uniformly} $\eps$-optimal memoryless Maximizer strategies
in countably infinite reachability games, not even if the game is turn-based
and finitely branching; cf.~\Cref{thm:TB-BMI-zplus}.
This highlights the difference between (turn-based) 2-player stochastic
games and MDPs. In the latter, there do exist uniformly $\eps$\nobreakdash-optimal memoryless
strategies for reachability \cite{Ornstein:AMS1969}.
\item[Q2.]
In concurrent reachability games with finite action sets,
\emph{uniformly} $\eps$-optimal Maximizer strategies
exist and they require only $1$ bit of public memory.
In turn-based games, these strategies can even be chosen as deterministic.
See \Cref{thm:conc-reach-uniform}.
\item[Q3.]
If Minimizer is allowed infinite action sets then
reachability games are much more difficult for Maximizer.
Even in turn-based reachability games with infinitely branching Minimizer
states, Maximizer strategies based on a step counter plus arbitrary finite
private memory are
\verynew{insufficient. In general, they cannot guarantee any positive
  attainment against all Minimizer strategies, even if the start state has
  value $1$. In fact, 
the counterexample in \Cref{thm:no-sc-plus-finite} satisfies the even
stronger property that all states in it admit an almost surely winning Maximizer strategy.}
\end{description}
}

\new{
Open questions for further work concern the strategy complexity of optimal
Maximizer strategies, where they exist.
In general, a step counter plus finite private memory is not sufficient for
optimal Maximizer strategies, even in turn-based reachability games, by
\Cref{prop:conc-optmax}.
However, under certain mild conditions, a step counter plus $1$ bit of public
memory suffices for optimal Maximizer strategies in turn-based reachability
games, by \Cref{thm:turn-fb-optmax-upper}.
A similar theorem might hold for concurrent reachability games with
finite action sets under suitably adapted conditions.
}

\backmatter

\newpage

\begin{appendices}

\section{Technical Lemmas}

\label{app:technical}
The following inequality is due to Weierstrass.

\begin{proposition}[{\cite{Bromwich:1955} p.~104--105}]\label{prop:Weierstrass}
    Given an infinite sequence of real numbers $a_n$ with $0 \le a_n \le 1$,
\new{the following holds for all $n\in\N$.}
\[
\prod_{k=1}^n (1-a_k) \le \frac{1}{1+ \sum_{k=1}^n a_k}
\]
\end{proposition}
\begin{proof}
By induction on $n$.
In case $n=1$ we have $(1-a_1)(1+a_1)={(1-a_1^2)}\le 1$ as
required.
For the induction hypothesis we assume
\[
\prod_{k=1}^n (1-a_k)\left(1+\sum_{k=1}^n a_k\right) \le 1
\]
For the induction step we have
\begin{align*}
  \prod_{k=1}^{n+1} (1-a_k)\left(1+\sum_{k=1}^{n+1} a_k\right)
  &= (1-a_{n+1})\prod_{k=1}^n (1-a_k)\left(\left(1+\sum_{k=1}^n a_k\right)+a_{n+1}\right)\\
  &\le (1-a_{n+1})\left(1 + a_{n+1}\prod_{k=1}^n (1-a_k)\right)\\
  &\le (1-a_{n+1})(1 + a_{n+1})\\
  &= (1-a_{n+1}^2) \le 1
  \
\end{align*}
\end{proof}

\begin{proposition}\label{prop:product-sum}
Given an infinite sequence of real numbers $a_n$ with $0 \le a_n < 1$, we
have
\[
\prod_{n=1}^\infty (1-a_n) > 0 \quad\Leftrightarrow\quad \sum_{n=1}^\infty a_n
< \infty.
\]
and the ``$\Rightarrow$'' implication holds even for the weaker assumption $0 \le a_n \le 1$.
\end{proposition}
\begin{proof}
If $a_n=1$ for any $n$ then the ``$\Rightarrow$'' implication is vacuously
true, but the ``$\Leftarrow$'' implication does not hold in general.
In the following we assume $0 \le a_n < 1$.

In the case where $a_n$ does not converge to zero, the property is trivial.
In the case where $a_n \rightarrow 0$, it is shown
by taking the logarithm of the product and using the limit comparison test as follows.

Taking the logarithm of the product gives the series
\[
\sum_{n=1}^{\infty} \ln(1 - a_n)
\]
whose convergence (to a finite number $\le 0$) is equivalent to the positivity of the product.
It is also equivalent to the convergence (to a number $\ge 0$) of its negation
$\sum_{n=1}^{\infty} -\ln(1 - a_n)$.
But observe that (by L'H\^{o}pital's rule)
\[
\lim_{x \rightarrow 0} \frac{-\ln(1-x)}{x} = 1.
\]
Since $a_n \rightarrow 0$ we have
\[
\lim_{n \rightarrow \infty} \frac{-\ln(1-a_n)}{a_n} = 1.
\]
By the limit comparison test, the series
$\sum_{n=1}^{\infty} -\ln(1 - a_n)$ converges if and only if the series
$\sum_{n=1}^\infty a_n$ converges.
\end{proof}

\begin{proposition}\label{prop:tail-product}
Given an infinite sequence of real numbers $a_n$ with $0 \le a_n \le 1$,
\[
\prod_{n=1}^\infty a_n > 0 \quad\Rightarrow\quad \forall\eps>0\,\exists N.\, \prod_{n=N}^\infty a_n \ge (1-\eps).
\]
\end{proposition}
\begin{proof}
  If there is $n$ with $a_n=0$ \new{or if $\eps \ge 1$} then the property is vacuously true.
  In the following we assume $a_n >0$ \new{and $\eps < 1$}.
  Since $\prod_{n=1}^\infty a_n > 0$, by taking the logarithm we obtain
  $\sum_{n=1}^\infty \ln(a_n) > -\infty$.
  Thus for every $\delta>0$ there exists an $N$
  s.t.\ $\sum_{n=N}^\infty \ln(a_n) \ge -\delta$.
  By exponentiation we obtain
  $\prod_{n=N}^\infty a_n \ge \exp(-\delta)$.
  By picking $\delta = -\ln(1-\eps)$ the result follows.
\end{proof}

\end{appendices}

\bibliography{journals,conferences,autocleaned}

\end{document}